\documentclass
  [
    reqno 
  ]
  {amsart}
\usepackage[margin=3.18cm, top=2.54cm, bottom=2.54cm]{geometry}

\allowdisplaybreaks

\usepackage{caption}
\usepackage[leftcaption]{sidecap}

\usepackage{enumitem}
\setlist[enumerate,1]{label=\textbf{\textup{(\roman*)}}}
\setlist{
  topsep=0pt,
  itemsep=20pt,
  leftmargin=1cm,
}
\setlist{before=\leavevmode}

\usepackage[
  backend=biber,
  style=alphabetic,
  giveninits=true, 
  backref=true
]{biblatex}
\usepackage{lmodern}
\usepackage{anyfontsize}
    
\usepackage{multirow}
\usepackage{hhline}

\usepackage{tabularray}
\UseTblrLibrary{booktabs}

\usepackage{adjustbox}
\usepackage{tcolorbox}

\newtcolorbox{standout}{
  colback=gray!15,
  boxrule=0pt,
  left=.3cm,
  right=.3cm,
  top=.18cm,
  bottom=.18cm,
  boxsep=0pt
}

\usepackage{mathtools}
\usepackage{amssymb}
\usepackage{amsthm}
\usepackage{xfrac}
\usepackage{stmaryrd} 
\usepackage{scalerel} 
\usepackage{stackengine} 
\usepackage[safe]{tipa} 
\usepackage[
  cal=euler,
  scr=dutchcal
]
 {mathalpha} 

 \newcommand{\bracket}[3]{%
  \stretchleftright
    {#1}
    {%
      \ensurestackMath{\addstackgap[1pt]{#2}}%
      \vrule width 0pt depth 2pt height 0pt
    }
    {#3}%
} 
\newcommand{\scaledbracket}[3]{%
  \ThisStyle{%
    \stretchleftright
      {#1}
      {
        \ensurestackMath{\addstackgap[1pt]{\SavedStyle #2}}%
        \vrule width 0pt depth 1.5pt height 0pt
      }
      {#3}%
  }%
}
\newcommand{\bracketmid}[4]{%
  \stretchleftright{#1}{%
    \ensurestackMath{%
      \addstackgap[2pt]{#2}%
      \,\stretchrel*{|}{\addstackgap[2pt]{#2#3}}\,%
      \addstackgap[2pt]{#3}%
    }%
  }{#4}%
}

\theoremstyle{plain}
\newtheorem{theorem}{Theorem}[section]
\newtheorem{lemma}[theorem]{Lemma}
\newtheorem{proposition}[theorem]{Proposition}

\theoremstyle{definition}
\newtheorem{definition}[theorem]{Definition}
\newtheorem{example}[theorem]{Example}

\theoremstyle{remark}
\newtheorem{remark}[theorem]{Remark}
\usepackage[
  colorlinks=true,
  linkcolor=darkgreen,
  citecolor=darkgreen,
  urlcolor=darkgreen
]{hyperref}
\usepackage{xurl} 

\usepackage{cleveref}
\crefname{equation}{}{}
\crefname{section}{\S}{\S\S}
\crefname{subsection}{\S}{\S\S}
\crefname{subsubsection}{\S}{\S\S}
\crefname{definition}{Def.}{Defs.}
\crefname{theorem}{Thm.}{Thms.}
\crefname{corollary}{Cor.}{Cors.}
\crefname{lemma}{Lem.}{Lems.}
\crefname{proposition}{Prop.}{Props.}
\crefname{remark}{Rem.}{Rems.}
\crefname{notation}{Ntn.}{Ntns.}
\crefname{fact}{Fact}{Fact}
\crefname{example}{Ex.}{Exs.}
\crefname{figure}{Fig.}{Figs.}
\crefname{table}{Tab.}{Tabs.}
\crefname{footnote}{ftn.}{ftns.}
\Crefname{footnote}{Ftn.}{Ftns.}

\usepackage{tikz}
\usetikzlibrary{
  cd,
  calc,
  arrows.meta,
  backgrounds,
  decorations,
  decorations.pathmorphing, 
}

\tikzcdset{
  arrow style=tikz, 
  diagrams={
    trim left=
    {([xshift=5pt]current bounding box.west)},
    trim right=
    {([xshift=-5pt]current bounding box.east)},
    baseline=
    {([yshift=-2pt]current bounding box.center)},
    >={
      Computer Modern Rightarrow[
        length=4pt, width=4pt
      ]
    },
    hook/.style={
      {Hooks[right, length=2pt, width=8pt]}->
    },
    hook'/.style={
      {Hooks[left, length=2pt, width=8pt]}->
    },
    shorten=0pt,
  }
}

\definecolor{darkblue}{rgb}{0.05,0.25,0.65}
\definecolor{darkgreen}{RGB}{20,140,10}
\definecolor{lightgray}{rgb}{0.9,0.9,0.9}
\definecolor{darkorange}{RGB}{200,100,5}
\definecolor{darkyellow}{rgb}{.91,.91,0}
\definecolor{lightolive}{RGB}{225, 220, 185}

\let\originalsslash\sslash
\renewcommand{\sslash}{\mathord{\originalsslash}}

\makeatletter
\newcommand{\cpt}{\mathpalette\cpt@inner\relax}
\newcommand{\cpt@inner}[2]{%
  \scalebox{0.5}[0.9]{$#1\cup$}
  #1\{\infty\}
}
\makeatother

\makeatletter
\newcommand{\plus}{\mathpalette\sqcpt@inner\relax}
\newcommand{\sqcpt@inner}[2]{%
  \scalebox{0.5}[0.9]{$#1\sqcup$}
  #1\{\infty\}
}
\makeatother

\DeclareRobustCommand{\rchi}{{\mathpalette\irchi\relax}}
\newcommand{\irchi}[2]{\raisebox{\depth}{$#1\chi$}} 

\tikzset{
  snake left/.style={
    rounded corners,
    to path={
      let \p1 = (\tikztostart.east),
          \p2 = (\tikztotarget.west),
          \p3 = ($(\p1)!0.5!(\p2)$),
          \n1 = {8pt} 
      in
      (\p1)
      -- (\x1 + \n1, \y1)
      -- (\x1 + \n1, \y3)
      -- (\x2 - \n1, \y3) \tikztonodes
      -- (\x2 - \n1, \y2)
      -- (\p2)
    }
  }
}

\tikzset{
  uphordown/.style={
    rounded corners,
    to path={
      let \p1 = (\tikztostart.north),
          \p2 = (\tikztotarget.north),
          \n1 = {max(\y1,\y2) + 8pt}
      in
      (\p1)
      -- (\x1, \n1)
      -- (\x2, \n1) \tikztonodes 
      -- (\p2)
    }
  }
}

\tikzset{
  downhorup/.style={
    rounded corners,
    to path={
      let \p1 = (\tikztostart.south),
          \p2 = (\tikztotarget.south),
          \n1 = {min(\y1,\y2) - 8pt}
      in
      (\p1)
      -- (\x1, \n1)
      -- (\x2, \n1) \tikztonodes 
      -- (\p2)
    }
  }
}

\tikzset{
  rightvertleft/.style={
    rounded corners,
    to path={
      let \p1 = (\tikztostart.east),
          \p2 = (\tikztotarget.east),
          \n1 = {max(\x1,\x2) + 8pt}
      in
      (\p1)
      -- (\n1, \y1)
      -- (\n1, \y2) \tikztonodes 
      -- (\p2)
    }
  }
}

\tikzset{
  leftvertright/.style={
    rounded corners,
    to path={
      let \p1 = (\tikztostart.west),
          \p2 = (\tikztotarget.west),
          \n1 = {min(\x1,\x2) - 8pt}
      in
      (\p1)
      -- (\n1, \y1)
      -- (\n1, \y2) \tikztonodes 
      -- (\p2)
    }
  }
}

\newcommand{\inlinetikzcd}[1]{\begin{tikzcd}[sep=small, ampersand replacement=\&]#1\end{tikzcd}}

\renewcommand{\setminus}{-}

\newcommand{\defneq}{\equiv}

\newcommand{\HilbertSpace}{%
  \mathscr{H}%
}

\newcommand{\BoundaryDomain}{N}
\newcommand{\BulkDomain}{\Sigma}
\newcommand{\DeepBulkDomain}{\overline{\Sigma}}

\newcommand{\BulkInclusion}{i_{\mathrm{blk}}}
\newcommand{\BoundaryInclusion}{i_{\mathrm{bdr}}}
\newcommand{\BulkRestriction}{r_{\mathrm{blk}}}
\newcommand{\BoundaryRestriction}{r_{\mathrm{bdr}}}

\newcommand
  {\RefinedBulkInclusion}
  {i_{\mathrm{blk}'}}

\newcommand
  {\RefinedBoundaryRestriction}
  {r_{\mathrm{bdr}'}}

\newcommand{\Unwinding}{\partial^{\mathrm{uw}}}
\newcommand{\Rewinding}{\partial^{\mathrm{rw}}}

\newcommand
  {\RefinedUnwinding}
  {\partial^{\mathrm{uw}'}}

\newcommand{\UnpointedTop}{\mathrm{Top}}
\newcommand{\Top}{\UnpointedTop^{\!\ast}}
\newcommand{\Map}{\mathrm{Map}^{\!\ast}}
\newcommand{\UnpointedMap}{\mathrm{Map}}

\newcommand{\ClassifyingA}{\mathcal{A}}
\newcommand{\ClassifyingB}{\mathcal{B}}
\newcommand{\ClassifyingF}{\mathcal{F}}

\newcommand
  {\PhaseSpace}
  {\mathrm{PhsSp}}

\newcommand{\CurrentDensity}{\lambda}

\newcommand{\shape}{%
  \mathord{\scalerel*{\raisebox{0.1ex}{\textesh}}{f}}%
  \mkern 2mu
}

\newcommand{\FQHPhaseSpace}{\mathrm{fqhPhsSp}}

\newcommand{\RRationalization}{L^{\!\mathbb{R}}}

\begin{document}

\setlength{\abovedisplayskip}{3.5pt}
\setlength{\belowdisplayskip}{3.5pt}
\setlength{\abovedisplayshortskip}{-5pt}
\setlength{\belowdisplayshortskip}{3pt}

\title{Bulk-Edge Correspondence via Higher Gauge Theory}

\thanks{\emph{Funding} by Tamkeen UAE under the 
NYU Abu Dhabi Research Institute grant {\tt CG008}.}

\author{Hisham Sati}
\address{Mathematics Program and Center for Quantum and Topological Systems, New York University Abu Dhabi}
\curraddr{}
\email{hsati@nyu.edu}
\thanks{}

\author{Urs Schreiber           }
\address{Mathematics Program and Center for Quantum and Topological Systems, New York University Abu Dhabi}
\curraddr{}
\email{us13@nyu.edu}
\thanks{}

\subjclass[2020]{%
  Primary:
  81V70, 
  55Q55, 
  81T70; 
  Secondary:
  83E50, 
  81T30, 
  18N60 
}

\keywords{%
  bulk-boundary correspondence,
  higher gauge theory,
  topological order,
  fractional quantum Hall systems,
  edge currents,
  nonabelian cohomology, 
  Hopf fibration,
  homotopy theory,
  cohomotopy,
  geometric homotopy theory,
  $\infty$-topos theory,
  geometric engineering,
  supergravity,
  M-branes,
  Chern-Simons theory,
  Floreanini-Jackiw theory%
}

\date{\today}

\dedicatory{
  \href{https://ncatlab.org/nlab/show/Center+for+Quantum+and+Topological+Systems}{\includegraphics[width=3.1cm]{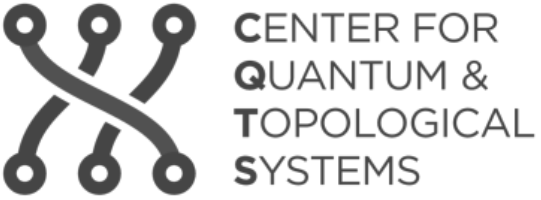}}
}

\begin{abstract}
 More profound than bulk topological order of quantum materials is only its unwinding via gapless excitations along boundaries of the sample. We recast this \emph{bulk-edge correspondence} --- for the experimentally relevant case of fractional quantum Hall (FQH) systems --- in terms of effective relative higher gauge theory, controlled by choices of classifying fibrations. For FQH systems, we identify the complex Hopf fibration as classifying the bulk/boundary topological effects, and find that it yields a non-Lagrangian reconstruction of Floreanini-Jackiw/Wess-Zumino-Witten chiral edge currents.
 
 Remarkably, the resulting effective FQH higher gauge theory turns out to be ``geometrically engineered'' on M2/M5-branes probing A-type orbi-singularities in 11D supergravity, globally completed by flux-quantization in twisted equivariant differential (TED) Cohomotopy: Here the M-string ends of M2-branes on M5-branes engineer the FQH liquid's boundary.
 This geometric engineering on M-branes might naturally elucidate the curious combination of $W_\infty$-symmetry and of super-symmetry that is known to govern the collective excitations of FQH liquids at long wavelengths.
\end{abstract}

\date{\today}

\maketitle

\newpage

\tableofcontents

\newpage

\section
{Introduction}

The field of \emph{topological phases of matter} (cf. \parencites{MoessnerMoore2021}{Stanescu2020}{BernevigHughes2013}) rests on a curious dichotomy: While its eponymous topological effects reside in the interior ``bulk'' of energetically gapped and thereby isolated quantum materials, any observation of these effects (and hence their technological exploitation, such as for much-anticipated topological quantum computing hardware \cite{FreedmanKitaevLarsenWang2003,SatiValera2025}) must occur at boundaries (at \emph{edges} for 2D electron liquids, cf. \cite[\S 1.3]{Wen1995}), where the system's topological twisting ``unwinds'' as it transitions into the topologically trivial vacuum surrounding the topological material (a phenomenon that we formalize in \cref{OnABNCFOrTopologicalOrders}). 

A \emph{bulk-boundary correspondence} (BBC, also \emph{bulk-edge correspondence}, \cite{KellendonkEtAl2002}, cf. \parencites[\S 6]{BernevigHughes2013}{ProdanSchulzBaldes2016}), in this context, is the precise characterization of how dynamical boundary phenomena reflect the bulk topology. 

We highlight that when described in terms of effective field theories (cf. \parencites{Fradkin2013}{Fradkin2024}), the BBC is closely related to the ``holographic principle'' originating in high energy physics (cf. \cite{Zaanen2015}), such as embodied by the famous relation between 3D Chern-Simons (CS) field theory and 2D Wess-Zumino-Witten (WZW) field theory (\parencites{Witten1989}, cf. \cite{CabraRossini1997,GukovEtAl2004}) ---  and we give a new non-Lagrangian perspective on this classical scenario (in \cref{OnIdentifyingEdgeCurrentsInTEDCoh,Conclusions}).

So far, the mathematically most developed formalization of BBC (\cite{KellendonkEtAl2002,KellendonkSchulzBaldes2004}, cf. \parencites{ProdanSchulzBaldes2016}[\S 3]{MathaiThiang2016}) restricts itself to non-interacting topological phases, models their configuration by a short exact sequence of $C^\ast$-algebras $A$ of quantum observables,
\begin{equation}
  \begin{tikzcd}[
    column sep=30pt,
    row sep=-2pt
  ]
    0 
    \ar[r]
    &
    A_{\mathrm{bdr}}
    \ar[rr, hook]
    &&
    A_{\mathrm{\mathrm{full}}}
    \ar[rr, ->>]
    &&
    A_{\mathrm{blk}}
    \ar[r]
    &
    0
    \mathrlap{\,,}
    \\
    &
    \mathclap{
      \substack{
        \text{\color{darkblue}boundary}
        \\
        \text{\color{darkblue}observables}
      }
    }
    \ar[
      rr,
      phantom,
      "{
        \substack{\text{\color{darkgreen}among}}
      }"
    ]
    &&
    \mathclap{
      \substack{
        \text{\color{darkblue}full}
        \\
        \text{\color{darkblue}observables}
      }
    }
    \ar[
      rr,
      phantom,
      "{
        \substack{\text{\color{darkgreen}covering}}
      }"
    ]
    &&
    \mathclap{
      \substack{
        \text{\color{darkblue}bulk}
        \\
        \text{\color{darkblue}observables}
      }
    }
  \end{tikzcd}
\end{equation}
and then interprets the \emph{connecting homomorphism} $\partial_{-1}$ in the induced long exact sequence of operator K-theory groups $K_n(-)$,
\begin{equation}
  \label{TheTraditionalCOnnectingHomomorphism}
  \begin{tikzcd}[column sep=large]
    &
    \cdots
    \ar[r]
    &
    \overset{\mathclap{
      \substack{
        \text{\color{darkblue}bulk phases}
      }}
    }{
    K_0\bracket({
      A_{\mathrm{blk}}
    })
    }
    \ar[
      dll,
      snake left,
      "{ 
        \partial_{-1} 
      }"{
        description,
        yshift=-2pt
      },
      "{
        \substack{\text{\color{darkgreen}
          correspondence
        }}
      }"{swap}
    ]
    \\
    \underset{\mathclap{
      \substack{\color{darkblue}%
        \text{boundary dynamics}}
    }}{
    K_{-1}\bracket({
      A_{\mathrm{bdr}}
    })
    }
    \ar[r]
    & 
    \cdots
    \mathrlap{\,,}
  \end{tikzcd}
\end{equation}
as mapping topological bulk phases to their effect on the boundary. (In applications,  while not an isomorphism by itself, this homomorphism induces an equality of relevant bulk/boundary data under suitable ``trace''-operations.)

But a similar formalization of the BBC for strongly interacting, long-range correlated  and topologically ordered phases has not yet been discussed. In particular, traditional theory for edge modes of fractional quantum Hall systems is in significant tension with experimental results \cite{GuerreroSuarezEtAl2025}, indicating that new theoretical approaches to this seemingly classical topic are still needed. While we will not try to make detailed comparison to experiment here, we take this as motivation to explore a new theoretical foundation for effective edge phenomena of a more robust global nature than traditional Lagrangian approaches.

\section{Methods}
\label{OnMethods}

Here we find a detailed bulk/boundary correspondence for effective (macroscopic, infrared) quantum observables in strongly-coupled systems exhibiting \emph{topological order} \parencites{Wen1991}{Wen2013}[\S II.B]{SS25-Crys} of the kind exhibited by  fractional quantum Hall systems (FQH, cf. \cite{Stormer99}) and their ``anomalous'' crystalline counterparts (FQAH, cf. \parencites{zhao2025exploring}{SS25-FQAH}).

\subsection
{Topological Homotopy}

Our analysis (following \cite{SS25-FQH}) is homotopy theoretical (cf. \cref{OnBackground} for background), which allows us to generalize away from topological K-theory \cref{TheTraditionalCOnnectingHomomorphism} to unstable/nonabelian cohomology theories \cite[\S 2]{FSS23-Char} that intrinsically reflect the topological order of FQH systems (as recalled in \cref{OnRecallingUnboundedFQHOrders}).

To this end, we are dealing with \emph{pointed topological spaces} (cf. \cref{OnSomeGeneralTopology}) 
of two kinds:
\begin{description}
  \item[(i) Domain spaces] $\BulkDomain$, $\BoundaryDomain$, ... 
  
  modeling slabs of quantum materials, including a \emph{point at infinity} $\infty_{{}_{\BulkDomain}} \in \BulkDomain$, ... .

  \item[(ii) Classifying spaces] $\ClassifyingB$, $\ClassifyingA$, ...

  for fields or Bloch Hamiltonians, including the point representing zero, $0_{{}_{\ClassifyingB}} \in \ClassifyingB$, ... .
\end{description}
Between these spaces we consider pointed continuous maps (just \emph{maps}, for short), denoted by arrows, where solid arrows are used to denote specified maps like boundary inclusions $\inlinetikzcd{\BoundaryDomain \ar[r, hook, "{\phi}"] \& \BulkDomain}$ or classifying fibrations $\inlinetikzcd{\ClassifyingA \ar[r, ->>, "\wp"] \& \ClassifyingB}$ or basepoint inclusions, while dashed arrows indicate unspecified maps like (classifying) topological field configurations $\inlinetikzcd{\Sigma \ar[r, dashed, "c_{\BulkDomain}"] \& \ClassifyingB}$. 

These maps are required to preserve the given base points, which for the maps from domain to classifying spaces expresses the \emph{vanishing at infinity} characterizing \emph{solitonic} field configurations:
\footnote{
Cubical diagrams of arrows are always understood to ``commute'', meaning that the two possible diagonal composite maps are implied or required to agree.
 
}
\begin{equation}
  \begin{tikzcd}[row sep=12pt, 
    column sep=75pt
  ]
    \BulkDomain
    \ar[
      r, 
      dashed,
      "{
        \text{classifying map}
      }"{color=darkgreen}
    ]
    &
    \ClassifyingB
    \\
    \{ \infty_{\BulkDomain}\}
    \ar[u, hook]
    \ar[
      r,
      "{
        \text{vanishes at infinity}
      }"{swap, color=darkgreen}
    ]
    &
    \{ 0_{\ClassifyingB} \}
    \mathrlap{\,.}
    \ar[u, hook]
  \end{tikzcd}
\end{equation}

What exactly this means physically is all encoded in the choice of the domain space. For example, for $S$ an unpointed  space (assumed locally compact)  we have
\begin{description}
  \item[-- $\Sigma := S \sqcup \{\infty\}$]
  the disjoint union with $\{\infty\}$,

  modeling the situation where $\infty$ cannot be reached from any other point,
  
  \item[--  $\Sigma := S_{\cpt}$]

  the \emph{one-point compactification}, 
  
  which identifies all the \emph{ends} of $S$ with $\infty$. 
\end{description}
For instance: $(\mathbb{R}^n)_{\cpt} \simeq S^n$ but $(S^n)_{\cpt} \simeq S^n \sqcup \{\infty\}$.

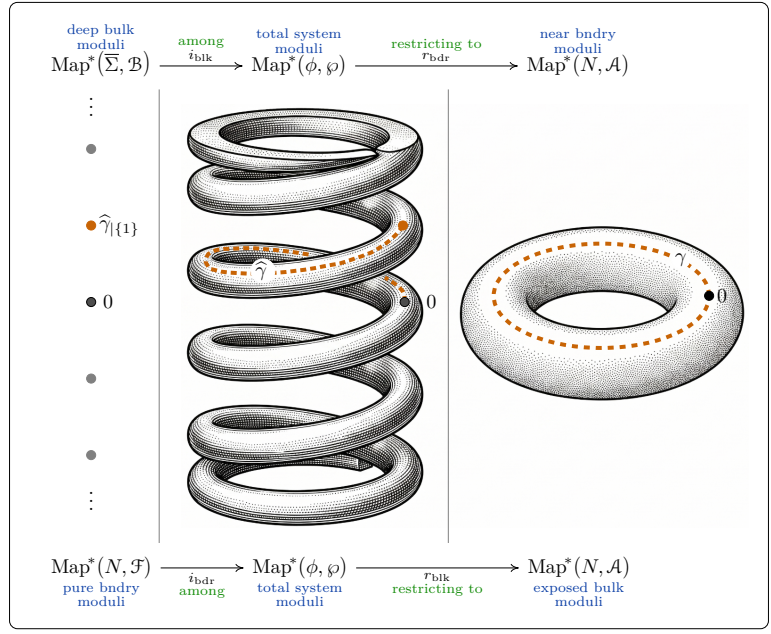
\begin{figure}[htb]
\centering
\adjustbox{
  rndfbox=4pt,
  scale=.75
}{
\hspace{5pt}
\begin{tikzpicture}

\node at (0,0) {
  \includegraphics[width=10cm]{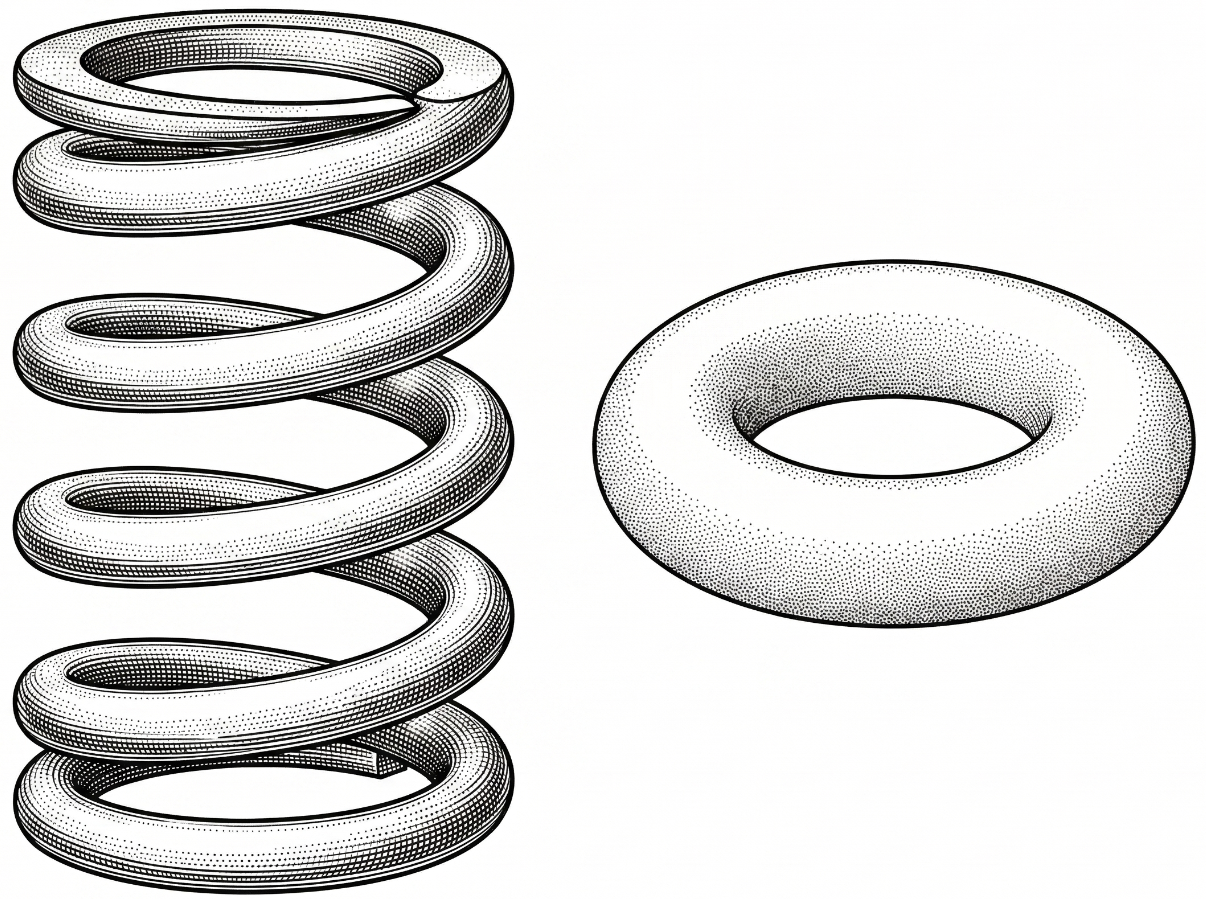}
};

\draw[
  shift={(2.4,.36)},
  darkorange,
  dashed,
  line width=2
]
  (0,0) ellipse (1.9 and .92);
\draw[
  shift={(2.4,.36)},
  fill=black
]
  (0:1.9) circle (.08);

\draw[
  darkorange,
  densely dashed,
  line width=2
]
  plot[
    smooth, 
    tension=1
  ]
  coordinates {
    (-1.1,1.6)
    (-1.45,1.25)
    (-2,1)
    (-2.5,.88)
    (-3,.8)
    (-3.5,.76)
    (-4,.76)
    (-4.4,.81)
    (-4.55,.9)
    (-4.56,1.05)
    (-4.45,1.14)
    (-4.3,1.18)
    (-4.1,1.19)
    (-3.8,1.2)
    (-3.5,1.18)
    (-3.2,1.15)
    (-2.9,1.11)
    (-2.73,1.08)
  };

\draw[
  darkorange,
  densely dashed,
  line width=2
]
  plot[
    smooth, 
    tension=1
  ]
  coordinates {
    (-1.07,.25)
    (-1.09,.4)
    (-1.17,.51)
    (-1.3,.61)
    (-1.42,.68)
  };

\draw[
  fill=black!70
]
  (-1.07,.25) circle (.08);

\draw[
  fill=darkorange,
  draw=darkorange
]
  (-1.1,1.6) circle (.08);

\node at
  (-.6,.25) {0};
\node at
  (4.53,.36) {0};

\node at
  (-6.6,+3.8) {$\vdots$};

\draw[
  fill=black!70
]
  (-6.6,.25) circle (.08);
\node at
  (-6.3,.25) {0};

\draw[
  fill=darkorange,
  draw=darkorange
]
  (-6.6,1.6) circle (.08);

\draw[
  draw=gray,
  fill=gray
]
  (-6.6,1.6+1.6-.25) circle (.08);

\draw[
  draw=gray,
  fill=gray
]
  (-6.6,.25-1.6+.25) circle (.08);

\draw[
  draw=gray,
  fill=gray
]
  (-6.6,.25-1.6+.25-1.6+.25) circle (.08);

\node at
  (-6.6,.25-1.6+.25-1.6+.25-.7) {$\vdots$};

\draw[
  draw=white,
  fill=white
]
  (3.8,1) circle (.16);
\node at (3.8,1) 
  {$\gamma$};

\draw[
  draw=white,
  fill=white
]
  (-3.57,.8) circle (.21);
\node at (-3.6,.8) 
  {$\widehat{\gamma}$};

\node at (-6.1,1.6) 
  {$\widehat{\gamma}_{\vert \{1\}}$};

\node at (-7.4,4.7) {
  \rlap{
    $
      \begin{tikzcd}[
        column sep=42pt,
        ampersand replacement=\&
      ]
        \overset{
          {\color{darkblue}
          \substack{
            \text{deep bulk}
            \\
            \text{moduli}
          }}
        }
        {
        \Map\bracket({
          \DeepBulkDomain,
          \ClassifyingB
        })
        }
        \ar[
          r,
          "{
            \overset{
              \substack{\color{darkgreen}
                \text{among}
              }
            }
            {
              i_{\mathrm{blk}}
            }
          }"
        ]
        \&
        \overset{
          {\color{darkblue}
          \substack{
            \text{total system}
            \\
            \text{moduli}
          }}
        }
        {
          \Map\bracket({
            \phi, \wp
          })
        }
        \ar[
          r,
          "{
            \overset{
            \substack{
              \color{darkgreen}
              \text{restricting to}
            }
            }
            {
              r_{\mathrm{bdr}}
            }
          }"
        ]
        \&[40pt]
        \overset{
          {\color{darkblue}
          \substack{
            \text{near bndry}
            \\
            \text{moduli}
          }}
        }{
        \Map\bracket({
          \BoundaryDomain,
          \ClassifyingA
        })
        }
      \end{tikzcd}
    $
  }
};

\node at (-7.4,-4.7) {
  \rlap{
    $
      \begin{tikzcd}[
        column sep=42pt,
        ampersand replacement=\&
      ]
        \underset{
          {\color{darkblue}
          \substack{
            \text{pure bndry}
            \\
            \text{moduli}
          }}
        }
        {
        \Map\bracket({
          \BoundaryDomain,
          \ClassifyingF
        })
        }
        \ar[
          r,
          "{
            \underset{
              \substack{\color{darkgreen}
                \text{among}
              }
            }
            {
              i_{\mathrm{bdr}}
            }
          }"{swap}
        ]
        \&
        \underset{
          {\color{darkblue}
          \substack{
            \text{total system}
            \\
            \text{moduli}
          }}
        }
        {
          \Map\bracket({
            \phi, \wp
          })
        }
        \ar[
          r,
          "{
            \underset{
            \substack{
              \color{darkgreen}
              \text{restricting to}
            }
            }
            {
              r_{\mathrm{blk}}
            }
          }"{swap}
        ]
        \&[40pt]
        \underset{
          {\color{darkblue}
          \substack{
            \text{exposed bulk}
            \\
            \text{moduli}
          }}
        }{
        \Map\bracket({
          \BoundaryDomain,
          \ClassifyingA
        })
        }
      \end{tikzcd}
    $
  }
};

\draw[
  gray
]
  (-5.4,4) --
  (-5.4,-4);

\draw[
  gray
]
  (-.3,4) --
  (-.3,-4);

\end{tikzpicture}
}
\caption{
  \label{UnwindingOfPhases}
  The topological moduli space of the total system forms a couple \cref{HomotopyFiberSequencesOfModuli} of homotopy fibrations (\cref{TheLongExactSequences}) exhibiting relations between bulk and boundary moduli. Schematically indicated is:
   on the left 
     the space of deep bulk or pure boundary moduli, which as such may form separate topological phases $\widehat{\gamma}_{\vert \{1\}}$;
    in the middle the ambient space of total system moduli, in which such nominally separate bulk/boundary phases may be connected by deformation paths $\widehat{\gamma}$ to the trivial phase $0$;
      on the right, the corresponding loop $\gamma$ in the complementary space of near boundary or exposed bulk moduli.
   This association constitutes a map of phases (homotopy classes) from right to left, this being the \emph{connecting homomorphisms}
   (\cref{ConnectingHomomorphism})
   $\Unwinding_n$
   \cref{TheUnwindingLES}
   or
   $\Rewinding_n$ \cref{TheRewindingLES}  in the induced homotopy long exact sequences (\cref{HomotopyLES}) shown in \cref{OnBoundaryUnwindingOfBulkOrders} and \cref{OnBulkRewindingOfBoundaryShifts}, respectively, as illustrated in the analogous \cref{ConnectingHomomorphismSchematics}.
}
\end{figure}

\subsection
{Bounded Topological Order in Relative Nonabelian Cohomology}
\label
{OnBoundedTopoOrderOnRelNonabCohomology}

The basis for our analysis of topological phases are the following general observations which we will illustrate in a moment:
\begin{enumerate}
\item
 Topological phases are not necessarily classified by a ``stable'' cohomology theory like K-theory with its classifying space $\mathrm{Fred}$ of Fredholm operators, $\tilde K(-) = \pi_0 \, \Map\bracket({-,\mathrm{Fred}})$, but generally by an ``unstable'', ``fragile''
 (cf. \cite{SS25-Crys})
 or ``nonabelian'' cohomology (cf. \parencites[\S 2]{FSS23-Char}[\S 1]{SS25-TEC}) with general pointed topological classifying spaces $\ClassifyingB$ \cref{CohomologyAsHomotopyClasses}:
 \footnote{
    The tilde notation ``$\tilde H$'' in \cref{DefinitionOnNonabelianCohomology} is, as usual, for ``reduced cohomology''.
    In the case $\BulkDomain := S \sqcup \{\infty\}$ this specializes to ``unreduced cohomology'':
    $
      H^\bullet({S; -})
      \defneq
      \tilde H^\bullet\bracket({
        S \sqcup \{\infty\};-
      })
    $.
  }
\begin{equation}
  \label{DefinitionOnNonabelianCohomology}
  \tilde H^{-k}\bracket({\BulkDomain;\ClassifyingB})
  :=
  \pi_k\, \Map\bracket({
    \BulkDomain; \ClassifyingB
  })
  =
  \pi_k
  \big\{
  \begin{tikzcd}[sep=15pt]
    \BulkDomain
    \ar[r, dashed]
    &
    \ClassifyingB
  \end{tikzcd}
  \big\}
  \mathrlap{\,,}
\end{equation}
\begin{enumerate}
\item
topological phases themselves are classified in degree 0,
\item
topological order is in degree -1 --- namely is the monodromy of gapped quantum ground states $\vert \psi \rangle \in \HilbertSpace$ along parameter paths, constituting linear representations 
\begin{equation}
  \label{UnboundedOrderRepresentation}
  \begin{tikzcd}
    \tilde H^{-1}\bracket({
      \BulkDomain, \ClassifyingA
    })
    \ar[r]
    &
    \mathrm{Aut}\bracket({
      \HilbertSpace
    })
  \end{tikzcd}
\end{equation}
of the fundamental group of the space of topological parameters (cf. \cite[\S II.B]{SS25-Crys}).

\end{enumerate}

\item In the presence of a \emph{boundary inclusion}, 
\begin{equation}
\label{BoundaryInclusionInIntroduction}
  \inlinetikzcd{ 
    \BoundaryDomain 
    \ar[r, hook, "{\phi}"] 
    \& 
    \BulkDomain 
  } 
  \text{ with quotient } 
  \DeepBulkDomain,
\end{equation}
this classifying space is refined to a \emph{classifying fibration} 
\begin{equation}
  \label{ClassifyingFibrationInIntroduction}
  \inlinetikzcd{
    \ClassifyingA 
    \ar[r, ->>, "{ \wp }"] 
    \& 
    \ClassifyingB
  }
  \text{ with fiber }
  \ClassifyingF
  \mathrlap{\,,}
\end{equation}
and the nonabelian cohomology to its \emph{twisted relative cohomology} version \parencites[\S 1]{SS25-Orient}[\S 3.2]{BaSS26-MString} given by the homotopy of compatible bulk/boundary pairs of classifying maps:
\begin{equation}
\label{TwistedRelativeNonabelianCohomology}
  \begin{aligned}
  \tilde H^{-k}\bracket({
    \phi; \wp
  })
  &
  :=
  \pi_k\, \Map({
    \phi, \wp
  })
  \\[8pt]
  &
  :=
  \pi_k
  \bracket({
  \Map\bracket({
    \BulkDomain, \ClassifyingB
  })
  \underset{
    \Map\scaledbracket({
      \BoundaryDomain,
      \ClassifyingB
    })
  }{
    \times
  }
  \Map\bracket({
    \BoundaryDomain, 
    \ClassifyingA
  })
  })
  =
  \pi_k
  \smash[t]{
  \left\{
  \begin{tikzcd}[row sep=10pt, 
    ampersand replacement=\&
  ]
    \BoundaryDomain
    \ar[d, hook, "{ \phi }"]
    \ar[
      r, 
      dashed,
      "{ c_N }"
    ]
    \&
    \ClassifyingA
    \ar[
      d,
      ->>,
      "{ \wp }"
    ]
    \\
    \BulkDomain
    \ar[
      r, 
      dashed,
      "{ c_\BulkDomain }"
    ]
    \&
    \ClassifyingB
  \end{tikzcd}
  \,\right\}
  }
  \mathrlap{\,.}
  \end{aligned}
\end{equation}
In words, the commuting diagram on the right encodes how the bulk topology, classified by maps $\inlinetikzcd{\BulkDomain \ar[r, dashed] \& \ClassifyingB}$ as in \cref{DefinitionOnNonabelianCohomology}, compatibly restricts to boundary effects classified by maps $\inlinetikzcd{\BoundaryDomain \ar[r, dashed] \& \ClassifyingA}$.

In particular, the ground state Hilbert space $\HilbertSpace$ of the total system adiabatically transforms under the monodromy (fundamental) group of this twisted relative moduli space, in generalization of \cref{UnboundedOrderRepresentation}:
\begin{equation}
\label{BoundedOrderRepresentation}
  \begin{tikzcd}
    \tilde H^{-1}\bracket({
      \phi; \wp
    })
    \ar[r]
    &
    \mathrm{Aut}\bracket({
      \HilbertSpace
    }).
  \end{tikzcd}
\end{equation}
\end{enumerate}

Central to our discussion is now the observation (\cref{TheLongExactSequences}) that these moduli spaces for bulk/boundary fields form a couple of homotopy fiber sequences of this form:
\begin{subequations}
  \label{HomotopyFiberSequencesOfModuli}
  \begin{align}
  \label{HomotopyFiberSequencesOfModuli-PureBdr-Full-ExposedBulk}
  \begin{tikzcd}[
    ampersand replacement=\&,
    column sep=45pt,
    row sep=-2pt
  ]
    \Map\bracket({
      \BoundaryDomain,
      \ClassifyingF
    })
    \ar[
      rr,
      "{
        \BoundaryInclusion
        \,:\,
        c_N 
        \,\mapsto\,
        (c_N,0)
      }"
    ]
    \&\&
    \Map\bracket({
      \phi, \wp
    })
    \ar[
      rr,
      "{
        \BulkRestriction
        \,:\,
        (c_\BoundaryDomain, 
        c_\BulkDomain)
        \,\mapsto\,
        c_\BulkDomain
      }"
    ]
    \&\&
    \Map\bracket({
      \BulkDomain,
      \ClassifyingB
    })
    \mathrlap{\,,}
    \\
    \substack{
      \text{\color{darkblue}pure bndry}
      \\
      \text{\color{darkblue}moduli}
    }
    \ar[
      rr,
      phantom,
      "{
        \substack{
          \text{\color{darkgreen}among}
        }
      }"
    ]
    \&\&
    \substack{
      \text{\color{darkblue}total system}
      \\
      \text{\color{darkblue}moduli}
    }
    \ar[
      rr,
      phantom,
      "{
        \substack{\text{
          \color{darkgreen}restricting to
        }}
      }"
    ]
    \&\&
    \substack{
      \text{\color{darkblue}exposed bulk}
      \\
      \text{\color{darkblue}moduli}
    }
  \end{tikzcd}
  \\
  \label{HomotopyFiberSequencesOfModuli-DeepBulk-Full-NearBdr}
  \begin{tikzcd}[
    ampersand replacement=\&,
    column sep=45pt,
    row sep=-2pt
  ]
    \Map\bracket({
      \DeepBulkDomain, 
      \ClassifyingB
    })
    \ar[
      rr,
      "{
        \BulkInclusion 
        \,:\,
        c_{\DeepBulkDomain}
        \,\mapsto\,
        (0,c_{\BulkDomain})
      }"
    ]
    \&\&
    \Map\bracket({
      \phi,\wp
    })
    \ar[
      rr,
      "{
        \BoundaryRestriction 
        \,:\,
        (c_\BoundaryDomain, 
        c_\BulkDomain)
        \,\mapsto\,
        c_N
      }"
    ]
    \&\&
    \Map\bracket({
      \BoundaryDomain, 
      \ClassifyingA
    })
    \mathrlap{\,.}
    \\
    \substack{
      \text{\color{darkblue}deep bulk}
      \\
      \text{\color{darkblue}moduli}
    }
    \ar[
      rr,
      phantom,
      "{
        \substack{
          \text{\color{darkgreen}among}
        }
      }"
    ]
    \&\&
    \substack{
      \text{\color{darkblue}total system}
      \\
      \text{\color{darkblue}moduli}
    }
    \ar[
      rr,
      phantom,
      "{
        \substack{\text{
          \color{darkgreen}
          restricting to
        }}
      }"
    ]
    \&\&
    \substack{
      \text{\color{darkblue}near bndry}
      \\
      \text{\color{darkblue}moduli}
    }
  \end{tikzcd}
  \end{align}
\end{subequations}

\begin{remark}[The domain and coefficient spaces]
\label[remark]{DomainAndCoefficientSpaces}
As indicated in \cref{HomotopyFiberSequencesOfModuli},
here we recognize, 
besides
\begin{description}
  \item[$N$] the \emph{boundary domain},
  \item[$\ClassifyingB$] the 
  \emph{bulk coefficients},
\end{description}
also
\begin{description}
\item[$\DeepBulkDomain$] the \emph{deep bulk domain}, where processes never reach the boundary,

\item[$\BulkDomain$] the \emph{exposed bulk domain}, where processes may reach the boundary,

\item[$\ClassifyingF$] the \emph{pure boundary coefficients}, available even when the bulk field is zero,

\item[$\ClassifyingA$] the \emph{near boundary coefficients}, available when extending into the bulk.
\end{description}
\end{remark}

This makes for a more fine-grained formalization of the situation than traditionally considered. In particular this resolves the following important subtlety, which traditionally has remained implicit: 
\begin{remark}[Topological nature of bounded phases]
A would-be topological phase over a bulk domain $\BulkDomain$ whose gap closes over the boundary $\BoundaryDomain$ is, strictly speaking, \emph{not topological in total} (since it is already not gapped in total). When one still addresses it, as usual, as a topological phase, one is actually referring to the deep bulk domain $\DeepBulkDomain$, where the presence of the boundary may be taken to be negligible. What this means precisely is exactly what is expressed by the above \cref{HomotopyFiberSequencesOfModuli}.
\end{remark}

\subsection
{A Bulk-Boundary Correspondence for Topological Orders}
\label
{OnABNCFOrTopologicalOrders}

Now, the long homotopy exact sequences (\cref{HomotopyLES}) associated with these homotopy fiber sequences \cref{HomotopyFiberSequencesOfModuli} involve the cohomology sets \cref{TwistedRelativeNonabelianCohomology},
on which the connecting homomorphisms (\cref{ConnectingHomomorphism}) are different from but conceptually akin to those traditionally considered in \cref{TheTraditionalCOnnectingHomomorphism}, as follows (\cref{TheLongExactSequences}).

\subsubsection
{Boundary unwinding of bulk orders}
\label
{OnBoundaryUnwindingOfBulkOrders}

First, for the homotopy fiber sequence \cref{HomotopyFiberSequencesOfModuli-DeepBulk-Full-NearBdr}, our \cref{TheLongExactSequences} yields a long exact sequence of the following form and interpretation:
\begin{equation}
  \label{TheUnwindingLES}
  \begin{tikzcd}[
    column sep=60pt
  ]
    &&
    \overset
    {
      {\color{darkblue}
      \substack{
        \text{near bndry}
        \\
        \text{unwindings}
      }}
    }
    {
      \tilde H^{-2}\bracket({
        \BoundaryDomain;
        \ClassifyingA
      })
    }
    \ar[
      dll,
      snake left,
      "{ 
        \Unwinding_1 
      }"{description, yshift=-4},
      "{
        \text{\color{darkgreen}
          deep bulk orders unwinding near bndry
        }
      }"{swap}
    ]
    \\
    \underset
    {
      \color{darkblue}
      \substack{
        \text{deep bulk}
        \\
        \text{top. orders}
      }
    }
    {
    \tilde H^{-1}\bracket({
      \DeepBulkDomain;
      \ClassifyingB
    })
    }
    \ar[
      r,
      "{
        \BulkInclusion
      }",
      "{
        \substack{\color{darkgreen}
          \text{
            among
          }
        }
      }"{swap}
    ]
    &
    \underset
    {
      {\color{darkblue}
      \substack{
        \text{total system}
        \\
        \text{top. orders}
      }}
    }
    {
      \tilde H^{-1}({
        \phi; \wp
      })
    }
    \ar[
      r,
      "{
        \BoundaryRestriction
      }",
      "{
        \color{darkgreen}
        \substack{
          \text{restrict to}
        }
      }"{swap}
    ]
    &
    \underset
    {
      \color{darkblue}
      \substack{
        \text{near bndry}
        \\
        \text{top. orders}
      }
    }
    {
      \tilde H^{-1}\bracket({
        \BoundaryDomain; \ClassifyingA
      })\,.
    }
  \end{tikzcd}
\end{equation}
That the connecting homomorphism $\Unwinding_1$ here indeed reflects the \emph{unwindings} of the deep bulk topological orders, when brought in contact with the boundary, may be seen from its explicit formula \cref{AConnectingHomomorphism}, illustrated in \cref{UnwindingOfPhases,ConnectingHomomorphismSchematics}. 

We highlight now that when the boundary is gapless, then also any system containing the boundary is gapless in its totality and hence in particular not topologically ordered, whence the last two terms in \cref{TheUnwindingLES} vanish:
\begin{equation}
\label
{GaplessBoundaryImpliesVanishingTotalCohomologyInNegativeDegree1}
  \text{boundary gapless}
  \;\;\Rightarrow\;\;
  \tilde H^{-1}\bracket({
    \phi; \wp
  })
  = 
  1
  \mathrlap{\,.}
\end{equation}
This is evident but may deserve highlighting in view of traditional terminology, which commonly speaks of \emph{topological phases with topologically trivial boundaries}. What this should really mean is that the bottom left term in \cref{TheUnwindingLES} is non-vanishing, hence that there may be topological order \emph{deep in the bulk}.

But then exactness of \cref{TheUnwindingLES} shows, first, that these deep bulk topological orders are all reflected in their unwinding processes near the boundary, in that the connecting homomorphism is surjective:
\begin{equation}
  \label{SurjectiveUnwindingBBC}
  \substack{
    \text{boundary} 
    \\
    \text{gapless}
  }
  \;\;\;\;
  \Rightarrow
  \;\;\;\;
  \begin{tikzcd}[
    column sep=30pt
  ]
    &&
    \overset
    {
      {\color{darkblue}
      \substack{
        \text{near bndry}
        \\
        \text{unwindings}
      }}
    }
    {
      \tilde H^{-2}\bracket({
        \BoundaryDomain;
        \ClassifyingA
      })
    }
    \ar[
      dll,
      ->>,
      snake left,
      "{ 
        \Unwinding_1 
      }"{description, yshift=-4},
      "{
        \text{\color{darkgreen}
          {\color{purple}all} deep bulk orders unwind near bndry
        }
      }"{swap}
    ]
    \\
    \underset
    {
      \color{darkblue}
      \substack{
        \text{deep bulk}
        \\
        \text{top. orders}
      }
    }
    {
    \tilde H^{-1}\bracket({
      \DeepBulkDomain;
      \ClassifyingB
    })
    }
    \mathrlap{\,.}
  \end{tikzcd}
\end{equation}

With this, the exactness of \cref{TheUnwindingLES} implies
furthermore that the possible ambiguities in the near boundary unwinding of bulk phases are the boundary-restrictions of unwindings that exist already in the total system, hence that we have a genuine \emph{correspondence} (isomorphism) between bulk orders and their near boundary unwindings in this case:
\begin{equation}
  \label{TheConditionalBBC}
  \substack{
    \text{Boundary} 
    \\
    \text{gapless}
  }
  \;\Rightarrow\;
  \left(
    \overset
    {
      {\color{darkblue}
      \substack{
        \text{bdry restrictions}
        \\
        \text{of total unwindings}
      }}
    }
    {
    (r_{\mathrm{bdr}})_\ast
    \bracket({
     \tilde H^{-2}\bracket({
       \phi; \wp
     })
    })
    }
    =
    0
    \;\;
    \Leftrightarrow
    \;\;
  \begin{tikzcd}[
    column sep=22pt
  ]
    &&
    \overset
    {
      {\color{darkblue}
      \substack{
        \text{near bndry}
        \\
        \text{unwindings}
      }}
    }
    {
      \tilde H^{-2}\bracket({
        \BoundaryDomain;
        \ClassifyingA
      })
    }
    \ar[
      dll,
      equals,
      snake left,
      "{ 
        \Unwinding_1 
      }"{description, yshift=-2},
      "{
        \text{\color{darkgreen}
          {deep bulk orders unwind
          \color{purple}uniquely} 
          near bndry
        }
      }"{swap}
    ]
    \\
    \underset
    {
      \color{darkblue}
      \substack{
        \text{deep bulk}
        \\
        \text{top. orders}
      }
    }
    {
    \tilde H^{-1}\bracket({
      \DeepBulkDomain;
      \ClassifyingB
    })
    }
  \end{tikzcd}
  \;\,
  \right)
  \mathrlap{.}
\end{equation}

This is a sharp version of the bulk/boundary correspondence, in a strict sense, for topological orders. Moreover, we will find that the full sequence \cref{TheUnwindingLES} yields relevant bulk/boundary information even when the boundary and total system are not unordered (cf. \cref{OnOverTheClosedAnnulus}).

\subsubsection
{Bulk rewinding of boundary shifts}
\label{OnBulkRewindingOfBoundaryShifts}

But there is also the converse perspective, where instead of watching the bulk topology unwind across the boundary (\cref{OnBoundaryUnwindingOfBulkOrders}), we look at the corresponding pure boundary phenomenon and then ``rewind'' how that originates in the bulk. It is now clear that this converse perspective is analogously described by the other homotopy fiber sequence \cref{HomotopyFiberSequencesOfModuli-PureBdr-Full-ExposedBulk}, inducing a long exact sequence of cohomology groups of the following form and interpretation:
\begin{equation}
  \label{TheRewindingLES}
  \begin{tikzcd}[
   column sep=60pt
  ]
    &&
    \overset{
      {\color{darkblue}
      \substack{
        \text{exposed bulk}
        \\
        \text{rewindings}
      }}
    }
    {
    \tilde H^{-2}\bracket({
      \BulkDomain; 
      \ClassifyingB
    })
    }
    \ar[
      dll,
      snake left,
      "{
        \color{darkgreen}
        \substack{
          \text{pure boundary shifts rewinding into the bulk}
        }
      }"{swap},
      "{ 
        \Rewinding_1
      }"{description, yshift=-4pt}
    ]
    \\
    \underset{\color{darkblue}
      \substack{
        \text{pure bndry}
        \\
        \text{spectral shifts}
      }
    }{
      \tilde H^{-1}\bracket({
        \BoundaryDomain;
        \ClassifyingF
      })
    }
    \ar[
      r,
      "{
        i_{\mathrm{bdr}}
      }",
      "{
        \color{darkgreen}
        \substack{\text{among}}
      }"{swap},
    ]
    &
    \underset
    {
      {\color{darkblue}
      \substack{
        \text{total system}
        \\
        \text{top. orders}
      }}
    }
    {
      \tilde H^{-1}\bracket({
        \phi; \wp
      })
    }
    \ar[
      r,
      "{
        i_{\mathrm{blk}}
      }",
      "{
        \color{darkgreen}
        \substack{\text{restrict to}}
      }"{swap},
    ]
    &
    \underset{
      {\color{darkblue}
      \substack{
        \text{exposed bulk}
        \\
        \text{top. orders}
      }}
    }
    {
    \tilde H^{-1}\bracket({
      \BulkDomain;
      \ClassifyingB
    })\,.
    }
  \end{tikzcd}
\end{equation}
Here \emph{spectral shifts} are the discrete topological effect on edge modes induced by adiabatic loops in parameter space (cf. \cite[(3.1)]{Wen1990}).

Now, analogously to before, if the exposed bulk is gapless (being gapless at its boundary), and with it so is the total system, then the necessary and sufficient condition for the connecting homomorphism in \cref{TheRewindingLES} to be an isomorphism and hence a BBC is that the bulk restrictions of the bulk rewindings are trivial:
\begin{equation}
  \label{TheSecondConditionalBBC}
  \substack{
    \text{Exposed}
    \\
    \text{bulk} 
    \\
    \text{gapless}
  }
  \;\Rightarrow\;
  \left(
    \overset
    {
      {\color{darkblue}
      \substack{
        \text{bulk restrictions}
        \\
        \text{of total rewindings}
      }}
    }
    {
    (r_{\mathrm{blk}})_\ast
    \bracket({
     \tilde H^{-2}\bracket({
       \phi; \wp
     })
    })
    }
    =
    0
    \;
    \Leftrightarrow
    \;\;\;
  \begin{tikzcd}[
    column sep=22pt
  ]
    &&
    \overset
    {
      {\color{darkblue}
      \substack{
        \text{exposed bulk}
        \\
        \text{rewindings}
      }}
    }
    {
      \tilde H^{-2}\bracket({
        \BulkDomain;
        \ClassifyingB
      })
    }
    \ar[
      dll,
      equals,
      snake left,
      "{ 
        \Rewinding_1 
      }"{description, yshift=-2},
      "{
        \text{\color{darkgreen}
          {pure bdry shifts rewind
          \color{purple}uniquely} 
          in bulk
        }
      }"{swap}
    ]
    \\
    \underset
    {
      \color{darkblue}
      \substack{
        \text{pure bdry}
        \\
        \text{top. orders}
      }
    }
    {
    \tilde H^{-1}\bracket({
      \BoundaryDomain;
      \ClassifyingF
    })
    }
  \end{tikzcd}
  \;\,
  \right)
  \mathrlap{.}
\end{equation}

\subsection
{Refinement to Differential Relative Nonabelian Cohomology}
\label
{OnBndryFieldsInHGT}

Tacitly assumed in the discussion so far is that the homotopy type of the topological moduli space of the system --- over which its local systems of Hilbert spaces of quantum ground states are adiabatically parameterized, forming representations \cref{BoundedOrderRepresentation} of its fundamental group --- is accurately modeled solely by relative maps into a classifying fibration. 

Closer analysis reveals that this statement may receive corrections depending on the exact nature of the equations of motion on the effective higher gauge fields whose global topological configurations constitute this moduli space. 

\subsubsection
{Lightning excursion through geometric homotopy}
\label
{OnExcursionThroughGeometricHomotopy}

In order to explain this one needs methods of geometric homotopy theory whose discussion we relegate to \cref{OnSomeGeometricHomotopy}; but very briefly:
\begin{enumerate}
\item
Equations of motion on flux densities of \emph{higher Maxwell type} are encoded by fibrations $\mathfrak{l}\wp$  of \emph{characteristic $L_\infty$-algebras}, in that their solution spaces are equivalently the spaces of closed (flat) relative $\mathfrak{l}\wp$-valued differential forms on the relative spatial domain $\phi$ (for which we assume a disjoint point at infinity now, not to overburden the discussion):
\begin{equation}
  \mathbf{Sol}(\phi)
  \simeq
  \mathbf{\Omega}^1_{\mathrm{cl}}\bracket({
    \phi; \mathfrak{l}\wp
  })
\end{equation}
(where the boldface indicates that these are generalized \emph{differential geometric} spaces, namely \emph{smooth 0-stacks}, or \emph{smooth sets}, cf. \cref{OnSmoothSetsOfLAlgValuedForms}). 

\item 
More general equations of motion, where some of the flux density species $F^{(i)}$ are of \emph{Chern-Simons type} in that they vanish on-shell, $F^{(i)} = 0$, are encoded by fibrations 
\begin{equation}
\label{ThePrimeMapInIntroduction}
  \begin{tikzcd}
    \mathfrak{l}\wp' 
    \ar[
      r,
      ->>,
      "{
        \prime
      }"
    ] 
    &
    \mathfrak{l}\wp
  \end{tikzcd}
\end{equation}
(of fibrations \cref{RelativeWhiteheadLAlgebra} of $L_\infty$-algebras) whose cokernel picks these $F^{(i)}$.

\item 
The actual solution space of the corresponding higher gauge fields (not just of their flux densities), hence the \emph{phase space} of the \emph{Maxwell/Chern-Simons type} higher gauge theory is then a \emph{smooth $\infty$-stack} or \emph{smooth $\infty$-groupoid} denoted
\begin{equation}
  \label{ThePhaseSpaceInIntro}
  \mathrm{PhsSp}'(\phi)
  =
  \UnpointedMap\bracket({
    \phi,
    \wp
  })
  \;\;\;
  \underset{\mathclap{
    \shape
    \mathbf{\Omega}^1_{\mathrm{cl}}
    \scaledbracket({
      \phi; \mathfrak{l}\wp
    })
  }}
    {\times}
  \;\;
  \mathbf{\Omega}^1_{\mathrm{cl}}\bracket({
    \phi;
    \mathfrak{l}\wp'
  })
  \mathrlap{\,.}
\end{equation}

\item
The underlying homotopy type of this phase space \cref{ThePhaseSpaceInIntro}, called its \emph{shape} $\shape(-)$, is that of the following homotopy fiber product of ordinary relative mapping spaces
\begin{equation}
\label
{PrimedPhaseSpaceInIntroduction}
  \shape \mathrm{PhsSp}'\bracket({
    \phi;
    \wp
  })
  \sim
  \UnpointedMap(\phi,\wp)
  \;\;
  \underset{\mathclap{
    \UnpointedMap\scaledbracket({
      \phi;
      \RRationalization\wp
    })
  }}
  {
    \times^h
  }
  \;\;
  \UnpointedMap\bracket({
    \phi;
    \RRationalization\wp'
  })
  \mathrlap{\,,}
  \hspace{.8cm}
  \begin{tikzcd}[
    column sep=15pt
  ]
    \shape \mathrm{PhsSp}'\bracket({
      \phi;
      \wp
    })
    \ar[
      d,
      "{
        \shape F
      }"
    ]
    \ar[
      r,
      "{ \shape \rchi }"
    ]
    \ar[
      dr,
      phantom,
      "{ \lrcorner_h }"{pos=.1}
    ]
    &
    \UnpointedMap\bracket({
      \phi; 
      \wp
    })
    \ar[
      d,
      "{
        \eta^{\mathbb{R}}_\ast
      }"
    ]
    \\
    \UnpointedMap\bracket({
      \phi;
      \RRationalization\wp'
    })
    \ar[
      r,
      "{
        \shape \prime_\ast
      }"
    ]
    &
    \UnpointedMap\bracket({
      \phi;
      \RRationalization\wp
    })
    \mathrlap{\,,}
  \end{tikzcd}
\end{equation}
where $\RRationalization(-)$ \cref{RRationalization} denotes \emph{rationalization} of (fibrations of) spaces, over the real numbers.
\end{enumerate}

\subsubsection
{The refined twisted relative cohomology set}

Hence, after the dust has settled, it is the homotopy groups of the space \cref{PrimedPhaseSpaceInIntroduction} that generalize those in \cref{TwistedRelativeNonabelianCohomology} used above:
\begin{equation}
\label
{PrimedCohomologySetsInIntro}
  H'^{-k}\bracket({
    \phi;
    \wp
  })
  :=
  \pi_k\bracket({
    \shape \mathrm{PhsSp}'
    \bracket({
      \phi;
      \wp
    })
  })
  \mathrlap{\,.}
\end{equation}

Incidentally, it is immediate that this reduces to the previous notion \cref{TwistedRelativeNonabelianCohomology} when the comparison map \cref{ThePrimeMapInIntroduction} is trivial --- hence when none of the effective flux densities is constrained to vanish on-shell:
\begin{equation}
  \prime = \mathrm{id}
  \;\;\;\;
  \Rightarrow
  \;\;\;\;
  \inlinetikzcd{
    H^{'-k}\bracket({
      \phi; \wp
    })
    \ar[
      r, 
      "{ \sim }",
      "{ \shape \rchi }"{swap}
    ]
    \&
    H^{-k}\bracket({
      \phi; \wp
    })
  }
  \mathrlap{\,,}
\end{equation}
while in general the comparison map $\shape \rchi$ \cref{PrimedPhaseSpaceInIntroduction} will detect a difference.

\subsubsection
{Differential refinement of bulk-boundary correspondence}
\label
{OnDifferentialRefinementOfBBC}

But it is also immediate (since homotopy fiber products commute over each other) that the adjusted space \cref{PrimedPhaseSpaceInIntroduction} compatibly sits in homotopy fiber sequences generalizing the previous ones \cref{HomotopyFiberSequencesOfModuli}. 

\begin{align}
\label{TheRefinedHomotopyFiberSequence}
\adjustbox{scale=.91,center}{
  \begin{tikzcd}[
    ampersand replacement=\&,
    column sep=10pt,
    row sep=-4pt
  ]
  \shape
  \mathrm{PhsSp}'\bracket({
    \DeepBulkDomain;
    \ClassifyingB
  })
  \ar[r]
  \&
  \shape
  \mathrm{PhsSp}'(\phi;\wp)
  \ar[r]
  \&
  \shape\mathrm{PhsSp}'\bracket({
    \BoundaryDomain,
    \ClassifyingA
  })
  \\
  \rotatebox[origin=c]{-90}{$:=$}
  \&
  \rotatebox[origin=c]{-90}{$:=$}
  \&
  \rotatebox[origin=c]{-90}{$:=$}
  \\
  \Map\bracket({
    \DeepBulkDomain,
    \ClassifyingB
  })
  \quad 
  \underset{\mathclap{
    \Map\scaledbracket({
      \DeepBulkDomain,
      \RRationalization
      \ClassifyingB
    })
  }}
  {\times^h}
  \quad 
  \Map\bracket({
    \DeepBulkDomain,
    \RRationalization
    \ClassifyingB'
  })
  \ar[r, shorten=-2pt]
  \&
  \UnpointedMap\bracket({
    \phi,
    \wp
  })
  \quad 
  \underset{\mathclap{
    \UnpointedMap\scaledbracket({
      \phi, 
      \RRationalization\wp
    })
  }}
  {\times^h}
  \quad 
  \UnpointedMap\bracket({
    \phi,
    \RRationalization\wp'
  })
  \ar[r, shorten=-2pt]
  \&
  \UnpointedMap\bracket({
    \BoundaryDomain,
    \ClassifyingA
  })
  \quad 
  \underset{\mathclap{
    \UnpointedMap\scaledbracket({
      \BoundaryDomain,
      \RRationalization
      \ClassifyingA
    })
  }}
  {\times^h}
  \quad 
  \UnpointedMap\bracket({
    \BoundaryDomain,
    \RRationalization
    \ClassifyingA'
  })
  \\
    \substack{
      \text{\color{gray}refined}
      \\
      \text{\color{darkblue}deep bulk}
      \\
      \text{\color{darkblue}moduli}
    }
    \ar[
      r,
      phantom,
      "{
        \substack{
          \text{\color{darkgreen}among}
        }
      }"
    ]
    \&
    \substack{
      \text{\color{gray}refined}
      \\
      \text{\color{darkblue}total system}
      \\
      \text{\color{darkblue}moduli}
    }
    \ar[
      r,
      phantom,
      "{
        \substack{\text{
          \color{darkgreen}
          restricting to
        }}
      }"
    ]
    \&
    \substack{
      \text{\color{gray}refined}
      \\
      \text{\color{darkblue}near bndry}
      \\
      \text{\color{darkblue}moduli}
    }
  \end{tikzcd}
  }
\end{align}

Therefore, our discussion of the bulk-boundary correspondence (\cref{OnBoundaryUnwindingOfBulkOrders}) generalizes straightforwardly to this refinement. In particular, the refined cohomology sets \cref{PrimedCohomologySetsInIntro} sit in long exact sequence refining the \emph{unwinding LES} \cref{TheUnwindingLES}:
\begin{equation}
  \label{TheRefinedUnwindingLES}
  \begin{tikzcd}[
    column sep=60pt
  ]
    &&
    \overset
    {
      {\color{darkblue}
      \substack{
        \text{\color{gray}refined}
        \\
        \text{near bndry}
        \\
        \text{unwindings}
      }}
    }
    {
      \tilde H'^{-2}\bracket({
        \BoundaryDomain;
        \ClassifyingA
      })
    }
    \ar[
      dll,
      snake left,
      "{ 
        \RefinedUnwinding_1 
      }"{description, yshift=-4},
      "{
        \substack{
        \text{\color{gray}refined}
        \\
        \text{\color{darkgreen}
          deep bulk orders unwinding near bndry
        }
        }
      }"{swap}
    ]
    \\
    \underset
    {
      \color{darkblue}
      \substack{
        \text{\color{gray}refined}
        \\
        \text{deep bulk}
        \\
        \text{top. orders}
      }
    }
    {
    \tilde H'^{-1}\bracket({
      \DeepBulkDomain;
      \ClassifyingB
    })
    }
    \ar[
      r,
      "{
        \RefinedBulkInclusion
      }",
      "{
        \substack{\color{darkgreen}
          \text{
            among
          }
        }
      }"{swap}
    ]
    &
    \underset
    {
      {\color{darkblue}
      \substack{
        \text{\color{gray}refined}
        \\
        \text{total system}
        \\
        \text{top. orders}
      }}
    }
    {
      \tilde H'^{-1}({
        \phi; \wp
      })
    }
    \ar[
      r,
      "{
        \RefinedBoundaryRestriction
      }",
      "{
        \color{darkgreen}
        \substack{
          \text{restrict to}
        }
      }"{swap}
    ]
    &
    \underset
    {
      \color{darkblue}
      \substack{
        \text{\color{gray}refined}
        \\
        \text{near bndry}
        \\
        \text{top. orders}
      }
    }
    {
      \tilde H'^{-1}\bracket({
        N; \ClassifyingA
      })\,.
    }
  \end{tikzcd}
\end{equation}

All this holds for any choice of refinement $\prime$ \cref{ThePrimeMapInIntroduction}. Below in \cref{OnIdentifyingThePhaseSpace} we will identify the refinement that accurately reflects the nature of topological order in FQH systems.

\section
{Results}
\label
{OnResults}

We now work out instances of this new bulk/boundary correspondence and compare to physical expectations.
To put this into context, first we briefly recall (\cref{OnRecallingUnboundedFQHOrders}) how our formalism captures FQH phenomenology in the absence of boundaries.

\subsection
{Recalling Unbounded FQH Orders in 2-Cohomotopy}
\label{OnRecallingUnboundedFQHOrders}

The specialization of the general formulation \cref{OnBoundedTopoOrderOnRelNonabCohomology} to fractional quantum Hall systems turns out to be given by taking the bulk classifying space to be the 2-sphere
\begin{equation}
  \ClassifyingB
  :=
  \mathbb{C}P^1 \simeq S^2
\end{equation}
This is the result of \cite{SS25-AbelianAnyons,SS25-FQH} (which may be understood as a refinement of the \emph{Hopfion model} for abelian anyons \cite{Wilczek1983}, cf. \cite[\S II.C]{Forte1992}), based in particular on the following two theorems:

\subsubsection
{FQH Order deep in the Plane}
\label{OnFQHOrderDeepInthePlane}

Consider the case of an FQH liquid in the deep bulk of the disk,
\begin{equation}
  \DeepBulkDomain 
  \defneq
  D^2/\partial D^2 
  \simeq S^2
\end{equation}

The algebra of quantum observables on the topological solitons over this domain is hence (\cite[\S 2.1]{SS25-FQH} following \parencites{SS24-Obs}) the group algebra of the integers:
\begin{equation}
  \mathrm{Obs}
  \defneq
  \mathbb{C}\bracket[{
    \pi_1\, \Map\bracket({
      \DeepBulkDomain,
      S^2
    })
  }]
  \simeq
  \mathbb{C}\bracket[{
   \mathbb{Z} 
  }]
\end{equation}

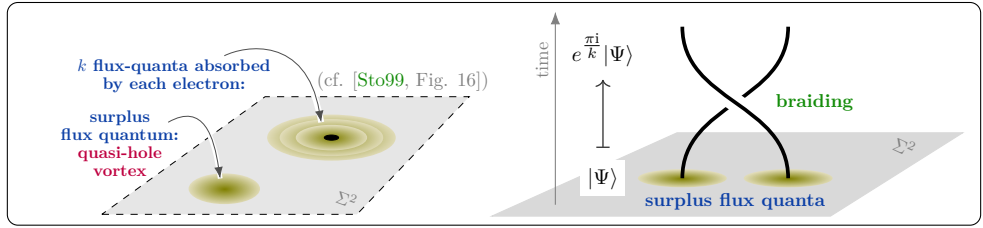
\begin{figure}[htb]
\centering
\adjustbox{
  rndfbox=4pt
}{
\adjustbox{
  scale=.9
}{
\begin{tikzpicture}[
  scale=.75
]

\node
  at (.3,.55+.8)
  {
    \adjustbox{
      bgcolor=white,
      scale=.7
    }{
      \color{darkblue}
      \bf
      \def\arraystretch{.9}
      \begin{tabular}{c}
        surplus
        \\
        flux quantum:
        \\
        \color{purple}
        quasi-hole
        \\
        \color{purple}
        vortex
      \end{tabular}
    }
  };

\draw[
  dashed,
  fill=lightgray
]
  (0,0)
  -- (5,0)
  -- (7+.3-.1,2+.3)
  -- (2.8+.3+.1,2+.3)
  -- cycle;

\begin{scope}[
  shift={(2.4,.5)}
]
\shadedraw[
  draw opacity=0,
  inner color=olive,
  outer color=lightolive
]
  (0,0) ellipse (.7 and .3);
\end{scope}

\begin{scope}[
  shift={(4.5,1.5)}
]

\begin{scope}[
 scale=1.8
]
\shadedraw[
  draw opacity=0,
  inner color=olive,
  outer color=lightolive
]
  (0,0) ellipse (.7 and .25);
\end{scope}

\begin{scope}[
 scale=1.45
]
\shadedraw[
  draw opacity=0,
  inner color=olive,
  outer color=lightolive
]
  (0,0) ellipse (.7 and .25);
\end{scope}

\shadedraw[
  draw opacity=0,
  inner color=olive,
  outer color=lightolive
]
  (0,0) ellipse (.7 and .25);

\begin{scope}[
  scale=.2
]
\draw[
  fill=black
]
  (0,0) ellipse (.7 and .25);
\end{scope}

\end{scope}

\draw[
  white,
  line width=2
]
  (1.3, 1.8) .. controls 
  (2,2.2) and 
  (2.2,1.5) ..
  (2.32,.7);
\draw[
  -Latex,
  black!70
]
  (1.3, 1.8) .. controls 
  (2,2.2) and 
  (2.2,1.5) ..
  (2.32,.7);

\node
  at (1.3,2.7)
  {
    \adjustbox{
      scale=.7
    }{
      \color{darkblue}
      \bf
      \def\arraystretch{.9}
      \def\tabcolsep{-5pt}
      \begin{tabular}{c}
        $k$ flux-quanta
        absorbed
        \\
        by each electron:
      \end{tabular}
    }
  };

\draw[
 line width=2.5pt,
  white
]
  (2.4, 3.1) .. controls 
  (2.8,3.3) and 
  (4,3.5) ..
  (4.3,1.8);

\draw[
  -Latex,
  black!70
]
  (2.4, 3.1) .. controls 
  (2.8,3.3) and 
  (4,3.5) ..
  (4.3,1.8);

\node at 
  (5.9,2.6)
  {
   \scalebox{.8}{
     \color{gray}
     (cf. \cite[Fig. 16]{Stormer99})  
   }
  };

\node[
  gray,
  rotate=-20,
  scale=.73
] 
  at (4.8,+.3) {$\Sigma^2$};

\end{tikzpicture}
}
\hspace{-.5cm}
\begin{tikzpicture}[
    xscale=.7
  ]
    \draw[
      gray!30,
      fill=gray!30
    ]
      (-4.6,-1.5) --
      (+1.8,-1.5) --
      (+1.8+3-.5,-.4) --
      (-4.6+3+.5,-.4) -- cycle;

    \begin{scope}[
      shift={(-1,-1)},
      scale=1.2
    ]
    \shadedraw[
      draw opacity=0,
      inner color=olive,
      outer color=lightolive
    ]
      (0,0) ellipse (.7 and .1);
    \end{scope}

    \draw[
     line width=1.4
    ]
      (-1,-1) .. controls
      (-1,0) and
      (+1,0) ..
      (+1,+1);

  \begin{scope}
    \clip 
      (-1.5,-.2) rectangle (+1.5,1);
    \draw[
     line width=7,
     white
    ]
      (+1,-1) .. controls
      (+1,0) and
      (-1,0) ..
      (-1,+1);
  \end{scope}
  
    \begin{scope}[
      shift={(+1,-1)},
      scale=1.2
    ]
    \shadedraw[
      draw opacity=0,
      inner color=olive,
      outer color=lightolive
    ]
      (0,0) ellipse (.7 and .1);
    \end{scope}
    \draw[
     line width=1.4
    ]
      (+1,-1) .. controls
      (+1,0) and
      (-1,0) ..
      (-1,+1);

  \node[
    rotate=-25,
    scale=.7,
    gray
  ]
    at (3.2,-.58) {
      $\Sigma^2$
    };

  \draw[
    -Latex,
    gray
  ]
    (-3.4,-1.35) -- 
    node[
      near end, 
      sloped,
      scale=.7,
      yshift=7pt
      ] {time}
    (-3.4, 1.2);

  \node[
    scale=.7
  ] at 
    (0,-1.3)
   {\bf \color{darkblue} 
   surplus flux quanta};

  \node[
    scale=.7
  ] at 
    (1.5,0)
   {\bf \color{darkgreen} braiding};

  \node[
    fill=white,
    scale=.8
  ] at (-2.5,-1) {$
    \vert \Psi \rangle
  $};

  \node[
    fill=white,
    scale=.8
  ] at (-2.5,.7) {$
    e^{\tfrac{\pi \mathrm{i}}{k}}
    \vert \Psi \rangle
  $};

  \draw[
    |->,
    black!80,
    line width=.5
  ]
    (-2.5,-.6) --
    (-2.5, .3);

  \end{tikzpicture}
}
\caption{
\label{fluxBraidingSchematics}
{Anyons in FQH liquids} are (quasi-hole vortices associated with) surplus magnetic flux quanta 
(relative to a given rational \emph{filling fraction} of $k$ flux quanta per electron)
through an electron gas occupying an effectively 2-dimensional semiconducting surface $\BulkDomain$. 
The adiabatic \emph{braiding} of worldlines of pairs of such anyons causes the quantum state of the entire system to pick up a fixed complex \emph{braiding phase} factor. 
}
\end{figure}

In order to see what physics these observables actually observe,  
notice, with \cite[Thm. 2.19]{SS25-AbelianAnyons}, that under the \emph{Pontrjagin theorem}:
\begin{enumerate}
\item
elements of the moduli space $\Map\bracket({ \DeepBulkDomain, S^2 })$ are identified with configurations of signed points in the interior of the disk, here understood as configurations of quasi-particle/holes in the FQH system,

\item loops in the moduli space
$\Map\bracket({
  \DeepBulkDomain, S^2
})$ are identified with \emph{framed oriented links}, subject to continuous deformation by link cobordism, here to be understood as vacuum-to-vacuum anyon braiding processes,

\item
passage to connected moduli components computes exactly the total crossing/braiding number or \emph{writhe} of these framed links (cf. \cref{FramedLinks}):
\begin{equation}
  \label{LinksToWrithe}
  \begin{tikzcd}[
   row sep=-2pt
  ]
    \Map\bracket({
      I/\partial I,
      \Map\bracket({
        \DeepBulkDomain,
        S^2  
      })
    })
    \ar[
      d,
      "{ \sim }"{sloped}
    ]
    \ar[
      rr,
      ->>,
      "{ [-] }"
    ]
    &&
    \pi_1\bracket({
      \Map\bracket({
        \DeepBulkDomain,
        S^2  
      })
    })
    \ar[
      d,
      "{ \sim }"{sloped}
    ]
    \\[20pt]
    \Map\bracket({
      \mathbb{R}^3_{\cpt},
      S^2
    })
    \ar[
      rr,
      ->>,
      "{ [-] }"
    ]
    &&
    \pi_0
    \Map\bracket({
      \mathbb{R}^3_{\cpt},
      S^2
    })    
    \\
    \underset{\mathclap{
      \color{darkblue}
      \substack{
        \text{framed link = }
        \\
        \text{anyon vac process}
      }
    }}{L}
    \ar[
      rr,
      |->,
      shorten=15pt
    ]
      &&
    \underset{\mathclap{
      \color{darkblue}
      \substack{
        \text{\textit{writhe} =}
        \\
        \text{net braiding number}
      }
    }}
    {
      \# L
      \in 
      \mathbb{Z}
    }
  \end{tikzcd}
\end{equation}

\begin{figure}[htb]
\centering
\adjustbox{
  rndfbox=4pt
}{
$
\adjustbox{
  raise=-1cm,
  scale=.9
}{
\begin{tikzpicture}[
  scale=1
]

\begin{scope}[
  shift={(1,0)}
]
\draw[line width=2, -Latex]
  (0:1) arc (0:180:1);
\end{scope}

\draw[line width=7,white]
  (0:1) arc (0:180:1);
\draw[line width=2, -Latex]
  (0:1) arc (0:180:1);

\draw[line width=2, -Latex]
  (180:1) arc (180:360:1);

\begin{scope}[shift={(1,0)}]
\draw[line width=7, white]
  (180:1) arc (180:360:1);
\draw[line width=2, -Latex]
  (180:1) arc (180:360:1);
\end{scope}

\node[gray]
  at (.5,.64) {\color{red} 
    \scalebox{.9}{$-$}
  };
\node[gray]
  at (.5,-.64) {\color{red}
    \scalebox{.9}{$-$}
  };
  
\end{tikzpicture}
}
\overset{\#}{\longmapsto}
-2
$
\hspace{.3cm}
$
\adjustbox{
  raise=-2.7cm,
  scale=.55
}{
\begin{tikzpicture}
\foreach \n in {0,1,2} {
\begin{scope}[
  rotate=\n*120-4
]
\draw[
  line width=2.8,
  -Latex
]
 (0-.1,-1+.14)
   .. controls
   (-1,.2) and (-2,2) ..
 (0,2)
   .. controls
   (1,2) and (1,1) ..
  (.9,.7);
\end{scope}

\node[darkgreen]
  at (\n*120+31:.6) {
    \scalebox{1}{$+$}
  };

};
\end{tikzpicture}
\hspace{-25pt}
}
\overset{\#}{\longmapsto}
+3
$
\hspace{.3cm}
$
\adjustbox{
  raise=-1.3cm,
  scale=1
}{
\begin{tikzpicture}

\node at (0,0) {
  \includegraphics[width=2.2cm]{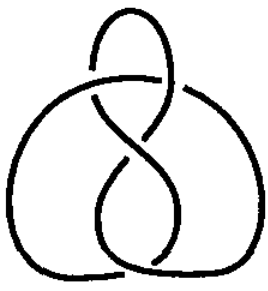}
};

\draw[line width=1.2pt,-Latex]
  (-.276,-.8) -- 
  (-.274,-.8-.02);
\draw[line width=1.2pt,-Latex]
  (.35, -.5) --
  (.35, -.47);
\draw[line width=1.2pt,-Latex]
  (-.09,1.09) --
  (-.09+.05, 1.13);
\draw[line width=1.2pt,-Latex]
  (-.18,.54) --
  (-.19-.01, .54);

\node[red]
  at (0,-1.2) {
    \scalebox{.7}{$-$}
  };
\node[red]
  at (.3,0) {
    \scalebox{.7}{$-$}
  };
\node[darkgreen]
  at (.5,.64) {
    \scalebox{.7}{$+$}
  };
\node[darkgreen]
  at (-.57,.64) {
    \scalebox{.7}{$+$}
  };

\end{tikzpicture}
\hspace{-10pt}
}
\overset{\#}{\longmapsto} 0
$
}
\caption{\label{FramedLinks}
  Some (blackboard-)framed Wilson loop/links and their total crossing/braiding number (writhe), cf. \cref{LinksToWrithe}. 
}
\end{figure}
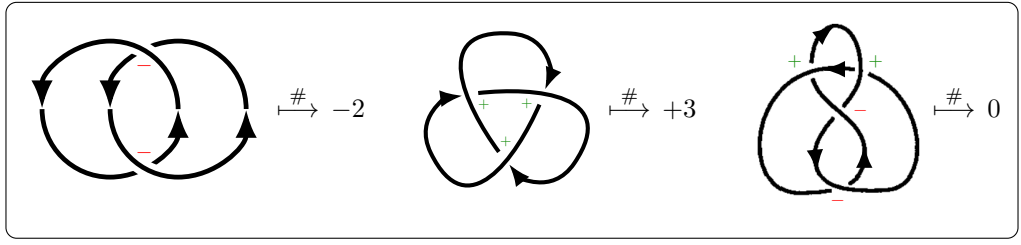

Now, a pure quantum state $\vert k \rangle$ on a commutative algebra of quantum observables corresponds to its expectation value map,
\begin{equation}
  \langle - \rangle
  :=
  \langle k \vert - \vert k \rangle
  \mathrlap{\,,}
\end{equation}
being
a star-algebra homomorphism from observables to probability amplitudes \cite[Prop. 3.2]{SS25-AbelianAnyons}:
\begin{equation}
  \label{ExpectationValueOfAnyonsInPlane}
  \begin{tikzcd}[
    row sep=0pt
  ]
    \mathbb{C}\bracket[{\mathbb{Z}}]
    \ar[
      rr,
      "{
        \langle k\vert 
          - 
        \vert k \rangle
      }"
    ]
    &&
    \mathbb{C}
    \\
    {[L]} 
    \ar[
      rr,
      |->,
      shorten=6pt
    ]
      &&
    \exp\bracket({
      \tfrac{\pi \mathrm{i}}{k}
      \# L
    })
    \mathrlap{\,,}
  \end{tikzcd}
\end{equation}
which as such is fixed, as shown, by an element $k \in \mathbb{R} \setminus \{0\}$ (which will be quantized to an integer in \cref{OnFQHOrderOnTheTorus}).
This means that the expectation value of these operators in a pure state $\vert k \rangle$ is
\begin{equation}
  \bracket\langle{
    [L]
  }\rangle
  =
  \exp\bracket({
    \tfrac{\pi \mathrm{i}}{k}
    \# L
  })
  \mathrlap{\,.}
\end{equation}
But this is:
\begin{enumerate}
  \item[\bf (a)]
  exactly the expectation value of Wilson loop observables $L$ in abelian Chern-Simons theory at level $k$ (\parencites[Rem. 3.4]{SS25-AbelianAnyons}[\S 2.2.4]{SS25-WilsonLoops}),
  \item[\bf (b)]
  given by assigning a unit \emph{braiding phase} 
  \begin{equation}
    \label{BraidingOnTheDisk}
    \zeta 
      = 
    e^{\pi \mathrm{i}/k}
  \end{equation}
  to each crossing of a pair of quasi-hole worldlines along their trajectory $L$, which is exactly the observed anyon braiding phase in FQH systems at filling fraction $\nu = 1/k$ (cf. \cref{fluxBraidingSchematics}).
\end{enumerate}

This shows that and how $\ClassifyingB \defneq \mathbb{C}P^1 \simeq S^2$ is indeed a classifying space for the topology of FQH systems deep in their bulk.

\end{enumerate}

\subsubsection
{FQH Order on the Torus}
\label{OnFQHOrderOnTheTorus}

Next, consider the case that the bulk domain is the torus $\BulkDomain = T^2 := S^1 \times S^1$. 
By a remarkable result going back to \parencites{Hansen1974} we have:
\begin{proposition}[{\parencites[Thm 1]{LarmoreThomas1980}[Prop. 1.5]{Kallel2001}[Thm. 3.3]{KSS26-HigherDimAnyons}}]
\label[proposition]{FundamentalGrouOfMapsFromT2ToS2}
  The fundamental group of maps from the torus to the 2-sphere is the integer Heisenberg group at level=2:
  \begin{equation}
    \label{IntegerHeisenbergGroup}
    \pi_1 \UnpointedMap\bracket({
      T^2, S^2
    })
    \simeq
    \mathrm{Heis}_3(\mathbb{Z})
    \defneq
    \bracket\langle{
      W_b, W_b, \zeta
    }\rangle
    \big/
    \bracket({
      [W_a, W_b] = \zeta^2,
      \;
      [\zeta,-] = \mathrm{e}
    })
    \,.
  \end{equation}
\end{proposition}
The corresponding group algebra, $\mathbb{C}\bracket[{ \mathrm{Heis}_3(\mathbb{Z}) }]$, is exactly the algebra of Wilson loop observables in abelian Chern-Simons theory on the torus, hence of FQH anyon observables on the torus (cf. \cite[(4.9)]{WenNiu1990}\cite[(4.14)]{IengoLechner1992}, reviewed in \cite[(4.21)]{Fradkin2013}\cite[(5.28)]{Tong2016} and Fig. \ref{FigureTorusAlgebra}).
\footnote{
  While routinely discussed in theory, the experimental realization of toroidal geometries for actual FQH liquids is elusive (remembering that it requires a strong magnetic field everywhere transversal to the surface). But in the ``anomalous'' momentum-space realization of the fractional quantum Hall effect in fractional Chern insulators (cf. \cite{zhao2025exploring}), where the role of the magnetic flux density is instead taken by the intrinsic Berry curvature of the system, this toroidal geometry is actually the default, being the Brillouin torus of electron quasi-momenta. Discussion of our cohomotopical description for these dual FQAH systems is in \cite{SS25-FQAH,SS25-Crys}.
}

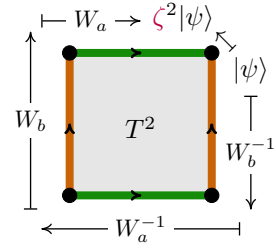
\begin{SCfigure}[1.6][htb]
\caption{\label{FigureTorusAlgebra}
  The group algebra of the integer Heisenberg group \cref{IntegerHeisenbergGroup} is that characteristic of observables on FQH anyons on a torus, where the quantum state $\vert \psi \rangle$ of the system changes by the \emph{square} of the anyon braiding phase $\zeta$ (\cref{fluxBraidingSchematics}) as one moves back and forth along a pair of basis 1-cycles in the torus (cf. \cite[\S 3.3]{SS25-FQH}).
}
\begin{tikzpicture}[scale=0.94,
  >={
    Computer Modern Rightarrow[
      length=4pt, width=4pt
    ]
  }
]

\draw[
  fill=lightgray,
]
  (0,0) rectangle (2,2);

\node at (1,.8) {%
  \clap{\smash{$T^2$}}%
};

\draw[
  line width=3,
  color=darkorange
] 
  (0,0) --
  (0,2);
\draw[
  ->,
  line width=1
]
  (0,1) --
  (0,1.01);

\draw[
  line width=3,
  color=darkorange
] 
  (2,0) --
  (2,2);
\draw[
  ->,
  line width=1
]
  (2,1) --
  (2,1.01);

\draw[
  line width=3,
  color=darkgreen
] 
  (0,2) --
  (2,2);
\draw[
  ->,
  line width=1
]
  (1,2) --
  (1.01,2);

\draw[
  line width=3,
  color=darkgreen
] 
  (0,0) --
  (2,0);
\draw[
  ->,
  line width=1
]
  (1,0) --
  (1.01,0);

\draw[
  fill=black
] 
(0,0) circle (.1);
\draw[
  fill=black
] 
(2,0) circle (.1);
\draw[
  fill=black
] 
(0,2) circle (.1);
(2,0) circle (.1);
\draw[
  fill=black
] 
(2,2) circle (.1);

\node at (2-.4,2+.5) {
  $\mathcolor{purple}{\zeta}^2\vert \psi \rangle$
};

\node[
  rotate=+135  
] at (2+.2, 2+.2) {%
  \clap{$\mapsto$}%
};

\node at (2.55,2-.25) {
  $\vert \psi \rangle$
};

\draw[
  line width=.5,
  |->,
]
  (2.55, 1.4) --
  node[
    xshift=0pt,
    scale=.9
  ] {\colorbox{white}{$
    W_b
      ^{-1}
  $}}
  (2.55,-.2);

\draw[
  line width=.5,
  |->,
]
  (2.4, -.47) --
  node[
    yshift=0pt,
    scale=.9
  ]{\colorbox{white}{$
    W_a^{-1}
  $}}
  (-.4, -.47);

\draw[
  line width=.5,
  |->,
]
  (-.55, -.2) --
  node[
    xshift=0pt,
    scale=.9
  ] {\colorbox{white}{$
    W_b
  $}}
  (-.55,2.3);

\draw[
  line width=.5,
  |->,
]
  (-.4, 2.45) --
  node[
    scale=.9,
  ]{\colorbox{white}{$
    W_a
  $}}
  (1, 2.45);
  
\end{tikzpicture}
\end{SCfigure}

This result (\cref{FundamentalGrouOfMapsFromT2ToS2}) further shows that, when regarded as a classifying space, the 2-sphere encodes fine detail of the characteristic topological order of FQH systems, at least on bounded domains (more on this in \cite{SS25-FQH}).

With this bulk result well established, next we ask how to generalize this to a classifying fibration over the 2-sphere such as to furthermore capture the behaviour of FQH systems on bounded domains.

\subsection
{Determining the Near Boundary Classifying Fibration}
\label
{OnDeterminingTheNearBdrClassfFib}

With $\ClassifyingB \defneq S^2$ thereby identified as the classifying space for unbounded FQH topological orders (\cref{OnRecallingUnboundedFQHOrders}) our task is to identify the appropriate classifying fibration $\inlinetikzcd{\ClassifyingA \ar[r, "{ \wp }"] \& S^2}$ for the topological properties of bounded FQH systems. 

\begin{SCfigure}[.5][htb] 
\caption{\label{SpatialDomainGraphics}
  Some examples of spatial domains
  \cref{TheDomainCofiber} considered here, cf. \cref{OnExamplesOfDomainSpaces}.
}
\centering
  \adjustbox{scale=0.9, 
    rndfbox=4pt
  }{
  \def\arraystretch{1.6}
  \begin{tabular}{rccc}
    closed disk
    &
    \adjustbox{
      raise=-.2cm
    }{
    \begin{tikzpicture}[scale=.7]
      \draw[
        line width=1,
        fill=gray!70
      ]
        (0,0)
        ellipse
        (2 and .5);
    \end{tikzpicture}
    }
    &
    $D^2$
    \\[5pt]
    closed annulus
    &
    \adjustbox{
      raise=-.4cm
    }{
    \begin{tikzpicture}[scale=.7]
      \draw[
        line width=1,
        fill=gray!70
      ]
        (0,0)
        ellipse
        (2 and .5);
      \draw[
        line width=1,
        fill=white
      ]
        (0,0)
        ellipse
        (2.4*.3 and .5*.3);
    \end{tikzpicture}
    }
    &
    $A^2$
    \\[5pt]
    constricted annulus
    &
    \adjustbox{
      raise=-.4cm
    }{
    \begin{tikzpicture}[scale=.7]
      \draw[
        line width=1,
        fill=gray!70
      ]
        (0,0)
        ellipse
        (2 and .5);
      \draw[
        line width=1,
        fill=white
      ]
        (.74,0)
        ellipse
        (2.4*.5 and .5*.32);
     \draw[fill=black]
       (2,0) circle (.1);
    \end{tikzpicture}
    }
    &
    $A^2_{\mathrm{cns}}$
    &
    (Def. \ref{ConstrictedAnnulus})
  \end{tabular}
  }
\end{SCfigure}
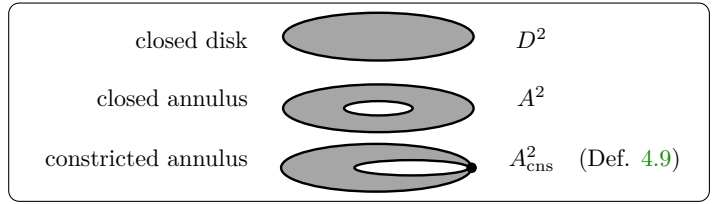

Concretely, to reflect that the total system over the disk $\BulkDomain := D^2 \sqcup \{\infty\}$ (cf. \cref{SpatialDomainGraphics}) is gapless (since it is so over the edge), we need to find $\inlinetikzcd{\ClassifyingA \ar[r, "{ \wp }"] \& S^2}$ such that \cref{GaplessBoundaryImpliesVanishingTotalCohomologyInNegativeDegree1}:
\begin{equation}
  H^{-1}({
    \iota_{D^2};
    \wp
  })
  = 
  0
  \mathrlap{\,.}
\end{equation}
Incidentally, this is \emph{not} the case for the na{\"i}ve choice $\wp := 0_{S^2}$, which instead makes the topological order of the total system be that of its deep bulk, by \cref{TotalHomotopyForTrivialA,ComplexHopfFibrationOnPi3}:
\begin{equation}
  H^{-1}\bracket({
    \iota_{D^2}
    ;
    0_{S^2}
  })
  \simeq
  \tilde H^{-1}\bracket({
    S^2;
    S^2
  })
  \simeq
  \mathbb{Z}
  \,.
\end{equation}

Instead, we find (\cref{MappingBoundaryInclusionOfDiskIntoComplexHopfFibration}) that a choice for $\wp$ which yields BBC as in \cref{TheConditionalBBC,TheSecondConditionalBBC} is the complex Hopf fibration \cref{TheComplexHopfFibration}
\begin{equation}
  \wp 
    := 
  h_{\mathbb{C}}
    :
  \begin{tikzcd}
    S^3 
      \ar[r, ->>] & 
    S^2
  \end{tikzcd}
\end{equation}
in that:
\begin{equation}
\label{PlainSituationOverTheClosedDisk}
  \left.
  \begin{aligned}
    \phi & = \iota_{D^2}
    \\
    \wp & = h_{\mathbb{C}}
  \end{aligned}
  \right\}
  \;\;
  \Rightarrow
  \;\;
  \begin{cases}
    H^{-1}\bracket({\partial D^2; S^3})
    =
    1
    & 
    \substack{
      \text{(no order near boundary)}
    }
    \\
    \inlinetikzcd{
    H^{-2}\bracket({
      \partial D^2; S^3
    })
    \ar[
      r, 
      "{ \Unwinding_1 }"{swap},
      "{ \sim }"{}
    ]
    \&
    \tilde H^{-1}\bracket({
      D^2/\partial D^2; S^2
    })
    }
    \simeq
    \mathbb{Z}
    &
    \substack{
      \text{(unwinding BBC)}
    }
    \\
    \inlinetikzcd{
    H^{-2}\bracket({
      D^2; S^2
    })
    \ar[
      r, 
      "{ \Rewinding_1 }"{swap},
      "{ \sim }"{}
    ]
    \&
    H^{-1}\bracket({
      \partial D^2; S^1
    })
    }
    \simeq
    \mathbb{Z}
    &
    \substack{
      \text{(rewinding BBC)}.
    }
  \end{cases}
\end{equation}

\begin{remark}
Concretely, the proof of \cref{MappingBoundaryInclusionOfDiskIntoComplexHopfFibration} shows that the complex Hopf fibration $h_{\mathbb{C}}$ works as a classifying fibration because: 
\begin{enumerate}
\item
the 2-connectivity of $S^3$ ensures that the boundary remains gapless, while
\item
its key property $\inlinetikzcd{\pi_3(S^3) \ar[r, "{ (h_{\mathbb{C}})_\ast }", "{ \sim }"{swap}] \& \pi_3(S^2) } $ \cref{ComplexHopfFibrationOnPi3} still allows it to unwind
the deep bulk phases in $\pi_1 \UnpointedMap(S^2, S^2) \simeq \pi_3(S^2)$.
\end{enumerate}
Conversely, this also means that under 3-truncation (\cite[(408)]{FSS23-Char}, which is all that matters for the computation, cf. \cite[Prop. A.43]{FSS23-Char}) the complex Hopf fibration  is the essentially unique classifying fibration with these properties. Of course, the 3-truncation of $h_{\mathbb{C}}$ is an infinite cell complex, so that the actual Hopf fibration, with its single cells, stands out as the minimal untruncated solution.
\end{remark}

\subsection
{Classifying Bounded Topological Orders}
\label{OnClassifyingBoundedTopologicalOrders}

With the appropriate classifying fibration identified as the complex Hopf fibration (\cref{OnDeterminingTheNearBdrClassfFib}) by comparison with the case over $\BulkDomain \defneq D^2$, we may now compute the implications for other bulk domains.

Most important in applications is the case where the material domain is the closed annulus $A^2$.
To prepare discussion for that, we first consider its \emph{constricted} version $A^2_{\mathrm{cns}}$  (\cref{ConstrictedAnnulus}, cf. \cref{SpatialDomainGraphics}):

\subsubsection
{Over the constricted annulus}
\label
{OnOverTheConstrictedAnnulus}

For the constricted annulus domain \cref{ConstrictedAnnulus}, we find the unwinding BBC \cref{SurjectiveUnwindingBBC},
to be surjective, of the following concrete form:
\begin{equation}
\label{PreviewResultOverA2Cns}
  \left.
  \begin{aligned}
    \phi & = \iota_{A^2_{\mathrm{cns}}}
    \\
    \wp & = h_{\mathbb{C}}
  \end{aligned}
  \right\}
  \;\;
  \Rightarrow
  \;\;
  \begin{cases}
    H^{-1}\bracket({
      \iota_{A^2_{\mathrm{cns}}}; 
      h_{\mathbb{C}}
    })
    =
    1
    & 
    \substack{
      \text{(no top. order} 
      \\
      \text{in tot. system)}
    }
    \\
    \begin{tikzcd}[
      row sep=4pt
    ]
    \tilde H^{-2}\bracket({
      \partial A^2_{\mathrm{cns}}; 
      S^3
    })
    \ar[
      d,
      equals
    ]
    \ar[
      r, 
      "{ \Unwinding_1 }"
    ]
    &
    \tilde H^{-1}\bracket({
      A^2_{\mathrm{cns}}
      /
      \partial A^2_{\mathrm{cns}}; 
      S^2
    })
    \ar[
      d,
      equals
    ]
    \\
    \mathbb{Z}^2
    \ar[
      r,
      "{
        \big(+1, -1\big)
      }"{description}
    ]
    &
    \mathbb{Z}
    \end{tikzcd}
    &
    \substack{
      \text{(unwinding BBC).}
    }
  \end{cases}
\end{equation}
In view of \cref{OnFQHOrderDeepInthePlane} this says that there is an algebra of observables, $\mathbb{C}[\mathbb{Z}]_{\pm}$, associated with each edge, with the expectation value of the observable $\widehat n_{\pm} \in \mathbb{Z} \subset \mathbb{C}\bracket[{\mathbb{Z}}]$ in the FQH state $\vert k \rangle$ being, by \cref{ExpectationValueOfAnyonsInPlane}:
\begin{equation}
  \label{EdgeCurrentObservables}
  \begin{tikzcd}
  \bracket\langle{ 
    \widehat n_{\pm} 
  }\rangle
  =
  e^{ \pm \tfrac{\pi i}{k} n_{\pm} }
  \mathrlap{\,,}
  \end{tikzcd}
\end{equation}
because
\begin{equation}
  \begin{tikzcd}[
    row sep=-3pt
  ]
    \mathbb{C}\bracket[{\mathbb{Z}}]
    \otimes_{_{\mathbb{C}}}
    \mathbb{C}\bracket[{\mathbb{Z}}]
    \ar[
      rr,
      "{
        \mathbb{C}\bracket[{
          \Unwinding_1
        }]
      }"
    ]
    &&
    \mathbb{C}\bracket[{\mathbb{Z}}]
    \ar[
      rr,
      "{
        \langle - \rangle
      }"
    ]
    &&
    \mathbb{C}
    \\
    \widehat{n}_+
    \otimes
    \widehat{n}_-
    \ar[
      rr,
      |->,
      shorten=11pt
    ]
    &&
    \widehat{n}_+ 
    - 
    \widehat{n}_-
    \ar[
      rr,
      |->,
      shorten=11pt
    ]
    &&
    \exp\bracket({
      \tfrac{\pi\mathrm{i}}{k}
      (n_+ - n_-)
    })
    \mathrlap{\,.}
  \end{tikzcd}
\end{equation}

But this formula \cref{EdgeCurrentObservables} expresses exactly the (time-independent, purely topological) observables on edge currents along (constricted) annuli (cf. \cite[p. 2335]{ChamonEtAl1997}), where $\pm \tfrac{\pi \mathrm{i}}{k} n_\pm$ is the (angular) \emph{Fermi momentum} of $n_{\pm}$ quasi-particles propagating along either edge.

This supports the idea that the choice of classifying fibration $\wp \defneq h_{\mathbb{C}}$ in \cref{OnDeterminingTheNearBdrClassfFib} correctly extends the \emph{Hopfion} model for bulk FQH anyons of \cref{OnRecallingUnboundedFQHOrders} to the situation with boundaries. We proceed to see what happens when we include the full bulk annulus into this picture:

\subsubsection
{Over the closed annulus}
\label
{OnOverTheClosedAnnulus}

The further computation for the actual closed annulus $A^2$ turns out to yield (\cref{MappingBoundaryInclusionOfClosedAnnulusIntoComplexHopfFibration}) an extension of the result for the constricted annulus $A^2_{\mathrm{cns}}$ (from \cref{OnOverTheConstrictedAnnulus}), induced by the canonical projection map $\inlinetikzcd{A^2\! \ar[r, ->>] \& \! A^2_{\mathrm{cns}}}$ \cref{TheProjectionFromClosedToConstrictedAnnulus}. Namely, the formalism of \cref{OnABNCFOrTopologicalOrders} yields that:
\begin{enumerate}
\item 
there is topological order now in the whole system, exhibited by a cyclic group's worth of monodromies acting on its ground states, 

\item the previous bulk/boundary correspondence map \cref{PreviewResultOverA2Cns} still holds away from this total system order, for its deep bulk order,
\end{enumerate}
in that (\cref{MappingBoundaryInclusionOfClosedAnnulusIntoComplexHopfFibration}):
\begin{equation}
  \label{PreviewResultOverA2}
  \left.
  \begin{aligned}
    \phi & = \iota_{A^2}
    \\
    \wp & = h_{\mathbb{C}}
  \end{aligned}
  \right\}
  \;\;
  \Rightarrow
  \;\;
  \begin{cases}
    H^{-1}\bracket({
      \iota_{A^2}; 
      h_{\mathbb{C}}
    })
    =
    \mathbb{Z}
    & 
    \substack{
      \text{(cyclic top. order} 
      \\
      \text{in tot. system)}
    }
    \\
    \begin{tikzcd}[
      ampersand replacement=\&,
      row sep=4pt
    ]
    \tilde H^{-2}\bracket({
      \partial A^2; 
      S^3
    })
    \ar[
      d,
      equals
    ]
    \ar[
      r, 
      "{ \Unwinding_1 }"
    ]
    \&
    \tilde H^{-1}\bracket({
      A^2
      /
      \partial A^2; 
      S^2
    })
    \ar[
      d,
      equals
    ]
    \\
    \mathbb{Z}^2
    \ar[
      r,
      "{ \scalebox{0.7}{$
        \left(
        \begin{matrix}
          + 1 & -1
          \\
          \,0 & \; 0
        \end{matrix}
        \right)
        $}
      }"{description}
    ]
    \&
    \mathbb{Z}^2
    \end{tikzcd}
    &
    \substack{
      \text{(partial unwinding BBC).}
    }
  \end{cases}
\end{equation}

Here the last line means that the $\mathbb{Z}^2$'s worth of boundary contributions (from the two edge modes) correspond only to one of two $\mathbb{Z}$ monodromy summands found in the deep bulk over the closed annulus (in a way that is isolated \cref{PreviewResultOverA2Cns} already over the constricted annulus), while one deep bulk copy of $\mathbb{Z}$ monodromy survives near the boundary.

Interestingly, this is \emph{not} the physically expected result: According to the above logic \cref{GaplessBoundaryImpliesVanishingTotalCohomologyInNegativeDegree1}, with the gap meant to be closed over the boundary, there should \emph{not} be nontrivial topological order \cref{BoundedOrderRepresentation} in the total system, it should only manifest nontrivially in the deep bulk. 

But we have not yet considered the \emph{differential refinement}, according to \cref{OnBndryFieldsInHGT}, of the classifying fibration, to account for the expected Chern-Simons type nature (cf. \cref{OnExcursionThroughGeometricHomotopy}) of the effective FQH fields in the bulk. This is what we turn to next in \cref{OnIdentifyingEdgeCurrentsInTEDCoh}, and we find in \cref{OnTheHomotopyOfThePhaseSpace} that this exactly rectifies the situation.

\subsection
{Identifying Edge Currents in TED Cohomotopy}
\label
{OnIdentifyingEdgeCurrentsInTEDCoh}

With the correct classifying fibration \cref{ClassifyingFibrationInIntroduction} for the topological nature of bulk/boundary FQH systems
plausibly identified (the complex Hopf fibration),
properly reproducing the expected phenomena over the closed disk
(\cref{OnDeterminingTheNearBdrClassfFib}) and over the constricted annulus (\cref{OnOverTheConstrictedAnnulus}), we turn to considering its differential refinement in the sense of \cref{OnBndryFieldsInHGT} such as to adjust for the one remaining discrepancy over the closed annulus found in \cref{OnOverTheClosedAnnulus}.

\subsubsection
{Identifying the global phase space}
\label
{OnIdentifyingThePhaseSpace}

Recall, from  \cref{ThePrimeMapInIntroduction,RelativeWhiteheadLAlgebra}, that a classifying fibration $\wp$ induces a fibration $\mathfrak{l}\wp$ of $L_\infty$-algebras (its \emph{relative Whitehead bracket $L_\infty$-algebra}), which encodes the Gauss laws of corresponding effective flux densities, as well as the local structure of their gauge potentials. For $\wp \defneq h_{\mathbb{C}}$ the complex Hopf fibration, the corresponding relations are shown in 
\cref{GaugePotentialForHopfFib}. 

The analogous analysis for the identity map $\wp' \defneq \mathrm{id}_{S^3}$ on the 3-sphere yields the same relations as in \cref{GaugePotentialForHopfFib}, except that the flux densities in degree 2 and 1 disappear, $F_2 = 0$ and $H_1 = 0$. But these are exactly the correct constraints expected on a Chern-Simons type bulk field. In detail:

\begin{remark}[A character image for abelian CS/WZW fields]
\label[remark]{IdentifyingCSFJUnderCharacterMap}
When the local gauge potentials are those induced by $h_\mathbb{C}$ while the flux densities are constrained to be induced by the image of $\inlinetikzcd{ \mathrm{id}_{S^3} \ar[r] \& h_{\mathbb{C}} }$, then, according to \cref{GaugePotentialForHopfFib}, the form of these local equations is that characteristic of:
\begin{enumerate}
  \item 
  an on-shell abelian Chern-Simons field $A_1$ in the bulk $\BulkDomain$ (cf. \cite{Dunne1998,nLab:AbelianChernSimons}),
  \[
    \mathrm{d}\, A_1 = 0
    \mathrlap{\,,}
  \]

  \item a corresponding Floreanini-Jackiw/Wess-Zumino-Witten current density $\CurrentDensity$ (cf. \parencites[(2.62)]{Wen1992}{nLab:FloreaniniJackiwTheory}) on the boundary $\inlinetikzcd{\BoundaryDomain \ar[r, hook, "{\phi}"] \& \BulkDomain}$:
  \[
    \mathrm{d}\lambda = 
    \phi^\ast A_1
    \mathrlap{\,.}
  \]
\end{enumerate}
This is hence exactly the effective on-shell field content of FQH liquids traditionally modeled by Lagrangian boundary abelian Chern-Simons theory, here reproduced by a non-Lagrangian construction (cf. \cite{SS25-ISQS29}).
\end{remark}

\begin{table}[htb]
\caption{
  \label{GaugePotentialForHopfFib}
  Under the twisted character map (\cref{OnTheCharacterMap}), classifying fibrations determine systems of Gau{\ss} laws on flux density species along domain embeddings $\phi : \BoundaryDomain \hookrightarrow \BulkDomain$, as well as the structure of local gauge potentials.  
  Shown is the case of the complex Hopf fibration $h_{\mathbb{C}}$, exhibiting besides a higher gauge field $B_2$, an abelian bulk gauge field $A_1$ and a boundary current $\CurrentDensity$ (\parencites[\S B]{SS25-Srni}[Rem. 3.20]{Banerjee2025-Potentials}, building on \parencites[\S 4.1]{GSS25-M5}). The analogous formulas for the identity fibration $\mathrm{id}_{S^3}$ are obtained by setting $F_2 = 0$ (which is the Chern-Simons equation of motion for $A_1$) and $H_1 = 0$ (which makes $\CurrentDensity$ be the corresponding FJ/WZW boundary current), cf. \cref{IdentifyingCSFJUnderCharacterMap}. 
}
\centering
\adjustbox{scale=0.95, 
 rndfbox=4pt
}{
  \begin{tikzcd}[
    ampersand replacement=\&,
    column sep=4pt
  ]
    \mbox{\textbf{\footnotesize
      \begin{tblr}{
        colspec=c,
        rowsep=0pt,
      } 
        Classifying
        \\
        fibrations
      \end{tblr}
    }}
    \&
    \mbox{\textbf{\footnotesize
    \begin{tblr}{
      colspec=c,
      rowsep=0pt,
    }
      Relative Sul- 
      \\
      livan models
    \end{tblr}
    }}
    \&
    \mbox{\textbf{\footnotesize
    \begin{tblr}{
      colspec={c},
      rowsep=0pt,
    }
      Gau{\ss} laws on
      \\
      flux densities
    \end{tblr}
    }}
    \&
    \mbox{\textbf{\footnotesize
    \begin{tblr}{
      colspec={c},
      rowsep=0pt,
    }
      Local gauge
      \\
      potentials
    \end{tblr}
    }}
    \\[-20pt]
    S^3
    \ar[
      d,
      "{ 
        h_{\mathbb{C}} 
      }"{description, pos=.45}
    ]
    \&
    \begin{tblr}{
      colspec={rl},
      rowsep=0pt,
      colsep=2pt
    }
      \mathrm{d}\, h_1 
      &
      \mathcolor{black!75}{
        = f_2
      }
    \end{tblr}
    \&
    \begin{tblr}{
      colspec={rl},
      rowsep=0pt,
      colsep=2pt
    }
      \mathrm{d}\, H_1 
      &
      = \phi^\ast F_2
    \end{tblr}
    \&
    \begin{tblr}{
      colspec={rl},
      rowsep=0pt,
      colsep=2pt
    }
      \mathrm{d}\,
      \CurrentDensity
      &
      =
      \phi^\ast A_1
         - 
       H_1 
    \end{tblr}
    \\
    S^2
    \&
    \begin{tblr}{
      colspec={rl},
      rowsep=0pt,
      colsep=2pt
    }
       \mathrm{d}\, h_3 &= f_2 f_2
       \\
       \mathrm{d}\, f_2 & = 0
    \end{tblr}
    \&
    \begin{tblr}{
      colspec={rl},
      rowsep=0pt,
      colsep=2pt
    }
       \mathrm{d}\, H_3 
         &= F_2 \wedge F_2
       \\
       \mathrm{d}\, F_2 
         & = 0
    \end{tblr}
    \&
    \begin{tblr}{
      colspec={rl},
      rowsep=0pt,
      colsep=2pt
    }
      \mathrm{d}\, B_2 
        & =  H_3 - A_1 \wedge F_2
      \\
      \mathrm{d}\, A_1 & = F_2
    \end{tblr}
  \end{tikzcd}
}

\end{table}

In view of the general discussion in \cref{OnBndryFieldsInHGT}, this suggests that the correct phase space object \cref{ThePhaseSpaceInIntro} to consider is that induced by the fibration of fibrations
\begin{equation}
\label
{TheRefinementFibrationForFQH}
  \big(
  \inlinetikzcd{
    \wp'
    \ar[r, "{ \prime }"]
    \&
    \wp
  }
  \big)
  :=
  \Big(
  \begin{tikzcd}
    \mathrm{id}_{S^3}
    \ar[
      r,
      "{ 
        (
          \mathrm{id}, 
          h_{\mathbb{C}}
        ) 
      }"
    ]
    &
    h_{\mathbb{C}}
  \end{tikzcd}
  \Big)
  ,
  \;\;\;\;
  \begin{tikzcd}[column sep=large]
    S^3
    \ar[
      d,
      "{ \mathrm{id}_{S^3} }"
    ]
    \ar[
      r,
      "{
        \mathrm{id}_{S^3}
      }"
    ]
    &
    S^3
    \ar[
      d,
      "{ h_{\mathbb{C}} }"
    ]
    \\
    S^3
    \ar[
      r,
      "{ h_{\mathbb{C}} }"
    ]
    &
    S^2
    \mathrlap{\,,}
  \end{tikzcd}
\end{equation}
hence:
\begin{definition}
\label[definition]
{FQHPhaseSpace}
As the cohomotopically flux-quantized phase space,
$\FQHPhaseSpace(\phi)$,
for the global effective description of bounded FQH liquids on $\inlinetikzcd{ \BoundaryDomain \ar[r, hook, "{ \phi }"] \& \BulkDomain }$ we take the following homotopy pullback in $\mathrm{SmthGrpd}_\infty$:
\begin{equation}
\label
{TheFQHPhaseSpace}
  \begin{tikzcd}[
    ampersand replacement=\&,
    column sep=35pt,
    row sep=0pt
  ]
    \FQHPhaseSpace(\phi)
    \ar[
      rr,
      "{\ }"{name=s, swap}
    ]
    \ar[
      d,
      "{\ }"{name=t}
    ]
    \ar[
      drr,
      phantom,
      "{ \lrcorner_h }"{pos=0}
    ]
    \ar[
      from=s, 
      to=t,
      Rightarrow,
      "{  
        \begin{subarray}{l}
          \mathrm{d}\CurrentDensity 
          = 
          \phi^\ast A_1
          \\
          \\
          \mathrm{d}A_1 
          = 0
          \\
          \mathrm{d}B_2 
          = H_3
        \end{subarray}
      }"{color=gray}
    ]
    \&[20pt]
    \&
    \UnpointedMap\bracket({
      \phi, h_{\mathbb{C}}
    })
    \ar[
      d,
      "{
        \mathbf{ch}^{h_{\mathbb{C}}}
      }"
    ]
    \\[+50pt]
    \mathbf{\Omega}^1_{\mathrm{cl}}\bracket({
      \phi; 
      \mathfrak{l}
      \mathrm{id}_{S^3}
    })
    \ar[
      r,
      "{
        \mathfrak{l}(
          \mathrm{id}_{S^3},
          h_{\mathbb{C}}
        )\ast
      }"
    ]
    \&
    \mathbf{\Omega}^1_{\mathrm{cl}}\bracket({
      \phi; 
      \mathfrak{l}
      h_{\mathbb{C}}
    })
    \ar[
      r,
      "{
        \eta^{\shape}
      }"
    ]
    \&
    \shape
    \mathbf{\Omega}^1_{\mathrm{cl}}\bracket({
      \phi; \mathfrak{l}h_{\mathbb{C}}
    }) 
    \mathrlap{\,,}
    \\
    \mathcolor{gray}{\substack{
      H_1  = 0
      \\
      \\
      \mathrm{d}\, H_3 = 0
      \\
      F_2 = 0
    }}
    \&
    \mathcolor{gray}{\substack{
      \mathrm{d}\, H_1 = \phi^\ast F_2
      \\
      \\
      \mathrm{d}\, H_3 = F_2 \wedge F_2
      \\
      \mathrm{d}\, F_2 = 0
    }}
  \end{tikzcd}
\end{equation}
where in gray we are indicating (using \cref{IdentifyingCSFJUnderCharacterMap,GaugePotentialForHopfFib}) which parts of this diagram encode the flux densities and their equations of motions (namely the bottom left corner, constraining the bottom middle object), and which encode the gauge potential relations (namely the universal homotopy filling this diagram).
\end{definition}

\begin{remark}
  As highlighted, various structures appearing here are familiar from traditional Lagrangian abelian Chern-Simons theory and its boundary FL/WZW field theory (\cref{IdentifyingCSFJUnderCharacterMap}). And yet the construction is fundamentally different (non-Lagrangian, immediately globally defined, intrinsically flux quantized in non-abelian cohomology) and answers to questions that are at least not readily addressed with Lagrangian field theory methods.
\end{remark}

We proceed to compute the homotopy groups, in low degree, of the shape \cref{PrimedPhaseSpaceInIntroduction} 
\begin{equation}
  \shape \FQHPhaseSpace(\phi)
  =
  \UnpointedMap\bracket({
    \phi; h_{\mathbb{C}}
  })
  \;\;\;\,
  \underset{\mathclap{
    \UnpointedMap\scaledbracket({
      \phi; 
      \RRationalization h_{\mathbb{C}}
    })
  }}
    {\times}
  \;\;\;\,
  \UnpointedMap\bracket({
    \phi;
    \RRationalization \mathrm{id}_{S^3}
  })
\end{equation}
of this phase space \cref{TheFQHPhaseSpace}.

\subsubsection
{Homotopy of the phase space}
\label{OnTheHomotopyOfThePhaseSpace}

We find indeed that with the differential refinement \cref{TheFQHPhaseSpace} of the relative mapping space, the previous conclusions over the closed disk and over the constricted annulus remain unchanged, while over the closed annulus exactly the one sticking point \cref{PreviewResultOverA2} gets rectified:
\begin{enumerate}

\item
(\cref{PhaseSpaceHomotopyOverClosedDisk})
For $\phi \defneq \iota_{D^2}$
we find that still
$
  \pi_1\bracket({
    \shape\FQHPhaseSpace(\iota_{D^2})
  })
  \simeq
  1
  \,.
$ 
and that the unwinding BBC remains an isomorphism, as was the case \cref{PlainSituationOverTheClosedDisk} before refinement.

\item (\cref{PhaseSpaceHomotopyOverAnnulus})
For $\phi \defneq \iota_{A^2}$
we find  that the total system is now unordered as it should be,
$
  \pi_1\bracket({
    \shape\FQHPhaseSpace(\iota_{A^2})
  })
  \simeq
  1
  \,.
$ 
with a bulk boundary correspondence reflecting that bulk order unwinds with one sign on one boundary and with the other sign on the other boundary, as was the case without refinement \cref{PreviewResultOverA2Cns} over the constricted annulus.

\end{enumerate}

We give detailed proofs of these results in \cref{Proofs} and put them in perspective in \cref{OnGeometricEngineeringOnMBraneProbes,Conclusions}.

\section
{Proofs}
\label{Proofs}

Here we give precise definitions and detailed proofs for the statements discussed above in \cref{OnMethods} and \cref{OnResults}.

\subsection
{The General Bulk-Boundary Correspondence}

First we establish in general the fiber sequences \cref{HomotopyFiberSequencesOfModuli} and the induced homotopy exact sequence \cref{TheUnwindingLES}.

\begin{definition}
\label[definition]{RelativeMappingSpace}
Given 

\begin{enumerate}

\item
a pointed cell complex inclusion $\inlinetikzcd{ N \ar[r, hook, "{\phi}"] \& \BulkDomain }$ with quotient $\DeepBulkDomain$ \cref{Cofiber},
\begin{equation}
  \label{TheDomainCofiber}
  \begin{tikzcd}[column sep=huge]
    N
    \ar[
      d,
      hook,
      "{ \phi }"{swap},
      "{ \in \mathrm{Cof} }"
    ]
    \ar[
      r
    ]
    \ar[
      dr,
      phantom,
      "{ \ulcorner }"{pos=.9}
    ]
    &
    \ast
    \ar[d]
    \\
    \BulkDomain
    \ar[r, "{ q }"]
    &
    \DeepBulkDomain
    \mathrlap{\,,}
  \end{tikzcd}
\end{equation}

\item a pointed Serre fibration $\inlinetikzcd{ \ClassifyingA \ar[r, "{ \wp }", "{ \in \mathrm{Fib} }"{swap}] \& \ClassifyingB }$ (cf. \cref{SerreFibration}) with fiber $\ClassifyingF$ \cref{Fiber},
\begin{equation}
  \label{TheCodomainFiber}
  \begin{tikzcd}[column sep=large]
    \ClassifyingF 
    \ar[r, hook, "{ i }"]
    \ar[d]
    \ar[
      dr,
      phantom,
      "{ \lrcorner }"{pos=.1}
    ]
    &
    \ClassifyingA
    \ar[
      d,
      "{ \wp }"{swap},
      "{ \mathrlap{\in \mathrm{Fib}} }"
    ]
    \\
    \ast
    \ar[r]
    &
    \ClassifyingB
    \mathrlap{\,,}
  \end{tikzcd}
\end{equation}
\end{enumerate}
we say that the \emph{$\wp$-twisted $\phi$-relative mapping space} (\parencites[Fig. 2]{SS25-Orient}[(68)]{BaSS26-MString}) is the fiber product space
\begin{equation}
  \label{TheRelativeMappingSpace}
  \Map({\phi,\wp})
  :=
  \Map({N,\ClassifyingA})
  \underset{\Map({N,\ClassifyingB})}{\times}
  \Map({\BulkDomain, \ClassifyingB})
\end{equation}
of compatible pairs of maps shown as dashed arrows in the following commuting diagram:
\begin{equation}
  \Map\bracket({
    \phi, \wp
  })
  =
  \left\{\,
  \begin{tikzcd}[column sep=large]
    N
    \ar[
      d,
      hook,
      "{ \phi }"{swap},
      "{ \in \mathrm{Cof} }"
    ]
    \ar[
      r,
      dashed
    ]
    &
    \ClassifyingA
    \ar[
      d,
      "{ \wp }"{swap},
      "{ \in \mathrm{Fib} }"
    ]
    \\
    \BulkDomain
    \ar[
      r,
      dashed
    ]
    &
    \ClassifyingB
    \mathrlap{\,.}
  \end{tikzcd}
  \, \right\}
  \mathrlap{.}
\end{equation}
\end{definition}
\begin{example}
  For $\ClassifyingA = \ClassifyingB$ and $\wp$ the identity, also $\Map(\BoundaryDomain,\wp)$ is an identity, so that \cref{TheRelativeMappingSpace} reduces to
  \begin{equation}
    \label{TwRelMapIntoIdentity}
    \Map\bracket({
      \phi, \mathrm{id}_{\ClassifyingB}
    })
    \simeq
    \Map\bracket({
      \BulkDomain,
      \ClassifyingB
    })
    \mathrlap{\,,}
    \hspace{.6cm}
    \begin{tikzcd}[row sep=12pt, column sep=large]
     \BoundaryDomain
     \ar[d, hook, "{ \phi }"]
     \ar[
       r,
       dashed,
       "{ \exists ! }"
     ]
     &
     \ClassifyingB
     \ar[d, equals]
     \\
     \BulkDomain
     \ar[r, "{ \forall }"]
     &
     \ClassifyingB
     \mathrlap{\,.}
    \end{tikzcd}
  \end{equation}
\end{example}
\begin{lemma}
In the situation of \textup{\cref{RelativeMappingSpace}}, we have Cartesian squares \cref{CartesianSquare} as follows:
\begin{equation}
  \label{ThePastingDiagram}
  \begin{tikzcd}[column sep=huge]
    &
    \Map({
      \DeepBulkDomain, 
      \ClassifyingB
    })
    \ar[r]
    \ar[
      d,
      "{ \BulkInclusion }"{swap}
    ]
    \ar[
      dr,
      phantom,
      "{ \lrcorner }"{pos=.05}
    ]
    &
    \ast
    \ar[d]
    \\
    \Map\bracket({
      \BoundaryDomain, 
      \ClassifyingF
    })
    \ar[
      dr,
      phantom,
      "{ \lrcorner }"{pos=.1}
    ]
    \ar[
      r,
      "{ \BoundaryInclusion }"
    ]
    \ar[d]
    &
    \Map({\phi,\wp})
    \ar[
      r,
      "{ \BoundaryRestriction }",
      "{ \in \mathrm{Fib} }"{swap}
    ]
    \ar[
      d,
      "{ \BulkRestriction }"{swap},
      "{
        \in \mathrm{Fib}
      }"
    ]
    \ar[
      dr,
      phantom,
      "{ \lrcorner }"{pos=.05}
    ]
    &
    \Map({N,\ClassifyingA})
    \ar[
      d,
      "{ \wp_\ast }"{swap},
      "{
        \in \mathrm{Fib}
      }"
    ]
    \\
    \ast
    \ar[r]
    &
    \Map({\BulkDomain,\ClassifyingB})
    \ar[
      r,
      "{ \phi^\ast }",
      "{ \in \mathrm{Fib} }"{swap}
    ]
    &
    \Map({N,\ClassifyingB})
    \mathrlap{\,.}
  \end{tikzcd}
\end{equation}
\end{lemma}
\begin{proof}
The bottom right square is Cartesian by \cref{RelativeMappingSpace}, expressing \cref{TheRelativeMappingSpace}.
But also the right and bottom total rectangles are Cartesian, since they are the images of the domain cofiber sequence \cref{TheDomainCofiber} under $\Map\bracket({-,\ClassifyingB})$, and of the codomain fiber sequence \cref{TheCodomainFiber} under $\Map\bracket({\BoundaryDomain,-})$, respectively,  
and since $\inlinetikzcd{\Map(-,-) : \bracket({\mathrm{Top}^{\ast}})^{\mathrm{op}} \times \mathrm{Top}^{\ast} \ar[r] \& \mathrm{Top}^{\ast} }$ preserves limits in each argument (\cref{MapPreservesLimits}).
Therefore, the pasting law \cref{PastingLaw} implies that also the top and the left squares are Cartesian.
\end{proof}

\begin{proposition}
\label[proposition]
{TheLongExactSequences}
In the situation of \textup{\cref{RelativeMappingSpace}}, we have long exact sequences of homotopy groups, as follows:
\begin{enumerate}
\item
The homotopy groups of the relative mapping space \eqref{TheRelativeMappingSpace} sit in a long exact sequence of this form:
\begin{equation}
  \label{HomotopyLESForRelativeMap}
  \begin{tikzcd}[column sep=38pt]
    &
    \;\;\;\cdots\;\;\;
    \ar[r]
    &
    \pi_2\,
    \Map({N,\ClassifyingA})
    \ar[
      dll,
      snake left,
      "{
        \Unwinding_1
      }"{description}
    ]
    \\
    \pi_1\, \Map({\DeepBulkDomain,\ClassifyingB})
    \ar[
      r,
      "{
        \pi_1(\BulkInclusion)
      }"
    ]
    &
    \pi_1\,
    \Map({\phi,\wp})
    \ar[
      r,
      "{
        \pi_1(\BoundaryRestriction)
      }"
    ]
    &
    \pi_1\,
    \Map({N,\ClassifyingA})
    \ar[
      dll,
      snake left,
      "{ 
        \Unwinding_0
      }"{description}
    ]
    \\
    \pi_0\, \Map({\DeepBulkDomain,\ClassifyingB})
    \ar[
      r,
      "{ \pi_0(\BulkInclusion) }"
    ]
    &
    \pi_0\,
    \Map({\phi,\wp})
    \ar[
      r,
      "{ \pi_0(\BoundaryRestriction) }"
    ]
    &
    \pi_0\,
    \Map({N,\ClassifyingA})
    \mathrlap{\,,}
  \end{tikzcd}
\end{equation}
and this has a morphism of sequences to the homotopy LES induced by 
$\Map\big({\inlinetikzcd{\DeepBulkDomain \ar[r, <-, "{ q } "] \& \BulkDomain \ar[r, <-, "{\phi}"] \& \BoundaryDomain}, \ClassifyingB}\big)$ \textup{(\cref{MappingFiberSequenceOfCWInclusion})}, of this form:
\begin{equation}
  \label{MapFromHomotopyLESOfRelMap}
  \begin{tikzcd}[
    column sep=30pt
  ]
    \pi_{n+1}
    \Map\bracket({
      N, 
      \ClassifyingA
    })
    \ar[
      r,
      "{ \Unwinding_{n} }"
    ]
    \ar[
      d,
      "{ 
        \pi_{n+1}(\wp_\ast) 
      }"
    ]
    &
    \pi_n \Map\bracket({
      \DeepBulkDomain, 
      \ClassifyingB
    })
    \ar[
      d,
      equals
    ]
    \ar[
      r,
      "{
        \pi_n(\BulkInclusion)
      }"
    ]
    &
    \pi_n \Map\bracket({
      \phi, \wp
    })
    \ar[
      d,
      "{
        \pi_n(\BulkRestriction)
      }"
    ]
    \ar[
      r,
      "{
        \pi_n(\BoundaryRestriction)
      }"
    ]
    &
    \pi_n \Map\bracket({
      N, \ClassifyingA
    })
    \ar[
      d,
      "{ 
        \pi_n(\wp_\ast) 
      }"
    ]
    \\
    \pi_{n+1}
    \Map\bracket({
      N,
      \ClassifyingB
    })
    \ar[
      r,
      "{ 
        \partial^{\phi^\ast}_n 
      }"
    ]
    &
    \pi_n\Map\bracket({
      \DeepBulkDomain, \ClassifyingB
    })
    \ar[
      r,
      "{
        \pi_n(q^\ast)
      }"
    ]
    &
    \pi_n\Map\bracket({
      \BulkDomain, \ClassifyingB
    })
    \ar[
      r,
      "{
        \pi_n(\phi^\ast)
      }"
    ]
    &
    \pi_n\Map\bracket({
      N, \ClassifyingB
    })
    \mathrlap{\,.}
  \end{tikzcd}
\end{equation}

\item
The homotopy groups of the relative mapping space also sit in a long exact sequence of this form:
\begin{equation}
  \label{SecondHomotopyLESForRelativeMap}
  \begin{tikzcd}[row sep=small]
    &
    \;\;\cdots\;\;
    \ar[r]
    &
    \pi_2\,
    \Map\bracket({
      \BulkDomain, \ClassifyingB
    })
    \ar[
      dll,
      snake left
    ]
    \\
    \pi_1\,
    \Map\bracket({
      \BoundaryDomain, 
      \ClassifyingF
    })
    \ar[r]
    &
    \pi_1\,
    \Map\bracket({
      \phi, \wp
    })
    \ar[r]
    &
    \pi_1\,
    \Map\bracket({
      \BulkDomain, \ClassifyingB
    })
    \ar[
      dll,
      snake left
    ]
    \\
    \pi_0\,
    \Map\bracket({
      \BoundaryDomain, 
      \ClassifyingF
    })
    \ar[r]
    &
    \pi_0\,
    \Map\bracket({
      \phi, \wp
    })
    \ar[r]
    &
    \pi_0\,
    \Map\bracket({
      \BulkDomain, \ClassifyingB
    })\,,
  \end{tikzcd}
\end{equation}
and there is a morphism of sequences
\begin{equation}
  \label{MapFromSecondHomotopyLESOfRelMap}
  \begin{tikzcd}
    \pi_{n+1}\, \Map\bracket({
      \BulkDomain, \ClassifyingB
    })
    \ar[
      d,
      "{ \phi^\ast }"
    ]
    \ar[
      r,
      "{ \partial_n }"
    ]
    &
    \pi_n\, \Map\bracket({
      \BoundaryDomain, 
      \ClassifyingF
    })
    \ar[
      d,
      equals
    ]
    \ar[r]
    &
    \pi_n\, \Map\bracket({
      \phi, \wp
    })
    \ar[d]
    \ar[r]
    &
    \pi_n\, \Map\bracket({
      \BulkDomain, \ClassifyingB
    })
    \ar[
      d,
      "{ \phi^\ast }"
    ]
    \\
    \pi_{n+1}\,
    \Map\bracket({
      \BoundaryDomain, 
      \ClassifyingB
    })
    \ar[r, "{ \partial_n }"]
    &
    \pi_n\, \Map\bracket({
      \BoundaryDomain, 
      \ClassifyingF
    }) 
    \ar[r]
    &
    \pi_n\, \Map\bracket({
      \BoundaryDomain, 
      \ClassifyingA
    }) 
    \ar[r]
    &
    \pi_n\,
    \Map\bracket({
      N, \ClassifyingB
    })
    \mathrlap{\,.}
  \end{tikzcd}
\end{equation}
\end{enumerate}
\end{proposition}
\begin{proof}
The claimed long exact sequence \cref{HomotopyLESForRelativeMap} is the homotopy long exact sequence \cref{HomotopyLES} induced by the homotopy fiber sequence exhibited by the top square in \cref{ThePastingDiagram}, and the morphism of sequences \cref{MapFromHomotopyLESOfRelMap} is that induced via naturality \cref{NaturalityOfHomotopyLES} by the inclusion of the top square into the right rectangle of \cref{ThePastingDiagram}.

In the same way, the claimed long exact sequence \cref{SecondHomotopyLESForRelativeMap} is the homotopy long exact sequence \cref{TheHomotopyLES} induced by the fiber sequence exhibited by the left square in \cref{ThePastingDiagram}, and the morphism of sequences \cref{MapFromSecondHomotopyLESOfRelMap} is that induced via naturality \cref{NaturalityOfHomotopyLES} by the inclusion of the left square into the bottom rectangle in \cref{ThePastingDiagram}.
\end{proof}

\begin{example}
  For $\ClassifyingA \simeq \ast$ we have
  \begin{equation}
    \label{TotalHomotopyForTrivialA}
    \pi_n
    \Map\bracket({
      \phi, \wp
    })
    \simeq
    \pi_n
    \Map\bracket({
      \DeepBulkDomain,
      \ClassifyingB
    })
    \mathrlap{\,.}
  \end{equation}
\end{example}
\begin{proof}
  The assumption immediately implies that $\Map\bracket({-,\ClassifyingA}) \simeq \ast$, whence the claim follows by the exactness of \cref{MapFromHomotopyLESOfRelMap}.
\end{proof}

\subsection
{The Correspondence over particular Domains}

Here we compute homotopy groups of $\phi$-related $\wp$-twisted mapping spaces \cref{TheRelativeMappingSpace} for specific choices of domain inclusions $\phi$ and classifying fibrations $\wp$.

\subsubsection
{Examples of domain spaces}
\label{OnExamplesOfDomainSpaces}

Before entering the actual computations in \cref{ExamplesOfTwistedRelativeCohomology}, here we record some definitions and facts concerning the spatial domains \cref{TheDomainCofiber} involved, cf. \cref{SpatialDomainGraphics}.

\begin{lemma}
  The boundary inclusions of both the closed annulus and the constricted annulus \cref{TheConstrictedAnnulus} are homotopy equivalent to codiagonals, whence in particular their induced homomorphisms on homotopy groups of mapping spaces into a fixed classifying space are diagonal maps $\Delta$:
  \begin{equation}
    \label{AnnulusBoundaryIncIsCodiagonal}
    \begin{aligned}
    \pi_n
    \Map\bracket({
      \iota_{A^2_{\mathrm{cns}}},
      \ClassifyingB
    })
    &
    :
    \inlinetikzcd{
      \pi_n\Map\bracket({
        S^1,
        \ClassifyingB
      })
      \ar[r, "{ \Delta }"]
      \&
      \pi_n\Map\bracket({
        S^1,
        \ClassifyingB
      })^2
    }
    \\
    \pi_n
    \UnpointedMap\bracket({
      \iota_{A^2},
      \ClassifyingB
    })
    &
    :
    \inlinetikzcd{
      \pi_n\UnpointedMap\bracket({
        S^1,
        \ClassifyingB
      })
      \ar[r, "{ \Delta }"]
      \&
      \pi_n\UnpointedMap\bracket({
        S^1,
        \ClassifyingB
      })^2
      \mathrlap{\,.}
    }
    \end{aligned}
  \end{equation}
\end{lemma}
\begin{proof}
  Since the (constricted) annulus evidently (pointed) deformation retracts onto any one of its boundary circles. 
\end{proof}
\begin{lemma}
  The quotient coprojection $q$ \cref{TheDomainCofiber} of the (constricted) annulus is null homotopic, whence in particular the induced homomorphisms on homotopy groups of mapping spaces are null:
  \begin{equation}
    \label{AnnulusCoprojectionIsNull}
    \begin{aligned}
      \pi_n 
      \bracket({
        q^\ast_{\smash{A^2_{\mathrm{cns}}}}
      })
      & = 0
      \mathrlap{\,,}
      \\
      \pi_n 
      \bracket({
        q^\ast_{A^2}
      })
      & = 0
      \mathrlap{\,.}
    \end{aligned}
  \end{equation}
\end{lemma}
\begin{proof}
  Since the (constricted) annulus deformation retracts onto any one of its boundary circles, which are identified with the basepoint of the quotient, by construction.
\end{proof}

\begin{SCfigure}[.8][htb]
\caption{\label{SphereWithPolesIdentified}
The quotient $A^2/\partial A^2$ of the closed annulus by its boundary is
homotopy equivalent to the sphere with an arc, $a_{\mathrm{ext}}$, attached to a pair of antipodal points (top map). But this, in turn, is also homotopy-equivalent to the result of contracting an arc, $a_{\mathrm{int}}$, connecting these two points inside the sphere (bottom map): This yields the wedge sum of the sphere with a circle (\cref{QuotientOfClsdAnnulusByBndry}).
}

\adjustbox{
  raise=-2cm,
  scale=1.13
}{
\begin{tikzpicture}
  \node at (0,0) {
  \includegraphics[width=5cm]{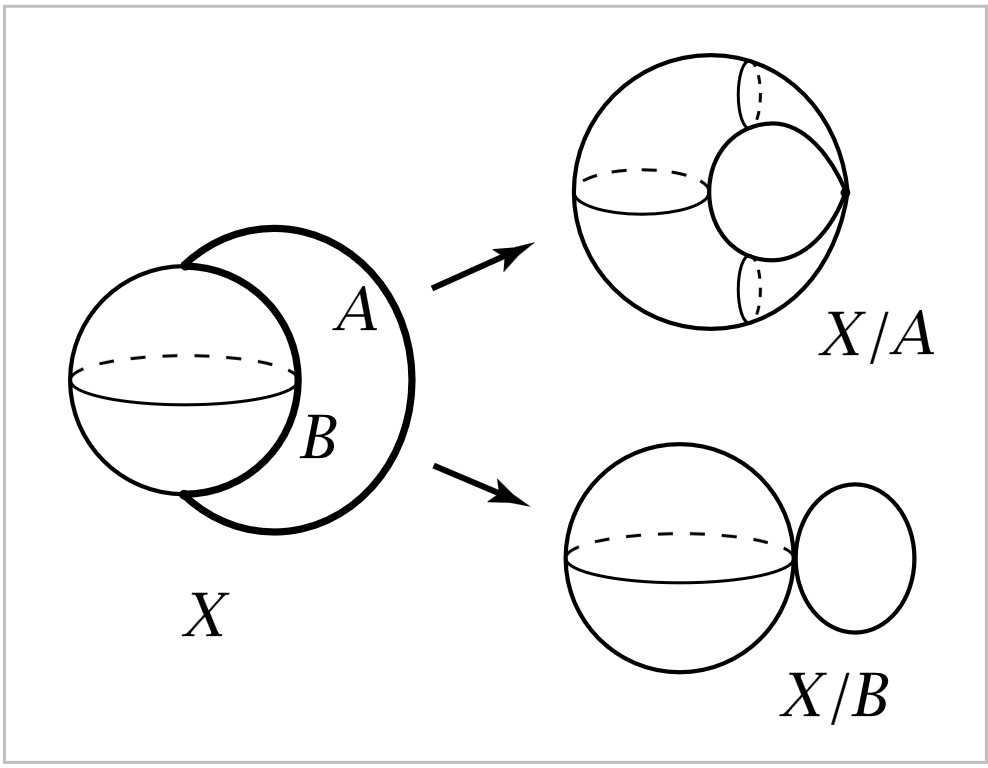}
  };
  \node[
    scale=.5,
    rotate=-40
  ] at (-1.87,.86)
  {
    \color{purple}
    \hspace{-20pt}bndry cmpnt
  };
  \begin{scope}[
    yscale=-1
  ]
  \node[
    scale=.5,
    rotate=+40
  ] at (-1.85,.85)
  { \color{purple}
    \hspace{-17pt}bndry cmpnt
  };
  \end{scope}
  \node[
    scale=.8
  ] 
    at (-2.15,-.53)
  {
  };
  \node[
    scale=.8
  ] 
    at (.9,.08)
  {$A^2/\partial A^2\;
  \raisebox{3pt}{\rotatebox[origin=c]{10}{$\simeq$}}$};
  \node[scale=.8] at (.7,-1.7)
  {
    $S^2 \vee S^1
   \;
  \adjustbox{rotate=10,raise=2pt}{$\simeq$}\,$
  };
  \node[
    rotate=-20,
    scale=.8
  ]
    at (-.2,-.6)
    {$\sim$};
  \node[
    rotate=+20,
    scale=.8
  ]
    at (-.18,+.66)
    {$\sim$};
  \node[
    scale=.7
  ]
  at (1.98,.97) 
  {$\infty$};
  \begin{scope}[
    shift={(-1.47,.57)}
  ]
  \draw[purple]
    (180-5:.2 and .1) arc
    (180-5:360+5:.11 and .05);
  \end{scope}
  \begin{scope}[
    yscale=-1,
    shift={(-1.47,.52)},
  ]
  \draw[purple]
    (180-5:.2 and .1) arc
    (180-5:360+5:.11 and .05);
  \end{scope}

  \draw[
    draw=white,
    fill=white
  ] 
    (-.7,.35) circle
    (.16);
  \node[scale=.7] at (-.7,.35) 
    {$a_{\mathrm{ext}}$};

  \draw[
    draw=white,
    fill=white  
  ]
    (-.86,.-.27) circle
    (.16);
  \node[scale=.7] at 
    (-.84,.-.27)
    {$a_{\mathrm{int}}$};

  \draw[
    draw=white,
    fill=white
  ]
    (2.15,.22) circle
    (.16);

  \node[scale=.7] at (2.2,.2)
    {$a_{\mathrm{ext}}$};

  \draw[
    draw=white,
    fill=white
  ]
    (1.95,-1.58) circle
    (.16);
  \node[scale=.7] 
    at (1.99,-1.64) 
     {$a_{\mathrm{int}}$};

  \node[scale=.8] at (-.55,-1.15) {
    $=\! S^2 \cup a_{\mathrm{ext}}$
  };

\end{tikzpicture}
  }
  \hspace{-10pt}
  \adjustbox{
    rotate=-90,
    scale=.7
  }{
    \color{gray}
    \clap{
    graphics adapted from 
    \cite[p 11]{Hatcher2002}
    }
  }

\end{SCfigure}
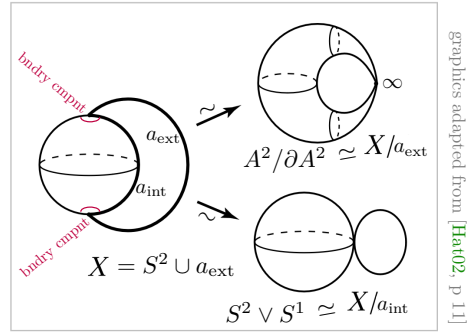

\begin{lemma}[Deep bulk domain of closed annulus]
  \label[lemma]{QuotientOfClsdAnnulusByBndry}
  The quotient of the boundary inclusion of the closed annulus \textup{(\cref{SpatialDomainGraphics})} is homotopy-equivalent to the wedge sum of the 2-sphere with a circle:
\begin{equation}
  \label{DeepBulkOfClosedAnnulus}
  \begin{aligned}
  A^2/\partial A^2
  & \simeq
  A^2/\bracket({
    S^1 \sqcup S^1
  })
  \simeq
  S^2 / S^0
  \\
  & \sim
  S^2 \vee S^1
  \,.
  \end{aligned}
\end{equation}
\end{lemma}
\begin{proof}
  This follows, as indicated in \cref{SphereWithPolesIdentified}, using that the quotient of a CW-complex by a contractible subcomplex, such as $a_{\mathrm{int}}, a_{\mathrm{ext}} \subset S^2 \cup a_{\mathrm{ext}}$, is a homotopy equivalence (by {\parencites[Ex. 0.8]{Hatcher2002}, cf. 
  \cite[Lem. 3.46]{SS25-FQH}}):
  \begin{equation}
    \label{HomotopyEquivalenceFromQuotientOfAnnulusToSphereWedgeCircle}
    \begin{tikzcd}
      A^2 / \partial A^2
      \simeq
      \bracket({
        S^2 \cup a_{\mathrm{ext}}
      })/a_{\mathrm{ext}}
      \ar[
        r,
        "{ \sim }"
      ]
      &
      \bracket({    
        S^2 \cup a_{\mathrm{ext}}
      })/a_{\mathrm{int}}      
      \simeq
      S^2 \vee S^1
      \mathrlap{\,.}
    \end{tikzcd}
    \qedhere
  \end{equation}
\end{proof}

Noting that the closed annulus is equivalently the closed cylinder over $S^1$, we consider also the following variant:
\begin{definition}
  \label[definition]{ConstrictedAnnulus}
  The \emph{constricted annulus} (cf. \cref{SpatialDomainGraphics}) is the ``reduced cylinder'' over $S^1$:
  \begin{equation}
    \label{TheConstrictedAnnulus}
    A^2_{\mathrm{cns}}
    :=
    \frac
      { S^1 \times I }
      { \{s_0\} \times I }
    \in \Top
    \mathrlap{\,,}
  \end{equation}
  with boundary
  \begin{equation}
    \label{BoundaryOfConstrictedAnnulus}
    \begin{tikzcd}
      S^1 \vee S^1
      \ar[
        rr,
        "{ \iota_{A^2_{\mathrm{cns}}} }"
      ]
      &&
      A^2_{\mathrm{cns}}
      \mathrlap{\,.}
    \end{tikzcd}
  \end{equation}
\end{definition}

\begin{lemma}
\label[lemma]{ProjectionFromClosedToConstrictedAnnulus}
  The canonical projection, $c$, from the closed annulus with a disjoint basepoint to the constricted annulus \cref{TheConstrictedAnnulus} (hence from the unreduced to the reduced cylinder over $S^1$) induces on quotients (by the boundary inclusions) the projection onto the $S^2$ wedge summand in \cref{DeepBulkOfClosedAnnulus}:
\begin{equation}
  \label{TheProjectionFromClosedToConstrictedAnnulus}
  \begin{tikzcd}[row sep=13pt]
    \bracket({
      S^1 \sqcup S^1
    }) \sqcup \{\infty\}
    \ar[
      r,
      ->>
    ]
    \ar[d, hook]
    &
    S^1 \vee S^1
    \ar[d, hook]
    \\
    A^2 \sqcup \{\infty\}
    \ar[
      r,
      ->>,
      "{ c }"
    ]
    \ar[d, ->>]
    &
    A^2_{\mathrm{cns}}
    \ar[d, ->>]
    \\
    A^2 / \partial A^2
    \ar[
      r,
      ->>,
      "{ \overline{c} }"
    ]
    &
    S^2
    \\[-8pt]
    S^2 \vee S^1
    \ar[
      ur,
      "{
        (\mathrm{id}, 0)
      }"{sloped, swap}
    ]
    \ar[
      u,
      shorten=-2pt,
      "{ 
        \sim 
      }"{sloped, pos=.35}
    ]
    \mathrlap{\,.}
  \end{tikzcd}
\end{equation}
\end{lemma}
\begin{proof}
In view of \cref{SphereWithPolesIdentified}, consider the following diagram
\begin{equation}
  \begin{tikzcd}[column sep=large]
    \bracket({
      S^2 \cup a_{\mathrm{ext}}
    })
      \big/ 
    a_{\mathrm{ext}}
    \ar[
      rr,
      uphordown,
      "{ \overline{c} }"      
    ]
    \ar[
      r,
      ->>
    ]
    &
    \bracket({
      \bracket({
        S^2 \cup a_{\mathrm{ext}}
      })
        \big/
      a_{\mathrm{int}}
    })
      \big/ 
    a_{\mathrm{ext}}
    \ar[
      r,
      "{ \simeq }"
    ]
    &
    S^2
    \\
    \bracket({
      S^2 \cup a_{\mathrm{ext}}
    }) 
      \big/
    a_{\mathrm{int}}
    \ar[
      u,
      "{ \sim }"{sloped, pos=.4},
      "{
        \cref{HomotopyEquivalenceFromQuotientOfAnnulusToSphereWedgeCircle}
      }"{swap, pos=.4}
    ]
    \ar[
      ur,
      ->>
    ]
    \ar[
      r,
      "{ \simeq }"
    ]
    &
    S^2 \vee S^1
    \mathrlap{\,.}
    \ar[
      ur,
      "{ 
        (\mathrm{id}, 0) 
      }"{sloped, swap}
    ]
    &
  \end{tikzcd}
\end{equation}
Here, the right square commutes by construction, and the left triangle commutes because the order of taking the quotients by $a_{\mathrm{ext}}$ and $a_{\mathrm{int}}$ is irrelevant. Therefore, the whole diagram commutes and the claim follows.
\end{proof}

\subsubsection
{Examples of twisted relative cohomology}
\label{ExamplesOfTwistedRelativeCohomology}

Central to our discussion in \cref{OnDeterminingTheNearBdrClassfFib} is this example:
\begin{proposition}[Mapping \textbf{disk} boundary inclusion into Hopf fibration]
\label[proposition]{MappingBoundaryInclusionOfDiskIntoComplexHopfFibration}
In the case where 
\begin{enumerate}

\item  $\phi \defneq \big({\iota_{D^2} : \inlinetikzcd{S^1 \ar[r, hook] \& D^2}}\big) \sqcup \mathrm{id}_{\{\infty\}}$ is the boundary inclusion of the closed disk $D^2$,

\item  $\wp \defneq h_{\mathbb{C}} : \inlinetikzcd{ S^3 \ar[r] \& S^2 }$ is the complex Hopf fibration \cref{TheComplexHopfFibration}
\end{enumerate}
so that compatible pairs of classifying maps are unpointed diagrams of the form 
\[
  \begin{tikzcd}[row sep=small]
    S^1
    \ar[
      d,
      hook,
      "{ \iota_{D^2} }"{swap}
    ]
    \ar[
      r,
      dashed
    ]
    &
    S^3
    \ar[
      d,
      "{ h_{\mathbb{C}} \,}"{swap},
    ]
    \\
    D^2
    \ar[
      r,
      dashed
    ]
    &
    S^2
    \mathrlap{\,,}
  \end{tikzcd}
\]
the low-degree  homotopy of the $\iota_{D^2}$-relative $h_{\mathbb{C}}$-twisted mapping space \cref{TheRelativeMappingSpace} is:
\begin{subequations}
  \label{PinMapIotaD2hC}
  \begin{align}
    \label{pi0MapIotaD2HC}
    \pi_0\, 
    \UnpointedMap\bracket({
      \iota_{D^2}
      ,
      h_{\mathbb{C}}
    })
    & \simeq
    \mathbb{Z}
    \mathrlap{\,,}
    \\
    \label{Pi1MapIotaD2hC}
    \pi_1\, 
    \UnpointedMap\bracket({
      \iota_{D^2}
      ,
      h_{\mathbb{C}}
    })
    & \simeq
    1
    \mathrlap{\,,}
    \\
    \label{Pi2MapIotaD2hC}
    \pi_2\, 
    \UnpointedMap\bracket({
      \iota_{D^2}
      ,
      h_{\mathbb{C}}
    })
    & \simeq
    0
    \mathrlap{\,,}
  \end{align}
\end{subequations}
and the unwinding homomorphism \cref{TheUnwindingLES} is an isomorphism:
\begin{equation}
  \label{UnwindingIsoForIotaD2HC}
  \begin{tikzcd}[row sep=small]
    &&
    \overbrace{
    \pi_2 \, \UnpointedMap\bracket({
      S^1, S^3
    })
    }^{ \smash[t]{\mathbb{Z}} }
    \ar[
      dll,
      snake left,
      "{ \Unwinding_1 }"{swap},
      "{ \sim }"
    ]
    \\
    \underbrace{
    \pi_1 \, \Map\bracket({
      D^2/\partial D^2
      ,
      S^2
    })
    \mathrlap{\,.}
    }_{
      \mathbb{Z}
    }
  \end{tikzcd}
\end{equation}
as is the rewinding homomorphism \cref{TheRewindingLES}:
\begin{equation}
  \label{RewindingIsoForIotaD2HC}
  \begin{tikzcd}[row sep=small]
    &&
    \overset{
      \mathbb{Z}
    }{
    \overbrace{
    \pi_2 \, \UnpointedMap\bracket({D^2, S^2})
    }}
    \ar[
      dll,
      snake left,
      "{
        \Rewinding_1
      }"{swap},
      "{ \sim }"
    ]
    \\
    \underbrace{
      \pi_1 \UnpointedMap\bracket({S^1, S^1})
      \mathrlap{\,.}
    }_{
      \mathbb{Z}
    }
  \end{tikzcd}
\end{equation}
\end{proposition}
\begin{proof}
  The deep bulk quotient domain is evidently
  $
    \DeepBulkDomain 
    \defneq 
    D^2 / \partial D^2
      \simeq 
    S^2
  $.
Using this in \eqref{HomotopyLESForRelativeMap} gives for $\pi_0$ the sequence
\begin{equation}
  \begin{tikzcd}[column sep=6pt]
    && 
    \pi_1\UnpointedMap\bracket({
      S^1, S^3
    })
    \ar[
      dll,
      snake left,
      "{ 
        \Unwinding_0 
      }"{description}
    ]
    \\
    \pi_0\Map\bracket({
      S^2, S^2
    })
    \ar[r, shorten=-2pt]
    &
    \pi_0 \Map\bracket({
      \phi, \wp
    })
    \ar[r, shorten=-2pt]
    &
    \pi_0\UnpointedMap\bracket({
      S^1, S^3
    })
  \end{tikzcd}
  \;\,=\,\;
  \begin{tikzcd}[
    column sep=6pt
  ]
    && 
    1
    \ar[
      dll,
      snake left,
      "{
        \Unwinding_0
      }"{description}
    ]
    \\
    \mathbb{Z}
    \ar[
      r,
      shorten=-2pt
    ]
    &
    \pi_0\, \UnpointedMap\bracket({
      \iota_{D^2},
      h_{\mathbb{C}}
    })
    \ar[
      r,
      shorten=-2pt
    ]
    &
    \ast
    \mathrlap{\,,}
  \end{tikzcd}
\end{equation}
where in the bottom left we used \cref{HopfDegreeTheorem}.
This implies the claim \cref{pi0MapIotaD2HC} by exactness. 

The analogous analysis in degree 1 gives the following exact sequence
\begin{equation}
  \label{TheLESForpi1OfDiskIntoCHopf}
  \begin{tikzcd}[column sep=6pt]
    && 
    \pi_2\UnpointedMap\bracket({
      S^1, S^3
    })
    \ar[
      dll,
      snake left,
      "{ 
        \Unwinding_1 
      }"{description}
    ]
    \\
    \pi_1\Map\bracket({
      S^2, S^2
    })
    \ar[r, shorten=-2pt]
    &
    \pi_1\Map\bracket({
      \phi, \wp
    })
    \ar[r, shorten=-2pt]
    &
    \pi_1\UnpointedMap\bracket({
      S^1, S^3
    })
  \end{tikzcd}
  \;\,=\,\;
  \begin{tikzcd}[column sep=6pt]
    && 
    \mathbb{Z}
    \ar[
      dll,
      snake left,
      "{
        \Unwinding_1
      }"{description}
    ]
    \\
    \mathbb{Z}
    \ar[r, shorten=-2pt]
    &
    \pi_1 \UnpointedMap\bracket({
      \iota_{D^2},
      h_{\mathbb{C}}
    })
    \ar[r, shorten=-2pt]
    &
    1
    \mathrlap{\,,}
  \end{tikzcd}
\end{equation}
where in the bottom left we used \cref{IterationOfHomotopyGroups,ComplexHopfFibrationOnPi3}, and on the right we used \cref{HomotopyGroupsOfFreeLoopSpace}.

We claim now that the connecting homomorphism $\Unwinding_1$ in \cref{TheLESForpi1OfDiskIntoCHopf} is surjective. 
To see this, we first observe with \cref{MapFromHomotopyLESOfRelMap} that it factors as
\begin{equation}
  \label{AFactorizationOfUnwinding}
  \begin{tikzcd}[
    row sep=6pt, column sep=large
  ]
    \pi_2\, \UnpointedMap\bracket({
      S^1, S^3
    })
    \ar[
      rr,
      uphordown,
      "{ \Unwinding_1 }"{description}
    ]
    \ar[
      r,
      "{ (h_{\mathbb{C}})_\ast }"
    ]
    \ar[
      d,
      equals
    ]
    &
    \pi_2\, \UnpointedMap\bracket({
      S^1, S^2
    })
    \ar[
      r,
      "{  \partial^{\iota_{\smash{D^2}}^\ast}_1 }"
    ]
    \ar[
      d,
      equals
    ]
    &
    \pi_1\, \Map\bracket({
      S^2, S^2
    })
    \\
    \pi_2\bracket({S^3})
    \times
    \pi_2\bracket({\Omega S^3})
    \ar[
      d,
      equals
    ]
    \ar[
      r,
      "{
        (h_{\mathbb{C}})_\ast
      }"
    ]
    &
    \pi_2\bracket({S^2})
    \times
    \pi_2\bracket({\Omega S^2}) 
    \ar[
      d,
      equals
    ]
    \\
    0 
    \oplus
    \mathbb{Z}
    \ar[
      r,
      "{ (0,\mathrm{id}) }"
    ]
    &
    \mathbb{Z} \oplus \mathbb{Z}
    \mathrlap{\,,}
  \end{tikzcd}
\end{equation}
where in the bottom left we used \cref{HomotopyGroupsOfFreeLoopSpace,ComplexHopfFibrationOnPi3}.
Secondly, we notice that the map $\partial^{\iota_{D^2}^\ast}_1$ in \cref{AFactorizationOfUnwinding} is surjective, 
with kernel the first $\mathbb{Z}$-summand,
by exactness of the following homotopy LES (from \cref{MappingFiberSequenceOfCWInclusion}):
\begin{equation}
  \label{InvokingTheTheQStarSequence}
  \begin{tikzcd}[
    row sep=5pt, column sep=32pt
  ]
    \pi_2\, \UnpointedMap\bracket({
      D^2, S^2
    })
    \ar[
      d, 
      equals
    ]
    \ar[
      r,
      "{
       \pi_2(\iota^\ast_{\smash{D^2}})
      }"
    ]
    &
    \pi_2\, \UnpointedMap\bracket({
      S^1, S^2
    })
    \ar[
      d, 
      equals
    ]
    \ar[
      r,
      ->>,
      "{  
        \partial^{\iota^\ast_{\smash{D^2}}}_1 
      }"
    ]
    &
    \pi_1\, \Map\bracket({
      S^2, S^2
    })
    \ar[
      d,
      equals
    ]
    \ar[
      r,
      "{ q^\ast }"
    ]
    &
    \pi_1\, \UnpointedMap\bracket({
      D^2, S^2
    })    
    \ar[
      d,
      equals
    ]
    \\
    \mathbb{Z}
    \ar[
      r,
      "{ (\mathrm{id},0) }"
    ]
    &
    \mathbb{Z} \oplus \mathbb{Z}
    \ar[r, ->>]
    &
    \mathbb{Z}
    \ar[r]
    &
    1
    \mathrlap{\,.}
  \end{tikzcd}
\end{equation}
In combination this shows that $\Unwinding_1$ is surjective. Together with the exactness of  \cref{TheLESForpi1OfDiskIntoCHopf} this yields the second claim \cref{Pi1MapIotaD2hC}.
By exactness of the rewinding LES \cref{TheRewindingLES} this implies that also $\Rewinding_1$ is surjective.
But surjective homomorphisms of the form $\inlinetikzcd{\mathbb{Z} \ar[r] \& \mathbb{Z}}$ are necessarily isomorphisms, whence we also have shown the third claim \cref{UnwindingIsoForIotaD2HC} and fourth claim \cref{RewindingIsoForIotaD2HC}.

Finally, the analogous analysis in degree 2 gives the following exact sequence
\begin{equation}
  \label{TheLESForpi2OfDiskIntoCHopf}
  \begin{tikzcd}[
    column sep=4pt
  ]
    && 
    \pi_3\UnpointedMap\bracket({
      S^1, S^3
    })
    \ar[
      dll,
      snake left,
      "{ 
        \Unwinding_2 
      }"{description}
    ]
    \\
    \pi_2\Map\bracket({
      S^2, S^2
    })
    \ar[r, shorten=-2pt]
    &
    \pi_2\Map\bracket({
      \phi, \wp
    })
    \ar[r, shorten=-2pt]
    &
    \pi_2\UnpointedMap\bracket({
      S^1, S^3
    })
    \ar[
      dll,
      snake left,
      "{ 
        \Unwinding_1 
      }",
      "{ \sim }"{swap}
    ]
    \\
    \pi_1\Map\bracket({
      S^2, S^2
    })
  \end{tikzcd}
  \;\,=\,\;
  \begin{tikzcd}[
    column sep=4pt
  ]
    &
    &[-15pt] 
    \mathbb{Z}
    \oplus
    \mathbb{Z}_{/2}
    \ar[
      dll,
      snake left,
      "{
        \Unwinding_2
      }"{description}
    ]
    \\
    \mathbb{Z}_{/2}
    \ar[r, shorten=-2pt]
    &
    \pi_2 \UnpointedMap\bracket({
      \iota_{D^2},
      h_{\mathbb{C}}
    })
    \ar[r, shorten=-2pt]
    &
    \mathbb{Z}
    \ar[
      dll,
      snake left,
      "{ \sim }"{swap}
    ]
    \\
    \mathbb{Z}
    \mathrlap{\,,}
  \end{tikzcd}
\end{equation}
where for the first two groups we used \cref{SecondStableStem} and the isomorphy of $\Unwinding_1$ is the previous result \cref{UnwindingIsoForIotaD2HC}. 

Therefore, by exactness, to prove the claim \cref{Pi2MapIotaD2hC} it is now sufficient to see that $\Unwinding_2$ is surjective:
This follows by factoring it
via \cref{MapFromHomotopyLESOfRelMap}:
\begin{equation}
  \label{AFactorizationOfUnwinding2}
  \begin{tikzcd}[
    row sep=6pt, column sep=large
  ]
    \pi_3\, \UnpointedMap\bracket({
      S^1, S^3
    })
    \ar[
      rr,
      uphordown,
      "{ \Unwinding_2 }"{description}
    ]
    \ar[
      r,
      "{ (h_{\mathbb{C}})_\ast }"
    ]
    \ar[
      d,
      equals
    ]
    &
    \pi_3\, \UnpointedMap\bracket({
      S^1, S^2
    })
    \ar[
      r,
      "{  
        \partial^{\iota_{\smash{D^2}}^\ast}_2 
      }"
    ]
    \ar[
      d,
      equals
    ]
    &
    \pi_2\, \Map\bracket({
      S^2, S^2
    })
    \\
    \pi_3\bracket({S^3})
    \times
    \pi_3\bracket({\Omega S^3})
    \ar[
      d,
      equals
    ]
    \ar[
      r,
      "{
        (h_{\mathbb{C}})_\ast
      }"
    ]
    &
    \pi_3\bracket({S^2})
    \times
    \pi_3\bracket({\Omega S^2}) 
    \ar[
      d,
      equals
    ]
    \\
    \mathbb{Z} 
    \oplus
    \mathbb{Z}_{/2}
    \ar[
      r,
      "{ 
        (\mathrm{id},\mathrm{id}) 
      }"
    ]
    &
    \mathbb{Z} 
      \oplus 
    \mathbb{Z}_{/2}
    \mathrlap{\,,}
  \end{tikzcd}
\end{equation}
where on the left we used \cref{ComplexHopfFibrationOnPi3},
and then using the
exactness of the following homotopy LES (from \cref{MappingFiberSequenceOfCWInclusion}):
\begin{equation}
  \label{InvokingTheTheQStarSequence2}
  \begin{tikzcd}[
    row sep=5pt, column sep=32pt
  ]
    \pi_3\, \UnpointedMap\bracket({
      D^2, S^2
    })
    \ar[
      d, 
      equals
    ]
    \ar[
      r,
      "{
       \pi_3(\iota^\ast_{\smash{D^2}})
      }"
    ]
    &
    \pi_3\, \UnpointedMap\bracket({
      S^1, S^2
    })
    \ar[
      d, 
      equals
    ]
    \ar[
      r,
      ->>,
      "{  
        \partial^{\iota^\ast_{\smash{D^2}}}_2 
      }"
    ]
    &
    \pi_2\, \Map\bracket({
      S^2, S^2
    })
    \ar[
      d,
      equals
    ]
    \ar[
      r,
      "{ q^\ast }"
    ]
    &
    \pi_2\, \UnpointedMap\bracket({
      D^2, S^2
    })    
    \ar[
      d,
      equals
    ]
    \\
    \mathbb{Z}
    \ar[
      r,
      "{ (\mathrm{id},0) }"
    ]
    &
    \mathbb{Z} 
      \oplus 
    \mathbb{Z}_{/2}
    \ar[r, ->>]
    &
    \mathbb{Z}_{/2}
    \ar[r]
    &
    0
    \mathrlap{\,,}
  \end{tikzcd}
\end{equation}
which jointly shows that $\Unwinding_2$ is projection onto the $\mathbb{Z}_{/2}$-factor. 
\end{proof}

In preparation of discussing the closed annulus domain in \cref{MappingBoundaryInclusionOfClosedAnnulusIntoComplexHopfFibration}, we now first consider the following statement:
\begin{proposition}[Mapping \textbf{constricted annulus} boundary inclusion into Hopf fibration]
\label[proposition]{MappingBoundaryInclusionOfConstrictedAnnulusIntoComplexHopfFibration}
In the case where 
\begin{enumerate}

\item  $\phi \defneq \big(\inlinetikzcd{ S^1 \vee S^1 \ar[r, "{ \iota_{A^2_{\mathrm{cns}}} }"] \&  A^2_{\mathrm{cns}}  }\big) $ is the boundary inclusion \cref{BoundaryOfConstrictedAnnulus} of the constricted annulus \textup{(\cref{ConstrictedAnnulus})},

\item  $\wp \defneq h_{\mathbb{C}} : \inlinetikzcd{ S^3 \ar[r] \& S^2 }$ is the complex Hopf fibration \cref{TheComplexHopfFibration},
\end{enumerate}
we have
\begin{subequations}
\begin{align}
  \label{pi0MapIotaA2CnsHC}
  \pi_0
  \Map\bracket({
    \iota_{A^2_{\mathrm{cns}}},
    h_{\mathbb{C}}
  })
  &
  \simeq 
  \mathbb{Z}
  \\
  \label{pi1MapIotaA2CnsHC}
  \pi_1
  \Map\bracket({
    \iota_{A^2_{\mathrm{cns}}},
    h_{\mathbb{C}}
  })
  & 
  \simeq
  1
  \mathrlap{\,,}
  \\
  \label{pi2MapIotaA2CnsHC}
  \pi_2
  \Map\bracket({
    \iota_{A^2_{\mathrm{cns}}},
    h_{\mathbb{C}}
  })
  & 
  \simeq
  \mathbb{Z}
  \mathrlap{\,,}
\end{align}
\end{subequations}
and the unwinding homomorphism \cref{TheUnwindingLES} is:
\begin{equation}
  \label{UnwindingIsoForIotaA2HC}
  \begin{tikzcd}[row sep=small,
   column sep=-10pt
  ]
    &&
    \pi_2 \, \Map\bracket({
      \partial A^2_{\mathrm{cns}} 
      , 
      S^3
    })
    \ar[
      dll,
      snake left,
      "{
        \Unwinding_1
      }"{description}
    ]
    \\
    \pi_1\, 
    \Map\bracket({
      A^2_{\mathrm{cns}}
      /
      \partial A^2_{\mathrm{cns}}, 
      S^2
    })
  \end{tikzcd}
  \;\;=\;\;
  \begin{tikzcd}[row sep=small,
    ampersand replacement=\&
  ]
    \&\& \mathbb{Z}^2
    \ar[
      dll,
      "{
        (
        \begin{matrix}
          +1 & -1
        \end{matrix}
        )
      }"{description}
    ]
    \\
    \mathbb{Z}
  \end{tikzcd}
\end{equation}
while the rewinding homomorphism \cref{TheRewindingLES} is:
\begin{equation}
  \label{RewindingIsoForIotaA2CnsHC}
  \begin{tikzcd}[row sep=small,
    column sep=0pt
  ]
    &&
    \pi_2 \Map\bracket({
      A^2_{\mathrm{cns}}, 
      S^2
    })
    \ar[
      dll,
      snake left,
      "{
        \Rewinding_1
      }"{description},
    ]
    \\
    \pi_1\, \Map\bracket({
      \partial A^2_{\mathrm{cns}}
      , 
      S^1
    })
    \mathrlap{\,.}
  \end{tikzcd}
  \;\;
  =
  \;\;
  \begin{tikzcd}[row sep=small]
    && \mathbb{Z}
    \ar[
      dll,
      "{ 0 }"{description}
    ]
    \\
    1
  \end{tikzcd}
\end{equation}
\end{proposition}
\begin{proof}
First, note that the quotient domain in this case is again the 2-sphere, $\DeepBulkDomain \defneq A^2_{\mathrm{cns}}/\partial A^2_{\mathrm{cns}} \simeq S^2$.
In particular, all groups are identified as claimed, by \cref{IterationOfHomotopyGroups,ComplexHopfFibrationOnPi3}, so that the claim \cref{RewindingIsoForIotaA2CnsHC} is already established.

Moreover, the long exact sequence \cref{HomotopyLESForRelativeMap} here is of this form:
\begin{equation}
  \label{UnwindingLESForIotaA2CnsIntoHC}
  \adjustbox{scale=.95}{
  \begin{tikzcd}[
    column sep=6pt
  ]
    &
    {}
    \ar[r, -, dotted]
    &
    \pi_3\Map\bracket({
      S^1, S^3
    })^2
    \ar[
      dll,
      snake left,
      "{ \Unwinding_2 }"{description}
    ]
    \\
    \pi_2\Map\bracket({
      S^2, S^2
    })
    \ar[r, shorten=-2pt]
    &
    \pi_2\Map\bracket({
      \phi,\wp
    })
    \ar[r, shorten=-2pt]
    &
    \pi_2 \Map\bracket({
      S^1,
      S^3
    })^2
    \ar[
      dll,
      snake left,
      "{
        \Unwinding_1
      }"{description}
    ]
    \\
    \pi_1 \Map\bracket({
      S^2,
      S^2
    })
    \ar[
      r,
      shorten=-2pt
    ]
    &
    \pi_1 \Map\bracket({
      \phi,
      \wp
    })
    \ar[
      r,
      shorten=-2pt
    ]
    &
    \pi_1 \Map\bracket({
      S^1, S^3
    })^2
    \ar[
      dll,
      snake left,
      "{ \Unwinding_0 }"{description}
    ]
    \\
    \pi_0\Map\bracket({
      S^2, S^2
    })
    \ar[r, shorten=-2pt]
    &
    \pi_0\Map\bracket({
      \phi, \wp
    })
    \ar[r, shorten=-2pt]
    &
    \pi_0\Map\bracket({
      S^1, S^3
    })^2
  \end{tikzcd}
  \;\;\;=\;\;
  \begin{tikzcd}[
    column sep=4.5pt
  ]
    && 
    \mathbb{Z}_{/2}^2
    \ar[
      dll,
      snake left,
      "{ \Unwinding_2 }"{description}
    ]
    \\
    \mathbb{Z}_{/2}
    \ar[r, shorten=-2pt]
    &
    \pi_2 \Map\bracket({
      \iota_{A^2_{\mathrm{cns}}},
      h_{\mathbb{C}}
    })
    \ar[r, shorten=-2pt]
    &
    \mathbb{Z}^2
    \ar[
      dll,
      snake left,
      "{
        \Unwinding_1
      }"{description}
    ]
    \\
    \mathbb{Z}
    \ar[
      r,
      shorten=-2pt
    ]
    & 
    \pi_1\Map\bracket({
      \iota_{A^2_{\mathrm{cns}}},
      h_{\mathbb{C}}
    })
    \ar[
      r,
      shorten=-2pt
    ]
    & 
    1
    \ar[
      dll,
      snake left,
      "{ \Unwinding_0 }"{description}
    ]
    \\
    \mathbb{Z}
    \ar[r, shorten=-2pt]
    &
    \pi_0\Map\bracket({
      \iota_{A^2_{\mathrm{cns}}},
      h_{\mathbb{C}}
    })
    \ar[r, shorten=-2pt]
    &
    \ast
    \mathrlap{\,,}
  \end{tikzcd}
  }
\end{equation}
where on the right we used \cref{PointedMapsOutOfWedgeSum} and on the left \cref{IterationOfHomotopyGroups,ComplexHopfFibrationOnPi3,SecondStableStem}.

Hence exactness of the tail end of this sequence immediately implies the claim \cref{pi0MapIotaA2CnsHC}.

To understand $\Unwinding_1$, 
we first observe with \cref{MapFromHomotopyLESOfRelMap} that it factors, much as in \cref{AFactorizationOfUnwinding}, as shown in the top row of the following diagram
\begin{equation}
  \label{AFactorizationOfAnotherUnwinding}
  \begin{tikzcd}[
    row sep=6pt, column sep=large
  ]
    \pi_2
    \Map\bracket({
      S^1, S^3
    })^2
    \ar[
      rr,
      uphordown,
      "{ \Unwinding_1 }"{description}
    ]
    \ar[
      r,
      "{ (h_{\mathbb{C}})_\ast }"
    ]
    \ar[
      d,
      equals
    ]
    &
    \pi_2
    \Map\bracket({
      S^1, S^2
    })^2
    \ar[
      r,
      "{  
        \partial
        ^{
          \iota
          _{\smash{
            A
              ^{2}
              _{\mathrm{cns}}
            }
          }^\ast
        }
        _1 
      }"
    ]
    \ar[
      d,
      equals
    ]
    &
    \pi_1\, \Map\bracket({
      S^2, S^2
    })
    \ar[
      d,
      equals
    ]
    \\
    \pi_3\bracket({S^3})^2
    \ar[
      d,
      equals
    ]
    \ar[
      r,
      "{
        (h_{\mathbb{C}})_\ast
      }"
    ]
    &
    \pi_3\bracket({S^2})^2
    \ar[
      d,
      equals
    ]
    &
    \pi_3\bracket({S^2})
    \ar[
      d,
      equals
    ]
    \\
    \mathbb{Z}^2
    \ar[
      r,
      "{ \mathrm{id} }"
    ]
    &
    \mathbb{Z}^2
    \ar[
       r,
       "{
         \big(
           +1 \; -1
         \big)
       }"
    ]
    &
    \mathbb{Z}
    \mathrlap{\,.}
  \end{tikzcd}
\end{equation}
Here the second row follows by \cref{IterationOfHomotopyGroups} and the last row by \cref{ComplexHopfFibrationOnPi3}.
To see the right square in this diagram,
observe that exactness of the sequence (\cref{MappingFiberSequenceOfCWInclusion})
\begin{equation}
  \begin{tikzcd}[
    column sep=60pt
  ]
    \pi_2\Map\bracket({
      S^2, S^2
    })
    \ar[
      r,
      "{  
        \pi_2(q^\ast_{\smash{A^2_{\mathrm{cns}}}})
      }"{yshift=3pt},
      "{ 
      = 0
      \,
      \scalebox{.7}{\cref{AnnulusCoprojectionIsNull}}
      }"{swap}
    ]
    &
    \pi_2\Map\bracket({S^1, S^2})
    \ar[
      r,
      "{
        \pi_2(
          \iota^\ast_{\smash{
            A^2_{\mathrm{cns}}
          }}
        )
      }"{yshift=3pt},
      "{ 
        = \Delta 
        \,
        \scalebox{.7}{\cref{AnnulusBoundaryIncIsCodiagonal}}
      }"{swap}
    ]
    &
    \pi_2\Map\bracket({S^1,S^2})^2
    \ar[
      dll,
      snake left,
      "{
        \partial_1^{
          \iota^\ast_{\smash{
            A^2_{\mathrm{cns}}
          }}
        }
      }"{description}
    ]
    \\
    \pi_1\Map\bracket({
      S^2, S^2
    })
    \ar[
      r,
      "{  
        \pi_1(q^\ast_{\smash{A^2_{\mathrm{cns}}}})
      }"{yshift=3pt},
      "{
       = 0
       \,
       \scalebox{.7}{\cref{AnnulusCoprojectionIsNull}}
      }"{swap}
    ]
    &
    \pi_1 \Map\bracket({S^1, S^2})
  \end{tikzcd}
\end{equation}
shows that 
\begin{equation}
 \label{IdentifyingPartial1IotaAstA2Cns}
 \mathrm{ker}\bracket({
   \Unwinding_1
 })
 \simeq
 \mathrm{ker}\bracket({
   \partial_1^{
     \iota^\ast_{\smash{
       A^2_{\mathrm{cns}}
     }}
   }
 })
 \simeq
 \Delta(\mathbb{Z})
 \subset 
 \mathbb{Z}^2
 \mathrlap{\,,}
\end{equation}
which demonstrates the claimed matrix representation \cref{UnwindingIsoForIotaA2HC} of $\Unwinding_1$ (up to a conventional global sign). This being surjective, exactness of \cref{UnwindingLESForIotaA2CnsIntoHC} implies the claim \cref{pi1MapIotaA2CnsHC}.

Similarly, to understand $\Unwinding_2$, we factor it via \cref{MapFromHomotopyLESOfRelMap}, as: 
\begin{equation}
  \label{AFactorizationOfAnotherUnwinding2}
  \begin{tikzcd}[
    row sep=6pt, column sep=huge
  ]
    \pi_3
    \Map\bracket({
      S^1, S^3
    })^2
    \ar[
      rr,
      uphordown,
      "{ \Unwinding_2 }"{description}
    ]
    \ar[
      r,
      "{ (h_{\mathbb{C}})_\ast }"
    ]
    \ar[
      d,
      equals
    ]
    &
    \pi_3
    \Map\bracket({
      S^1, S^2
    })^2
    \ar[
      r,
      "{  
        \partial
        ^{
          \iota
          _{\smash{
            A
              ^{2}
              _{\mathrm{cns}}
            }
          }^\ast
        }
        _2 
      }"
    ]
    \ar[
      d,
      equals
    ]
    &
    \pi_2\, \Map\bracket({
      S^2, S^2
    })
    \ar[
      d,
      equals
    ]
    \\
    \pi_4\bracket({S^3})^2
    \ar[
      d,
      equals
    ]
    \ar[
      r,
      "{
        (h_{\mathbb{C}})_\ast
      }"
    ]
    &
    \pi_4\bracket({S^2})^2
    \ar[
      d,
      equals
    ]
    &
    \pi_4\bracket({S^2})
    \ar[
      d,
      equals
    ]
    \\
    \mathbb{Z}_{/2}^2
    \ar[
      r,
      "{ \mathrm{id} }"
    ]
    &
    \mathbb{Z}_{/2}^2
    \ar[
       r,
       "{
         (
           [+1] \; [-1]
         )
       }"
    ]
    &
    \mathbb{Z}_{/2}
    \mathrlap{\,,}
  \end{tikzcd}
\end{equation}
where the bottom right morphism follows analogously to \cref{IdentifyingPartial1IotaAstA2Cns}. This shows that $\Unwinding_2$ is surjective.

Therefore the exactness of \cref{UnwindingLESForIotaA2CnsIntoHC} implies that $\pi_2\Map\bracket({\iota_{A^2_{\mathrm{cns}}}, h_{\mathbb{C}}}) \simeq \mathrm{ker}\bracket({\Unwinding_1})$. With \cref{IdentifyingPartial1IotaAstA2Cns} this yields the claim \cref{pi2MapIotaA2CnsHC}.
\end{proof}

Now we use \cref{MappingBoundaryInclusionOfConstrictedAnnulusIntoComplexHopfFibration} to analyze the situation over the actual closed annulus:

\begin{proposition}
[Mapping \textbf{closed annulus} boundary inclusion into Hopf fibration]
\label[proposition]{MappingBoundaryInclusionOfClosedAnnulusIntoComplexHopfFibration}
In the case where 
\begin{enumerate}

\item  $\phi \defneq \big(\inlinetikzcd{ S^1 \sqcup S^1 \ar[r, "{ \iota_{A^2} }"] \&  A^2  }\big) \sqcup \{\infty\}  $ is the boundary inclusion of the closed annulus,

\item  $\wp \defneq h_{\mathbb{C}} : \inlinetikzcd{ S^3 \ar[r] \& S^2 }$ is the complex Hopf fibration \cref{TheComplexHopfFibration},
\end{enumerate}
we have
\begin{subequations}
  \label{pinOfMapsOfBdrOfClosedAnnulusIntoHC}
  \begin{align}
  \label{pi0OfMapsOfBdrOfClosedAnnulusIntoHC}
  \pi_0\Map\bracket({
    \iota_{A^2}
    ;
    h_{\mathbb{C}}
  })
  \simeq
  \mathbb{Z}
  \\
  \label{pi1OfMapsOfBdrOfClosedAnnulusIntoHC}
  \pi_1\Map\bracket({
    \iota_{A^2}
    ;
    h_{\mathbb{C}}
  })
  \simeq
  \mathbb{Z}_{\mathrlap{u}}
  \\
  \label{pi2OfMapsOfBdrOfClosedAnnulusIntoHC}
  \pi_2\Map\bracket({
    \iota_{A^2}
    ;
    h_{\mathbb{C}}
  })
  \simeq
  \mathbb{Z}
  \mathrlap{\,,}
  \end{align}
\end{subequations}
and the unwinding homomorphism \cref{TheUnwindingLES} is:
\begin{equation}
\label
{UnwindingHomomorphismForClosedAnnulus}
  \begin{tikzcd}[row sep=small,
    column sep=-10pt
  ]
    &&
    \pi_2\UnpointedMap\bracket({
      S^1, S^3
    })^2
    \ar[
      dll,
      snake left,
      "{ \Unwinding_1 }"{description}
    ]
    \\
    \pi_1\Map\bracket({
      S^2 \vee S^1,
      S^2
    })
  \end{tikzcd}
  \;\;=\;\;
  \begin{tikzcd}[row sep=small,
    ampersand replacement=\&, 
    column sep=35pt
  ]
    \&\& 
    \mathbb{Z}^2.
    \ar[
      dll,
      "{
      \scalebox{0.7}{$
        \left(
        \begin{matrix}
          +1 & -1
          \\
          \,0 & \;0
        \end{matrix}
        \right)
        $}
      }"{description}
    ]
    \\
    \mathbb{Z}
      \times 
    \mathbb{Z}_u
  \end{tikzcd}
\end{equation}
\end{proposition}
Here, the subscript ``$u$'' is just to track copies of $\mathbb{Z}$ that arise as ``\emph{u}nshifted'' homotopy groups, in the sense of the computation \cref{SampleComputationUShifted}. These copies disappear in the analogous computation for the differentially refined situation in \cref{OnHomotopyOfThePhaseSpace}, cf. \cref{WhereTheUnshiftedCopyDisappears}.
\begin{proof}
With \cref{DeepBulkOfClosedAnnulus}, the long exact sequence \cref{HomotopyLESForRelativeMap} here is of this form:
\begin{equation}
  \label{UnwindingLESForIotaA2IntoHC}
  \adjustbox{scale=.95,center}{
  \begin{tikzcd}[
    column sep=5pt
  ]
    &
    {}
    \ar[
      r, -, dotted
    ]
    &
    \pi_3\UnpointedMap\bracket({
      S^1, S^3
    })^2
    \ar[
      dll,
      snake left,
      "{ \Unwinding_2 }"{description}
    ]
    \\
    \pi_2\Map\bracket({
      S^2\!\vee\!S^1,
      S^2
    })
    \ar[
      r, shorten=-2pt
    ]
    &
    \pi_2\Map(\phi,\wp)
    \ar[r, shorten=-2pt]
    &
    \pi_2 \UnpointedMap\bracket({
      S^1,
      S^3
    })^2
    \ar[
      dll,
      snake left,
      "{
        \Unwinding_1
      }"{description}
    ]
    \\
    \pi_1 \Map\bracket({
      S^2\!\vee\!S^1,
      S^2
    })
    \ar[
      r,
      shorten=-2pt
    ]
    &
    \pi_1 \Map\bracket({
      \phi,
      \wp
    })
    \ar[
      r,
      shorten=-2pt
    ]
    &
    \pi_1 \UnpointedMap\bracket({
      S^1, S^3
    })^2
    \ar[
      dll,
      snake left,
      "{ \Unwinding_0 }"{description}
    ]
    \\
    \pi_0 \Map\bracket({
      S^2\!\vee\!S^1,
      S^2
    })
    \ar[
      r,
      shorten=-2pt
    ]
    &
    \pi_0 \Map\bracket({
      \phi,
      \wp
    })
    \ar[
      r,
      shorten=-2pt
    ]
    &
    \pi_0 \UnpointedMap\bracket({
      S^1, S^3
    })^2
  \end{tikzcd}
  \;\;\;=\;\;\;
  \begin{tikzcd}[
    column sep=5pt
  ]
    &
    {}
    \ar[r, -, dotted, shorten >=20pt]
    &
    \mathllap{\mathbb{Z}_{/2}^2
    \oplus\,
    }
    \mathbb{Z}^2
    \ar[
      dll,
      snake left,
      "{ \Unwinding_2 }"{description}
    ]
    \\
    \mathbb{Z}_{/2}
    \oplus
    \mathbb{Z}
    \ar[r, shorten=-2pt]
    &
    \pi_2\UnpointedMap\bracket({
      \iota_{A^2},
      h_{\mathbb{C}}
    })
    \ar[r, shorten=-2pt]
    &
    \mathbb{Z}^2
    \ar[
      dll,
      snake left,
      "{
        \Unwinding_1
      }"{description}
    ]
    \\
    \mathbb{Z}
    \!\times\!
    \mathbb{Z}_u
    \ar[r, shorten=-2pt]
    & 
    \pi_1\UnpointedMap\bracket({
      \iota_{A^2},
      h_{\mathbb{C}}
    })
    \ar[r, shorten=-2pt]
    & 
    1
    \ar[
      dll,
      snake left,
      "{ \Unwinding_0 }"{description}
    ]
    \\
    \mathbb{Z}
    \ar[r, shorten=-2pt]
    &
    \pi_0\UnpointedMap\bracket({
      \iota_{A^2},
      h_{\mathbb{C}}
    })
    \ar[r, shorten=-2pt]
    &
    \ast
    \mathrlap{\,,}
  \end{tikzcd}
  }
\end{equation}
where the subscript $\mathbb{Z}_u$ is just to track this ``unshifted'' copy of the integers, in the sense of computations like this one:
  \begin{equation}
  \label{SampleComputationUShifted}
    \begin{aligned}
    \pi_1 \, 
    \Map\bracket({S^2 \vee S^1, S^2})
    &
    \;
    \underset{\mathclap{
      \scalebox{.7}{\cref{MapTakingPushoutInFirstArgumentToPullback}}
    }}{\simeq}
    \;
    \pi_1\bracket({
      \Map\bracket({S^2, S^2})
      \times
      \Map\bracket({ S^1, S^2 })
    })
    \\
    & 
   \; \underset{\mathclap{
      \cref{PiNPreservesProducts}
    }}
    {\simeq}
   \; \pi_1\bracket({
      \Map\bracket({S^2, S^2})
    })
      \times
   \; \pi_1\bracket({
      \Map\bracket({ S^1, S^2 })
    })
    \\
    &
    \; \underset{\mathclap{
      \cref{IterationOfHomotopyGroups}
    }}
    {\simeq}
    \; \pi_3\bracket({S^2})
    \times
    \pi_2\bracket({S^2})
    \\
    &
   \; \underset{\mathclap{
      \scalebox{.7}{\cref{ComplexHopfFibrationOnPi3}}
    }}{\simeq}
   \; \mathbb{Z} \times \mathbb{Z}_u
    \mathrlap{\,.}
    \end{aligned}
  \end{equation}

Now, first, exactness of the tail end of the above sequence immediately implies the claim \cref{pi0OfMapsOfBdrOfClosedAnnulusIntoHC}.

Next, as before, $\Unwinding_1$ factors via \cref{MapFromHomotopyLESOfRelMap}:
\begin{equation}
  \label{FactorizationOfUnwindingForClosedAnnulus}
  \begin{tikzcd}[
    row sep=6pt, column sep=large
  ]
    \pi_2\, \UnpointedMap\bracket({
      S^1, S^3
    })^2
    \ar[
      rr,
      uphordown,
      "{ \Unwinding_1 }"{description}
    ]
    \ar[
      r,
      "{ (h_{\mathbb{C}})_\ast }"
    ]
    \ar[
      d,
      equals
    ]
    &
    \pi_2\, \UnpointedMap\bracket({
      S^1, S^2
    })^2
    \ar[
      r,
      "{  \partial^{\iota_{\smash{A^2}}^\ast}_1 }"
    ]
    \ar[
      d,
      equals
    ]
    &
    \pi_1\, \Map\bracket({
      S^2 \vee S^1, S^2
    })
    \\
    \pi_2\bracket({S^3})^2
    \times
    \pi_2\bracket({\Omega S^3})^2
    \ar[
      d,
      equals
    ]
    \ar[
      r,
      "{
        (h_{\mathbb{C}})_\ast
      }"
    ]
    &
    \pi_2\bracket({S^2})^2
    \times
    \pi_2\bracket({\Omega S^2})^2 
    \ar[
      d,
      equals
    ]
    \\
    0_u 
    \oplus
    \mathbb{Z}^2
    \ar[
      r,
      "{
        0
        \oplus
        \mathrm{id}
      }"
    ]
    &
    \mathbb{Z}_u^2 \oplus \mathbb{Z}^2
    \mathrlap{\,.}
  \end{tikzcd}
\end{equation}
Since the left map injects the unshifted copy of $\mathbb{Z}^2$ it is now sufficient to compute $\partial_1^{\iota^\ast_{\smash{A^2}}}$ on this unshifted copy.

For this, consider naturality \cref{NaturalityOfHomotopyLES} of the corresponding exact sequence under pullback along the map \cref{TheProjectionFromClosedToConstrictedAnnulus} from the closed to the constricted annulus:
\begin{equation}
  \begin{tikzcd}[row sep=-1pt, column sep=5pt]
    \pi_2\Map\bracket({
      S^1, S^2
    })^2
    \ar[
      ddd,
      "{ \pi_2(c^\ast) }"
      {swap}
    ]
    \ar[
      rrr,
      "{
        \partial_1^{
          \iota^\ast_{\smash{A^2_{\mathrm{cns}}}}
        }
      }"
    ]
    &&[50pt]&
    \pi_1 \Map\bracket({
     S^2,
     S^2
    })
    \ar[
      ddd,
      "{ \pi_1(c^\ast) }"
    ]
    \\
    & 
    \mathbb{Z}^2
    \ar[
      r,
      "{ (+1 \; -1) }"
    ]
    \ar[
      d,
      "{
        (0,\mathrm{id})
      }"{swap}
    ]
    &
    \mathbb{Z}
    \ar[
      d,
      "{
        (\mathrm{id}, 0)
      }"
    ]
    \\[30pt]
    &
    \mathllap{
      \mathbb{Z}^2_u
      \oplus 
    }
    \mathbb{Z}^2
    \ar[r]
    &
    \mathbb{Z}
    \mathrlap{
      \oplus \mathbb{Z}_u
    }
    \\
    \pi_2 \UnpointedMap\bracket({
      S^1, S^2
    })^2    
    \ar[
      rrr,
      "{
        \partial_1^{
          \iota^\ast_{\smash{A^2}}
        }
      }"
    ]
    &&&
    \pi_1 \Map\bracket({
     S^2 \vee S^1,
     S^2
    })
    \mathrlap{\,.}
  \end{tikzcd}
\end{equation}
Here the shown identification of the left map holds by \cref{ComparingHomotopyGroupsUpPointedAndUnpointedMappingSpace}, and
the identification of the right map follows by \cref{TheProjectionFromClosedToConstrictedAnnulus} in \cref{ProjectionFromClosedToConstrictedAnnulus}. 
The inner diagram thus shows that away from the unshifted generators, $\partial_1^{\iota^\ast_{\smash{A^2}}}$ coincides with $\partial_1^{\iota^\ast_{\smash{A^2_{\mathrm{cns}}}}}$. But the latter was computed in \cref{IdentifyingPartial1IotaAstA2Cns} and with that our claim \cref{UnwindingHomomorphismForClosedAnnulus} follows. By exactness of \cref{UnwindingLESForIotaA2IntoHC}, this furthermore implies the claim \cref{pi1OfMapsOfBdrOfClosedAnnulusIntoHC}.

On the other hand, $\Unwinding_2$ factors via \cref{MapFromHomotopyLESOfRelMap} as:
\begin{equation}
  \label{FactorizationOfUnwinding2ForClosedAnnulus}
  \begin{tikzcd}[
    row sep=6pt, column sep=large
  ]
    \pi_3\, \UnpointedMap\bracket({
      S^1, S^3
    })^2
    \ar[
      rr,
      uphordown,
      "{ \Unwinding_2 }"{description}
    ]
    \ar[
      r,
      "{ (h_{\mathbb{C}})_\ast }"
    ]
    \ar[
      d,
      equals
    ]
    &
    \pi_3\, \UnpointedMap\bracket({
      S^1, S^2
    })^2
    \ar[
      r,
      "{  \partial^{\iota_{\smash{A^2}}^\ast}_2 }"
    ]
    \ar[
      d,
      equals
    ]
    &
    \pi_2\, \Map\bracket({
      S^2 \vee S^1, S^2
    })
    \\
    \pi_3\bracket({S^3})^2
    \times
    \pi_3\bracket({\Omega S^3})^2
    \ar[
      d,
      equals
    ]
    \ar[
      r,
      "{
        (h_{\mathbb{C}})_\ast
      }"
    ]
    &
    \pi_3\bracket({S^2})^2
    \times
    \pi_3\bracket({\Omega S^2})^2 
    \ar[
      d,
      equals
    ]
    \\
    \mathbb{Z}_u^2 
    \oplus
    \mathbb{Z}_{/2}^2
    \ar[
      r,
      "{ 
        \mathrm{id}
        \oplus
        \mathrm{id}
    }"
    ]
    &
    \mathbb{Z}_u^2 
      \oplus 
    \mathbb{Z}_{/2}^2
    \mathrlap{\,,}
  \end{tikzcd}
\end{equation}
while exactness of the sequence (\cref{MappingFiberSequenceOfCWInclusion})
\begin{equation}
  \begin{tikzcd}[
    column sep=45pt
  ]
    \pi_3\Map\bracket({
      S^2 \vee S^1, S^2
    })^2
    \ar[
      r,
      "{  
        \pi_3(q^\ast_{\smash{A^2}})
      }"{yshift=1pt},
      "{ 
        = 0 
        \,
        \scalebox{.7}{\cref{AnnulusCoprojectionIsNull}}
      }"{swap}
    ]
    &
    \pi_3\UnpointedMap\bracket({S^1, S^2})
    \ar[
      r,
      "{
        \pi_3(
          \iota^\ast_{\smash{
            A^2
          }}
        )
      }"{yshift=1pt},
      "{
        = \Delta
        \,
        \scalebox{.7}{\cref{AnnulusBoundaryIncIsCodiagonal}}
      }"{swap}
    ]
    &
    \pi_3\UnpointedMap\bracket({
      S^1,S^2
    })^2
    \ar[
      dll,
      snake left,
      "{
        \partial_2^{
          \iota^\ast_{\smash{
            A^2
          }}
        }
      }"{description}
    ]
    \\
    \pi_2\Map\bracket({
      S^2 \vee S^1, S^2
    })
    \ar[
      r,
      "{  
        \pi_2(q^\ast_{\smash{A^2}})
      }"{yshift=3pt},
      "{
       = 0
       \,
       \scalebox{.7}{\cref{AnnulusCoprojectionIsNull}}
      }"{swap}
    ]
    &
    \pi_1 \Map\bracket({S^1, S^2})
    \mathrlap{\,,}
  \end{tikzcd}
\end{equation}
shows that $\partial_2^{\iota^\ast_{\smash{A^2}}}$ is surjective (with kernel the diagonal map). This implies that also $\inlinetikzcd{ \mathbb{Z}^2 \oplus \mathbb{Z}_{/2}^2 \ar[r, "{ \Unwinding_2 }"] \&  \mathbb{Z} \oplus \mathbb{Z}_{/2} }$ is surjective (with kernel the diagonal).
Together with the previous result \cref{UnwindingHomomorphismForClosedAnnulus} that $\mathrm{ker}\bracket({\Unwinding_1}) \simeq \mathbb{Z}$, this implies with the exactness of \cref{UnwindingLESForIotaA2IntoHC} the claim \cref{pi2OfMapsOfBdrOfClosedAnnulusIntoHC}.
\end{proof}

\begin{lemma}
\label[lemma]
{ImageOfMapA2S3InMapiotaA2hC}
The map
\begin{equation}
  \begin{tikzcd}
    \UnpointedMap\bracket({
      A^2, S^3
    })
    \underset{\mathclap{
      \scalebox{.7}{\cref{TwRelMapIntoIdentity}}
    }}{\simeq}
    \UnpointedMap\bracket({
      \iota_{A^2},
      \mathrm{id}_{S^3}
    })
    \ar[
      rr,
      "{
        (\mathrm{id},h_{\mathbb{C}})_\ast
      }"
    ]
    &&
    \UnpointedMap\bracket({
      \iota_{A^2},
      h_{\mathbb{C}}
    })    
  \end{tikzcd}
\end{equation}
induces on $\pi_2$ an isomorphism
\begin{equation}
  \label{pi2IdHC}
  \Big(
  \begin{tikzcd}
    \pi_2
    \UnpointedMap\bracket({
      A^2, S^3
    })
    \ar[
      rr,
      "{
        \pi_2
        (
          \mathrm{id},
          h_{\mathbb{C}}
        )_\ast
      }"
    ]
    &&
    \pi_2
    \UnpointedMap\bracket({
      \iota_{A^2},
      h_{\mathbb{C}}
    })    
  \end{tikzcd}
  \Big)
  =
  \big(
  \begin{tikzcd}
    \mathbb{Z}
    \ar[
      r,
      "{
        \sim
      }"
    ]
    &
    \mathbb{Z}
  \end{tikzcd}
  \big)
  \mathrlap{\,.}
\end{equation}
\end{lemma}
\begin{proof}
  The group $\mathbb{Z}$ on the left follows by \cref{HomotopyGroupsOfFreeLoopSpace,HopfDegreeTheorem}, while the group $\mathbb{Z}$ on the right is from \cref{pi2OfMapsOfBdrOfClosedAnnulusIntoHC}. It remains to identify the induced homomorphism between them.

  To this end, consider the commuting diagram
  \begin{equation}
    \begin{tikzcd}[
      column sep=10pt,
      row sep=-3pt
    ]
      \UnpointedMap\bracket({
        A^2,
        S^3
      })
      \underset{\mathclap{
        \scalebox{.7}{\cref{TwRelMapIntoIdentity}}
      }}{\simeq}
      &
      \UnpointedMap\bracket({
        \iota_{A^2},
        \mathrm{id}_{S^3}
      })
      \ar[
        rr,
        "{
          (
            \mathrm{id},
            h_{\mathbb{C}}
          )_\ast
        }"
      ]
      \ar[
        dr,
        "{
          r_{\mathrm{bdr}}
        }"{swap}
      ]
      &[20pt]
      &[20pt]
      \UnpointedMap\bracket({
        \iota_{A^2},
        h_{\mathbb{C}}
      })
       \mathrlap{\,,}
      \ar[
        dl,
        "{
          r'_{\mathrm{bdr}}
        }"
      ]
      \\
      &&
      \UnpointedMap\bracket({
        \partial A^2,
        S^3
      })
    \end{tikzcd}
  \end{equation}
  where the diagonal maps are the boundary restrictions from \cref{ThePastingDiagram}.
  The claim follows by observing that on $\pi_2$ both these boundary restrictions are the diagonal maps $\Delta$ 
  \begin{equation}
    \begin{tikzcd}[
      column sep=70pt,
      row sep=0pt
    ]
      \mathbb{Z}
      \ar[
        dr,
        "{
          \pi_2(r_{\mathrm{bdr}})
          =
          \Delta
        }"{sloped, swap}
      ]
      \ar[
        rr,
        "{
          \pi_2\big(
            (\mathrm{id}, h_{\mathbb{C}})_\ast
          \big)
        }"
      ]
      &&
      \mathbb{Z}
      \mathrlap{\,.}
      \ar[
        dl,
        "{
          \pi_2(r'_{\mathrm{bdr}})
          =
          \Delta
        }"{sloped, swap}
      ]
      \\
      &
      \mathbb{Z}^2
    \end{tikzcd}
  \end{equation}
  For the boundary map on the left this is \cref{AnnulusBoundaryIncIsCodiagonal}, while for the boundary map on the right this follows by exactness of the sequence 
  \cref{HomotopyLESForRelativeMap}
  \begin{equation}
    \begin{tikzcd}[row sep=small]
      {}
      \ar[r, -, dotted]
      &[-20pt] 
      \pi_2 \UnpointedMap\bracket({
        \iota_{A^2},
        h_{\mathbb{C}}
      })
      \ar[
        r,
        "{
          \pi_2(r_{\mathrm{bdr}})
        }"
      ]
      &[+20pt]
      \pi_2\UnpointedMap\bracket({
        \partial A^2,
        S^3
      })
      \ar[
        dll,
        snake left,
        "{ 
          \Unwinding_1 
        }"{description}
      ]
      \\
      \pi_1\Map\bracket({
        S^2 \vee S^1,
        S^2
      })
      \ar[r, -, dotted]
      &
      {}
    \end{tikzcd}
  \end{equation}
  using the form of $\Unwinding_1$ established in \cref{UnwindingHomomorphismForClosedAnnulus}, due to which $\mathrm{ker}\bracket({\Unwinding_1}) \simeq \Delta\bracket({\mathbb{Z}})$.
\end{proof}

\subsection
{The Differentially-Refined Correspondence}

We turn to the differential refinement of the above topological computations in the sense of \cref{OnBndryFieldsInHGT}.

\subsubsection
{Homotopy of the phase space}
\label
{OnHomotopyOfThePhaseSpace}

The following \cref{ShapeOfFQHPhaseSpaceViaRHT} is the key observation that translates the geometric homotopy theory of phase spaces back to pure homotopy theory. The reader not concerned with geometric homotopy theory may take the following \cref{ShapeOfFQHPhaseSpaceViaRHT} as the definition of the homotopy type that we shall be concerned with in the following.

Recall here that $\RRationalization$ \cref{RRationalization} denotes \emph{rationalization} over the real numbers.
\begin{proposition}
\label[proposition]
{ShapeOfFQHPhaseSpaceViaRHT}
  The shape of the FQH phase space \cref{TheFQHPhaseSpace} fits into a homotopy Cartesian square of this form:
  \begin{equation}
    \label{TheShapeOfFQHPhaseSpaceViaRHT}
    \begin{tikzcd}[
      column sep=60pt
    ]
      \shape\FQHPhaseSpace(\phi)
      \ar[r]
      \ar[d]
      \ar[
        dr,
        phantom,
        "{ \lrcorner_h }"{pos=.05}
      ]
      &
      \UnpointedMap\bracket({
        \phi, h_{\mathbb{C}}
      })
      \ar[
        d,
        "{
          \eta^\mathbb{R}_\ast
        }"
      ]
      \\
      \UnpointedMap\bracket({
        \BulkDomain,
        \RRationalization
        S^3
      })
      \ar[
        r,
        "{
          \RRationalization(
            \mathrm{id},h_{\mathbb{C}}
          )_\ast
        }"{}
      ]
      &
      \UnpointedMap\bracket({
        \phi,
        \RRationalization
        h_{\mathbb{C}}
      })      
      \mathrlap{\,,}
    \end{tikzcd}
  \end{equation}
  where the right map is induced by \cref{TheRRationalizationUnit}.
\end{proposition}
\begin{proof}
  This follows by combining results compiled in \cref{OnTheCohesiveModalities}:
\begin{equation}
\begin{alignedat}{2}
    \shape \FQHPhaseSpace(\phi)
    & \defneq
    \shape
    \bracket({
      \mathbf{\Omega}^1_{\mathrm{cl}}
      \bracket({
        \BulkDomain; \mathfrak{l}S^3
      })
      \underset{
        \shape
        \mathbf{\Omega}^1_{\mathrm{cl}}
        \scaledbracket({
          \phi; \mathfrak{l}h_{\mathbb{C}}
        })     
      }
        {\times}
      \UnpointedMap\bracket({
        \phi, h_{\mathbb{C}}
      }) 
    })
    &\;\;&
   \substack{ \text{by \cref{TheFQHPhaseSpace}}}
    \\
    & 
    \sim
    \shape\mathbf{\Omega}^1_{\mathrm{cl}}
    \bracket({
        \BulkDomain; \mathfrak{l}S^3
      })
      \underset{
        \shape
        \mathbf{\Omega}^1_{\mathrm{cl}}
        \scaledbracket({
          \phi; \mathfrak{l}h_{\mathbb{C}}
        })     
      }
        {\times}
      \UnpointedMap\bracket({
        \phi, h_{\mathbb{C}}
      }) 
    &\;\;&
   \substack{ \text{by \cref{ShapePreservingFiberProductsOverDiscrete}}}
    \\
    & 
    \simeq
    \UnpointedMap\bracket({
      \BulkDomain,
      \RRationalization S^3
    })
    \underset{
      \UnpointedMap
      \scaledbracket({
        \phi;
        \RRationalization
        h_{\mathbb{C}}
      })
    }
      {\times}
    \UnpointedMap\bracket({
      \phi,
      h_{\mathbb{C}}
    })
    &\;\;&
   \substack{ \text{by \cref{ShapeOfClosedFormsOnSmoothManifold}.}}
\end{alignedat}
\end{equation}
Finally, that the corresponding maps (over which this fiber product is taken) are as claimed follows by idempotency of the shape operator \cref{IdempotencyOfShapeInIntro}.
\end{proof}
\begin{lemma}
  We have natural isomorphisms
  \begin{equation}
    \label{HomotopyGroupsOfMapsToRRationalization}
    \begin{aligned}
    \pi_n
    \UnpointedMap\bracket({
      \phi,
      \RRationalization 
      h_{\mathbb{C}}
    })
    & 
    \simeq
    \pi_n
    \UnpointedMap\bracket({
      \phi,
      h_{\mathbb{C}}
    })
    \otimes_{{}_{\mathbb{Z}}}
    \mathbb{R}\,.
    \end{aligned}
  \end{equation}
\end{lemma}
\begin{proof}
  This follows generally from the fact that mapping spaces out of finite CW complexes into $\mathbb{R}$-rationalizations of some $X$ are $\mathbb{R}$-rationalizations of the corresponding mapping spaces into $X$. Concretely, one may observe that all the above computations of homotopy groups go through with $h_{\mathbb{C}}$ replaced by $\RRationalization(h_{\mathbb{C}})$ and yield the claimed result. 
\end{proof}

\begin{lemma}
We have
\begin{equation}
\label
{FundamentalGrpOfFQHPhsSpaceInDeepBulkOfDisk}
  \pi_1\bracket({
    \Map\bracket({
      S^2, S^2
    })
    \underset{
      \Map\scaledbracket({
        S^2, 
        \RRationalization S^2
      })
    }
    {\times}
    \Map\bracket({
      S^2, 
      \RRationalization
      S^3
    })
  })
  \simeq
  \mathbb{Z}\,.
\end{equation}
\end{lemma}
\begin{proof}
The Mayer-Vietoris sequence (\cref{MayerVietorisSequence}) gives:
\begin{equation}
  \begin{tikzcd}[row sep=small,
    column sep=16pt
  ]
    \overbrace{
    \pi_2\Map\bracket({
      S^2, 
      \RRationalization
      S^2
    })
    }^{0}
    \ar[r, shorten=-2pt]
    &
    \pi_1\bracket({Z})
    \ar[r, shorten=-2pt]
    &
    \overbrace{
    \pi_1
    \Map\bracket({
      S^2, S^2
    })
    }^{ \mathbb{Z} }
    \times
    \overbrace{
    \pi_1
    \Map\bracket({
      S^2, 
      \RRationalization
      S^3
    })
    }^{ \mathbb{R} }
    \ar[
      dll,
      snake left,
      "{ \partial_1 }"{description}
    ]
    \\
    \underbrace{
    \pi_1\Map\bracket({
      S^2, 
      \RRationalization
      S^2
    })
    }_{ \mathbb{R} }
  \end{tikzcd}
\end{equation}
Here the left term vanishes since $\pi_2\Map\bracket({S^2, S^2}) \simeq \pi_4\bracket({S^2})$ is pure torsion, by  \cref{SecondStableStem}. Moreover, using \cref{ComplexHopfFibrationOnPi3}, the connecting homomorphism is seen to be
\begin{equation}
  \begin{tikzcd}[
    sep=0pt
  ]
    \mathbb{Z}
    \times
    \mathbb{R}
    \ar[
      rr,
      "{ \partial_1 }"
    ]
    &&
    \mathbb{R}
    \\
    \bracket({
      n,r
    })
    &\mapsto&
    n-r
    \mathrlap{\,.}
  \end{tikzcd}
\end{equation}
Therefore exactness of the above sequence implies the claim: $\pi_1(Z) \simeq \mathrm{ker}\bracket({ \partial_1 }) \simeq \mathbb{Z}$.
\end{proof}

\begin{proposition}[Phase space homotopy over \textbf{closed disk}]
\label[proposition]
{PhaseSpaceHomotopyOverClosedDisk}
Over the closed disk, the fundamental group of the shape of the FQH phase space \cref{TheShapeOfFQHPhaseSpaceViaRHT} is trivial:
\begin{subequations}
  \label{ThePhaseSpaceHomotopyOverClosedDisk}
  \begin{align}
    \pi_0 
    \bracket({
      \shape\FQHPhaseSpace(\iota_{D^2})
    })
    &
    \simeq
    \ast
    \mathrlap{\,,}
    \\
    \label{Pi1FQHPhaseSpaceIotaD2}
    \pi_1 
    \bracket({
      \shape\FQHPhaseSpace(\iota_{D^2})
    })
    &
    \simeq
    1
    \mathrlap{\,,}
  \end{align}
\end{subequations}
and the unwinding homomorphism \cref{TheRefinedUnwindingLES} is still an isomorphism:
\begin{equation}
\label
{RefinedUnwinding1ForClosedDisk}
  \begin{tikzcd}[row sep=small,
    column sep=-10pt
  ]
    & &
    H^{-2}\bracket({
      \partial D^2;
      S^3
    })
    \ar[
      dll,
      snake left,
      "{
        \RefinedUnwinding_1
      }"{description}
    ]
    \\
    \tilde H'^{-1}\bracket({
      D^2/\partial D^2;
      S^2
    })
  \end{tikzcd}
  \;\;
  =
  \;\;
  \begin{tikzcd}[row sep=small,
    column sep=40pt
  ]
    &&
    \mathbb{Z}
     \mathrlap{\,.}
    \ar[
      dll,
      "{
        \sim
      }"{sloped}
    ]
    \\
    \mathbb{Z}
  \end{tikzcd}
\end{equation}
\end{proposition}
\begin{proof}
The Mayer-Vietoris sequence (\cref{MayerVietorisSequence}) induced by \cref{TheShapeOfFQHPhaseSpaceViaRHT} gives:
\begin{equation}
  \begin{tikzcd}[row sep=small,
    column sep=8pt
  ]
    \overbrace{
    \pi_{2}
      \UnpointedMap\bracket({
        \iota_{D^2},h_{\mathbb{C}}
      })
     \!\otimes \mathbb{R}
    }^{0}
    \,
    \ar[r, shorten=-2pt]
    &
    \pi_1\bracket({
      \shape\FQHPhaseSpace(\iota_{D^2})
    })
    \ar[r, shorten=-2pt]
    &
    \overbrace{
    \pi_1
      \UnpointedMap\bracket({
        \iota_{D^2}, h_{\mathbb{C}}
      })
    }^{1}
    \times
    \overbrace{
      \pi_1 
      \UnpointedMap\bracket({
        D^2,
        \RRationalization\bracket({S^3})
      })
    }^{1}
    \ar[
      dll,
      snake left
    ]
    \\
    \underbrace{
    \pi_{1}
      \UnpointedMap\bracket({
        \iota_{D^2}, h_{\mathbb{C}}
      })
    \!\otimes\! \mathbb{R}
    }_0
    \ar[r, shorten=-2pt]
    &
    \pi_0\bracket({
      \shape\FQHPhaseSpace(\iota_{D^2})
    })
    \ar[r, shorten=-2pt]
    &
    \underbrace{
      \pi_0
        \UnpointedMap\bracket({
          \iota_{D^2}, h_{\mathbb{C}}
        })
    }_{ \mathbb{Z} }
    \times
    \underbrace{
      \pi_0 
      \UnpointedMap\bracket({
        D^2,
        \RRationalization\bracket({S^3})
      })
    }_{\ast}
    \mathrlap{\,,}
  \end{tikzcd}
\end{equation}
where the last map surjects onto the fiber product
\[
  \overbrace{
  \pi_0\UnpointedMap\bracket({
    \iota_{D^2},
    h_{\mathbb{C}}
  })
  }^{\mathbb{Z}}
  \underset{
    \underbrace{
    \pi_0 \UnpointedMap
    \scaledbracket({
      \iota_{D^2},
      \RRationalization(h_{\mathbb{C}})
    })
    }_{\mathbb{R}}
  }{\times}
  \overbrace{
  \pi_0\UnpointedMap\bracket({
    D^2, \RRationalization(S^3)
  })
  }^{\ast}
  \simeq
  \ast
  \mathrlap{\,.}
\]
Here, over/under the braces we used \cref{PinMapIotaD2hC} and the 2-connectivity of $\RRationalization\bracket({S^3})$. This implies the claim \cref{ThePhaseSpaceHomotopyOverClosedDisk}, by exactness.

It also follows that $\RefinedUnwinding_1$ is surjective, and with \cref{FundamentalGrpOfFQHPhsSpaceInDeepBulkOfDisk} it follows that it is of the shown form $\inlinetikzcd{\mathbb{Z} \ar[r] \& \mathbb{Z}}$. As such it must be an isomorphism, as claimed in \cref{RefinedUnwinding1ForClosedDisk}. 
\end{proof}

\begin{lemma}
\label[lemma]
{HomotopyGroupsOfAdjustedMapsFromClosedAnnulusDeepBulkToS2}
Under $\pi_1$, the projection onto the first fiber factor of $\shape \FQHPhaseSpace\bracket({ A^2/\partial A^2 })$ \cref{TheRefinedHomotopyFiberSequence} is
\begin{equation}
\label
{pi1OfAdjustedMapDeepBulkOfClosedAnnulus}
  \begin{tikzcd}[
    row sep=5pt
  ]
    \pi_1
    \bigg(
    \overbrace{
    \Map\bracket({
      A^2/\partial A^2,
      S^2
    })
    \qquad \;
    \underset
    {\mathclap{
      \Map\scaledbracket({
        A^2/\partial A^2,
        \RRationalization S^2
      })
    }}
    {\times}
    \qquad \;
    \Map\bracket({
      A^2/\partial A^2,
      \RRationalization S^3
    })
    }^{
      Z :=
    }
    \bigg)
    \ar[
      d,
      equals
    ]
    \ar[r]
    &
    \pi_1
    \left(
    \Map\bracket({
      A^2/\partial A^2,
      S^2
    })
    \right)
    \ar[
      d,
      equals
    ]
    \\
    \mathbb{Z}
    \ar[
      r, 
      hook,
      "{
        \left(
        \begin{matrix}
          1
          \\
          0
        \end{matrix}
        \right)
      }"{
        description,
        scale=.8
      }
    ]
    &
    \mathbb{Z}
    \times
    \mathbb{Z}_u
    \mathrlap{\,.}
  \end{tikzcd}
\end{equation}
\end{lemma}
\begin{proof}
Recalling that $A^2/\partial A^2 \sim S^2 \vee S^1$ \cref{DeepBulkOfClosedAnnulus}, and using \cref{SampleComputationUShifted,HomotopyGroupsOfMapsToRRationalization}, the Mayer-Vietoris LES (\cref{MayerVietorisSequence}) gives:
\begin{equation}
  \begin{tikzcd}[row sep=small,
    column sep=12pt
  ]
    \pi_3
      \Map\bracket({
        S^2 \!\vee\! S^1,
        \RRationalization S^2
      })
    \ar[r, shorten=-2pt]
    &
    \pi_2(Z)
    \ar[r, shorten=-2pt]
    &
    \overbrace{
    \pi_2 \Map\bracket({
        S^2 \!\vee\! S^1,
      S^2
    })
    }^{
      \mathbb{Z}_{/2}
      \times 
      \mathbb{Z}_u
    }
    \times
    \overbrace{
    \pi_2 \Map\bracket({
      S^2 \!\vee\! S^1,
      \RRationalization S^3
    })
    }^{
      \mathbb{R}_u
    }
    \ar[
      dll,
      snake left,
      "{ 
        \partial_1 
      }"{description}
    ]
    \\
    \underbrace{
    \pi_2
      \Map\bracket({
        S^2 \!\vee\! S^1,
        \RRationalization S^2
      })
    }_{
      \mathbb{R}_u
    }
    \ar[r, shorten=-2pt]
    &
    \pi_1(Z)
    \ar[r, shorten=-2pt]
    &
    \underbrace{
    \pi_1 \Map\bracket({
      S^2 \!\vee\! S^1,
      S^2
    })
    }_{
      \mathbb{Z} 
        \times 
      \mathbb{Z}_u
    }
    \times
    \underbrace{
    \pi_1 \Map\bracket({
      S^2 \!\vee\! S^1,
      \RRationalization S^3
    })}_{
      \mathbb{R}
    }
    \ar[
      dll,
      snake left,
      "{ \partial_0 }"{description}
    ]
    \\
    \underbrace{
    \pi_1
      \Map\bracket({
        S^2 \!\vee\! S^1,
        \RRationalization S^2
      })
      \mathrlap{\,,}
    }_{
      \mathbb{R}
      \times
      \mathbb{R}_u
    }
  \end{tikzcd}
\end{equation}
where
\begin{equation}
  \begin{tikzcd}[row sep=-3pt,
   column sep=0pt
  ]
    \mathbb{Z}_{/2} 
      \times 
    \mathbb{Z}_u
      \times 
    \mathbb{R}_u
    \ar[
      rr,
      "{ \partial_1 }"
    ]
    &&
    \mathbb{R}_u
    \\
    \bracket({
      [m], n , r
    })
    &\mapsto&
    r - n
    \mathrlap{\,,}
  \end{tikzcd}
  \hspace{.9cm}
  \begin{tikzcd}[
    sep=0pt
  ]
    \mathbb{Z} 
      \times 
    \mathbb{Z}_{u}
      \times 
    \mathbb{R}
    \ar[
      rr,
      "{ \partial_0 }"
    ]
    &&
    \mathbb{R} 
      \times 
    \mathbb{R}_u
    \\
    (m,n, r)
    &\mapsto&
    ( r-m, -n  )
    \mathrlap{\,.}
  \end{tikzcd}
\end{equation}
Hence $\partial_1$ is surjective and the kernel of $\partial_0$ is $\mathrm{ker}\bracket({\partial_0}) \simeq \mathbb{Z} \subset \bracket({\mathbb{Z} \times \mathbb{Z}_u}) \times \mathbb{R}$. By exactness this implies the claim \cref{pi1OfAdjustedMapDeepBulkOfClosedAnnulus}.
\end{proof}

\begin{proposition}
[Phase space homotopy over \textbf{closed annulus}]
\label[proposition]
{PhaseSpaceHomotopyOverAnnulus}
Over the closed annulus, the fundamental group of the shape of the FQH phase space \cref{TheShapeOfFQHPhaseSpaceViaRHT} is trivial:
\begin{equation}
    \label{Pi1FQHPhaseSpaceIotaA2}
    \pi_1 
    \bracket({
      \shape\FQHPhaseSpace(\iota_{A^2})
    })
    \simeq
    1
    \mathrlap{\,,}
\end{equation}
and the refined unwinding homomorphism \cref{TheRefinedUnwindingLES} is
\begin{equation}
\label
{RefinedUnwinding1ForClosedAnnulus}
  \begin{tikzcd}[row sep=small,
    column sep=-10pt
  ]
    & &
    H'^{-2}\bracket({
      \partial A^2;
      S^3
    })
    \ar[
      dll,
      snake left,
      "{
        \RefinedUnwinding_1
      }"{description}
    ]
    \\
    \tilde H'^{-1}\bracket({
      A^2/\partial A^2;
      S^2
    })
  \end{tikzcd}
  \;\;
  =
  \;\;
  \begin{tikzcd}[row sep=small,
    column sep=40pt
  ]
    &&
    \mathbb{Z}^2 
    \mathrlap{\,.}
    \ar[
      dll,
      "{
        (
          +1 \, -1
        )
      }"{description}
    ]
    \\
    \mathbb{Z}
  \end{tikzcd}
\end{equation}
\end{proposition}
\begin{proof}
The Mayer-Vietoris sequence (\cref{MayerVietorisSequence}) induced by \cref{TheShapeOfFQHPhaseSpaceViaRHT} gives:
\begin{equation}
  \begin{tikzcd}[
    column sep=12pt
  ]
    && 
    \overbrace{
    \pi_2
      \UnpointedMap\bracket({
        \iota_{A^2}, h_{\mathbb{C}}
      })
    }^{\mathbb{Z}}
    \times
    \overbrace{
      \pi_2 
      \UnpointedMap\bracket({
        A^2,
        S^3
      })
      \otimes
      \mathbb{R}
    }^{\mathbb{R}}
    \ar[
      dll,
      snake left,
      "{
        \partial_1
        =
        \pi_2
        \big(
          (-)\otimes(\mathbb{Z} \hookrightarrow\mathbb{R})
          ,\,
          (
            \mathrm{id},
            h_{\mathbb{C}}
          )_\ast
        \big)
      }"{description, pos=.62}
    ]
    \\
    \overbrace{
    \pi_{2}
      \UnpointedMap\bracket({
        \iota_{A^2},h_{\mathbb{C}}
      })
     \!\otimes \mathbb{R}
    }^{\mathbb{R}}
    \ar[r, shorten=-2pt]
    &
    \pi_1\bracket({
      \FQHPhaseSpace(\iota_{A^2})
    })
    \ar[r, shorten=-2pt]
    &
    \overbrace{
    \pi_1
      \UnpointedMap\bracket({
        \iota_{A^2}, h_{\mathbb{C}}
      })
    }^{
      \mathbb{Z}_{\mathrlap{u}}
    }
    \times
    \overbrace{
      \pi_1 
      \UnpointedMap\bracket({
        A^2,
        S^3
      })
      \otimes
      \mathbb{R}
    }^{1}
    \ar[
      dll,
      snake left,
      "{
        \partial_0
        =
        \pi_1
        \big(
          (-)\otimes(\mathbb{Z} \hookrightarrow\mathbb{R})
          ,\,
          (
            \mathrm{id},
            h_{\mathbb{C}}
          )_\ast
        \big)
      }"{description, pos=.62}
    ]
    \\
    \underbrace{
    \pi_{1}
      \UnpointedMap\bracket({
        \iota_{A^2}, h_{\mathbb{C}}
      })
    \!\otimes\! \mathbb{R}
    }_{
      \mathbb{R}_{\mathrm{u}}
    }
    \ar[r, shorten=-2pt]
    &
    \pi_0\bracket({\FQHPhaseSpace(\iota_{A^2})})
    \ar[r, shorten=-2pt]
    &
    \underbrace{
      \pi_0
        \UnpointedMap\bracket({
          \iota_{A^2}, 
          h_{\mathbb{C}}
        })
    }_{ \mathbb{Z} }
    \times
    \underbrace{
      \pi_0 
      \UnpointedMap\bracket({
        A^2,
        S^3
      })
      \otimes
      \mathbb{R}
    }_{\ast}
    \mathrlap{\,,}
  \end{tikzcd}
\end{equation}
where over/under the braces we used
\cref{pinOfMapsOfBdrOfClosedAnnulusIntoHC} and the 2-connectivity of $\RRationalization(S^3)$. But since $\partial_1$ is surjective, by \cref{pi2IdHC}, and $\partial_0 : \mathbb{Z}_u \hookrightarrow \mathbb{R}_u$ is clearly injective, this implies the claim \cref{Pi1FQHPhaseSpaceIotaA2}, by exactness of the above sequence.

Now in order to identify the refined unwinding homomorphism, consider the naturality square \cref{NaturalityOfHomotopyLES} of homotopy exact sequences induced by the projection $\shape \rchi$ \cref{PrimedPhaseSpaceInIntroduction} to the left factor in the homotopy fiber products: 
\begin{equation}
\label
{ANaturalitySquare}
  \begin{tikzcd}[
    column sep=25pt
  ]
    \pi_2\bracket({
      \shape
      \PhaseSpace\bracket({
        \partial A^2
        ;
        S^3
      })
    })
    \ar[
      d,
      equals
    ]
    \ar[
      r,
      "{
        \RefinedUnwinding_1
      }"
    ]
    &
    \pi_1\bracket({
      \shape
      \PhaseSpace'\bracket({
        A^2/\partial A^2
        ;
        S^2
      })
    })
    \ar[
      d,
      "{
         \shape\chi_\ast
      }"
    ]
    \\
    \pi_2
    \UnpointedMap\bracket({
      \partial A^2,
      S^3
    })
    \ar[
      r,
      "{ \Unwinding_1 }"
    ]
    &
    \pi_1
    \Map\bracket({
      A^2/\partial A^2,
      S^2
    })
  \end{tikzcd}
  \;\;
  =
  \;\;
  \begin{tikzcd}[
    ampersand replacement=\&,
    column sep=60pt
  ]
    \mathbb{Z}^2
    \ar[
      d,
      equals
    ]
    \ar[
      r,
      "{
        \RefinedUnwinding_1
      }"
    ]
    \&
    \mathbb{Z}
    \ar[
      d,
      hook
    ]
    \\
    \mathbb{Z}^2
    \ar[
      r,
      "{
        \left(
        \begin{matrix}
          + 1 & -1
          \\
          \,0 & \;0
        \end{matrix}
        \right)
      }"{
        description,
        scale=.8
      }
    ]
    \&
    \mathbb{Z}
    \mathrlap{
      \times
      \mathbb{Z}_u
      \mathrlap{\,.}
    }
  \end{tikzcd}
\end{equation}
Here:
\begin{enumerate}
\item
the left vertical map is the identity,  because we have \cref{TheRefinementFibrationForFQH} $\ClassifyingA' = \ClassifyingA$ (both being $S^3$)
and hence 
$\UnpointedMap\bracket({\BoundaryDomain, \RRationalization \ClassifyingA'}) = \UnpointedMap\bracket({\BoundaryDomain, \RRationalization \ClassifyingA})$ 
in the homotopy fiber product on the right of \cref{TheRefinedHomotopyFiberSequence}

\item
the right vertical map is the inclusion of the ``unshifted'' copy of $\mathbb{Z}$, by \cref{pi1OfAdjustedMapDeepBulkOfClosedAnnulus},

\item
the bottom map is as shown, by \cref{UnwindingHomomorphismForClosedAnnulus}.
\end{enumerate}
Therefore the commutativity of \cref{ANaturalitySquare} implies the claim \cref{RefinedUnwinding1ForClosedAnnulus}.
\end{proof}
\begin{remark}
\label[remark]
{WhereTheUnshiftedCopyDisappears}
  So \cref{PhaseSpaceHomotopyOverAnnulus} shows that in passing from the plain relative mapping space to the shape of the phase space, the nontrivial fundamental group from \cref{MappingBoundaryInclusionOfClosedAnnulusIntoComplexHopfFibration} disappears. We may also understand this via \cref{HomotopyGroupsOfAdjustedMapsFromClosedAnnulusDeepBulkToS2}: 
  
  Because, the proof of \cref{MappingBoundaryInclusionOfClosedAnnulusIntoComplexHopfFibration} shows that the fundamental group $\mathbb{Z}$ of $\Map\bracket({\iota_{A^2}}, h_{\mathbb{C}})$ is inherited from the ``unshifted'' copy $\mathbb{Z}_u$ of integers in 
  \[
    \begin{aligned}
    \pi_1
    \Map\bracket({
      S^2 \vee S^1,
      S^2
    })
    &
    \simeq
    \pi_1
    \Map\bracket({
      S^2,
      S^2
    })
    \times
    \pi_1
    \Map\bracket({
      S^1,
      S^2
    })
    \\
    & 
    \simeq
    \mathbb{Z} \times \mathbb{Z}_u
    \mathrlap{\,.}
    \end{aligned}
  \]
  But \cref{HomotopyGroupsOfAdjustedMapsFromClosedAnnulusDeepBulkToS2} shows that exactly this unshifted copy $\mathbb{Z}_u$ is removed by passage from the plain mapping space to the shape of the phase space.
\end{remark}

\section
{Geometric Engineering on M-Branes}
\label
{OnGeometricEngineeringOnMBraneProbes}

Finally, we discuss how the structures seen above arise when \emph{geometrically engineering} effective FQH field theory on M5-brane probes of 11D Supergravity, globally completed by electromagnetic flux quantization.

The key observation here is that the flux-quantization for M-strings on M5-branes on $\mathrm{A}_1$-singularities in 11D Sugra that is considered in \cite{SS25-Seifert,SS25-Srni,BaSS26-MString} leads exactly to the phase space considered in \cref{FQHPhaseSpace}, with the M-string inclusion into the orbi-fixed transversal part of the M5-brane playing the role of the boundary domain inclusion \cref{BoundaryInclusionInIntroduction}.

In more detail: We are considering this brane configuration in the \emph{OM theory limit} \cite{nLab:NncommutativeOpenStringTheory}, where the angle vanishes between the M5-brane and the M2-brane that is incident on it (ending at an M-string), hence where the M2-brane lies entirely inside the M5:
This way the otherwise perpedicular direction of the M2 may be included in the light cone common to the M5 brane and the bulk, in that we have the following brane configuration:
\begin{equation}
  \begin{tikzcd}[
    column sep=-5pt
  ]
   {}
   \\[-5pt]
    \substack{\text{
      \color{darkblue}
      M2-brane
    }}
    &
    \overset{
      \mathclap{
        \adjustbox{
          rotate=30,
          scale=.75
        }
        {\rlap{
          \color{darkblue}
          lightcone
        }}
      }
    }{
    \mathbb{R}^{1,1}
    }
    &
    \times
    \ar[
      d,
      hook,
      "{ \phi }"
    ]
    &
    \BoundaryDomain^1
    \ar[
      rrrrr,
      dashed
    ]
    &&&&
    &[40pt]
    S^7
    \ar[
      d,
      ->>,
      "{
        h_{\mathbb{C}}
      }"
    ]
    &[-11pt]
    \phantom{S}
    \ar[
      in=58-90,
      out=180-58-90,
      looseness=4,
      shorten <=1pt,
      shorten >=-1pt,
      shift right=3pt,
      "{
        \,\mathclap{G}\,
      }"{description}
    ]
    \\
    \substack{\text{
      \color{darkblue}
      M5-brane
    }}
    &
    \mathbb{R}^{1,1} 
    &
    \times
    \ar[
      d,
      hook,
      "{ \Phi }"
    ]
    &
    \BulkDomain^2
    &
    \times
    &
    \mathbb{C}
    \ar[
      in=60,
      out=180-60,
      looseness=4,
      shift right=2pt,
      "{
        \,\mathclap{G}\,
      }"{description}
    ]
    \ar[
     rrr,
     dashed,
     shorten >=6pt
    ]
    &&
    & 
    \mathllap{\mathbb{C}}P^3
    \ar[
      d,
      ->>,
      "{
        t_{\mathbb{C}}
      }"
    ]
    &
    \phantom{S}
    \ar[
      in=58-90,
      out=180-58-90,
      looseness=4,
      shorten <=1pt,
      shorten >=-1pt,
      shift right=3pt,
      "{
        \,\mathclap{G}\,
      }"{description}
    ]
    \\
    \substack{\text{
      \color{darkblue}
      11D bulk
    }}
    &
    \mathbb{R}^{1,1}
    &\times&
    \mathbb{R}^5
    &\times&
    \mathbb{C}
    \ar[
      in=60,
      out=180-60,
      looseness=4,
      shift right=2pt,
      "{
        \,\mathclap{G}\,
      }"{description}
    ]
    &
    \underset{\mathclap{
      \substack{\text{
        \color{darkblue}
        $\mathrm{A}_1$-orbifold
      }}
    }}
      {\times}
    &
    {\mathbb{C}}
    \ar[
      in=60,
      out=180-60,
      looseness=4,
      shift right=2pt,
      "{
        \,
        \mathclap{
          G^{\mathrlap{\mathrm{op}}}
        }
        \,
      }"{description}
    ]
    \ar[
      r,
      dashed
    ]
    &
    S^4
    &
    \phantom{S}
    \ar[
      in=58-90,
      out=180-58-90,
      looseness=4,
      shorten <=1pt,
      shorten >=-1pt,
      shift right=3pt,
      "{
        \,\mathclap{G}\,
      }"{description}
    ]
  \end{tikzcd}
\end{equation}

(Here $G$ is a finite cyclic group acting as a subgroup of $\mathrm{SU}(2) \simeq \mathrm{Sp}(1)$ on $\mathbb{C}^2 \simeq_{\mathbb{R}} \mathbb{H}$, thus making an A-type orbifold.)

Since the orbifold factors here are equivariantly contractible, the twisted relative mapping spaces of dashed classifying maps, in the above diagram, are equivalent to the corresponding maps at their $G$-fixed locus:
\begin{equation}
  \begin{tikzcd}[
    column sep=-5pt
  ]
    \substack{
      \text{\color{darkblue}M2-brane}
    }
    &
    \mathbb{R}^{1,1}
    &
    \times
    \ar[
      d,
      hook,
      "{ \phi }"
    ]
    &
    \BoundaryDomain^1
    \ar[
      rrrrr,
      dashed
    ]
    &&&&
    &[40pt]
    S^3
    \ar[
      d,
      ->>,
      "{
        h_{\mathbb{C}}
      }"
    ]
    &[-11pt]
    \\
    \substack{
      \text{\color{darkblue}
        M5-brane
      }
      \\
      \text{\color{darkblue}
        singularity
      }
    }
    &
    \mathbb{R}^{1,1} 
    &
    \times
    \ar[
      d,
      hook,
      "{ \Phi }"
    ]
    &
    \BulkDomain^2
    &
    &
    \ar[
     rrr,
     dashed,
    ]
    &&
    & 
    S^2
    \ar[
      d,
      ->>
    ]
    &
    \\
    \substack{
      \text{\color{darkblue}7D $\mathrm{A}_1$-}
      \\
      \text{\color{darkblue}singularity}
    }
    &
    \mathbb{R}^{1,1}
    &
      \times
    &
    \mathbb{R}^5
    \ar[
      rrrrr,
      dashed
    ]
    &
    &
    &
    &
    &
    \ast
  \end{tikzcd}
\end{equation}

By the rules of topological light cone quantization (\cite[\S 2]{SS25-Complete}, following \cite{SS24-Obs}), this makes the topological quantum states be representations of 
\begin{equation}
  \pi_0
  \bracket({
    \Map\bracket({
      \mathbb{R}^{1}_{\cpt}
      \wedge
      \phi_{\plus}
      , 
      h_{\mathbb{C}}
    })    
  })
  \simeq
  \pi_1
  \UnpointedMap\bracket({
    \phi, h_{\mathbb{C}}
  })
\end{equation}
But this coincides just with the topologically ordered states \cref{BoundedOrderRepresentation} implied by the choice of classifying fibration $\wp \defneq h_{\mathbb{C}}$ which we derived for FQH systems in \cref{OnDeterminingTheNearBdrClassfFib}.

Better yet, the actual flux density present on the M5-brane has characteristic $L_\infty$-fibration
\begin{equation}
  \begin{tikzcd}[
    column sep=-5pt
  ]
    \mathfrak{l}
    \bracket({
    S^7 
    \rotatebox[origin=c]{-90}{$
      \curvearrowright
    $}
    \,G
    })
    \ar[
      d,
      ->>,
      "{ \mathrm{id} }"
    ]
    \\
    \mathfrak{l}
    \bracket({
    S^7 
    \rotatebox[origin=c]{-90}{$
      \curvearrowright
    $}
    \,G
    })
    \ar[
      d,
      ->>,
      "{
        \mathfrak{l}
        \bracket({
        h_{\mathbb{C}} 
        \rotatebox[origin=c]{-90}{$
          \curvearrowright
        $}
        \,G
        })
          }"
    ]
    \\
    \mathfrak{l}
    \bracket({
    S^4
    \rotatebox[origin=c]{-90}{$
      \curvearrowright
    $}
    \,G
    })
  \end{tikzcd}
  \;\;\;\;\;\;
  \text{\small with fixed locus}
  \;\;\;\;\;\;
  \begin{tikzcd}
    \mathfrak{l}S^3
    \ar[
      d,
      ->>,
      "{ \mathrm{id} }"
    ]
    \\
    \mathfrak{l}S^3
    \ar[
      d,
      ->>
    ]
    \\
    \ast
    \mathrlap{\,,}
  \end{tikzcd}
\end{equation}
reflecting the fact that on-shell non-vanishing is only a 3-flux $H_3$, while the 2-flux $F_2$ and 1-flux $H_1$ characterized by $\mathfrak{l}h_{\mathbb{C}}$ are Chern-Simons type and vanish on-shell (\parencites[p. 7]{SS25-Seifert}[(4.3)]{BaSS26-MString}).

But this means that we are dealing with the refinement map exactly as in \cref{TheRefinementFibrationForFQH} and that 
the actual phase space of this brane configuration is exactly as in \cref{FQHPhaseSpace}! 

In this sense, the above brane configuration ``geometrically engineers'' (\cite{Duplij2017,nLab:GeometricEngineering}) the topological bulk/boundary quantum observables discussed in the main text, and with it the bulk-boundary correspondence for FQH systems as discussed in \cref{OnResults}. 

\section
{Conclusion}
\label
{Conclusions}

We have laid out (in \cref{OnABNCFOrTopologicalOrders}) a new mathematical formulation of bulk-boundary correspondence (BBC), applicable to strongly interacting systems whose classical bulk/boundary parameter space (over which quantum ground states are adiabatically parameterized) admits a classifying fibration $\wp$. In contrast to existing formulations, this new proposal mathematically resolves the crucial distinction between the total system (which cannot as such be gapped if its boundary is ungapped) and its \emph{deep bulk} that may remain gapped and topologically ordered.

For the case of fractional quantum Hall (FQH) systems, we narrowed down (in \cref{OnDeterminingTheNearBdrClassfFib}) the choice of $\wp$ essentially to the complex Hopf fibration, $\wp \defneq h_{\mathbb{C}}$, by demanding that over the closed disk, $\BulkDomain \defneq D^2$, the boundary and hence the total system be ungapped, while topological order with the usual abelian Chern-Simons/FQH anyon braiding phases is seen in the deep bulk (via results recalled in \cref{OnRecallingUnboundedFQHOrders}).  

Using this choice of classifying fibration, $h_{\mathbb{C}}$, over the constricted annulus, $\BulkDomain \defneq A^2_{\mathrm{cns}}$, we found (in \cref{OnClassifyingBoundedTopologicalOrders}) a bulk-boundary correspondence which exhibits the unwinding of the deep bulk FQH order over the two edges, with opposite signs, as expected in view of the opposite ungapped chiral edge currents characteristic of these configurations.

Over the closed annulus, $\wp \defneq A^2$, however, the analogous purely topological analysis with classifying fibration $\wp \defneq h_{\mathbb{C}}$ superficially predicts nontrivial monodromy (further in \cref{OnClassifyingBoundedTopologicalOrders}), contrary to the physical condition that with its boundary also the total system is gapped and hence not topologically ordered.

In resolution of this situation, we switched gears (in \cref{OnBndryFieldsInHGT}) from plain to geometric homotopy theory (cohesive $\infty$-topos theory, \cref{OnSomeGeometricHomotopy}) in order to give a more fine-grained construction of the actual phase space of the effective dynamics, modeled as an effective higher gauge theory of Maxwell/Chern-Simon type. Such refinement depends on the further datum of a cover $\mathfrak{l}\wp'$ of the image $\mathfrak{l}(-)$ of the classifying fibration $\wp$ under $\mathbb{R}$-rationalization, concretely modeled by its relative Whitehead $L_\infty$-algebra $\mathfrak{l}\wp$: Where $\mathfrak{l}\wp$ characterizes the phase space structure of Maxwell-type effective higher gauge field fluxes, the cokernel of the refinement map $\inlinetikzcd{\mathfrak{l}\wp' \ar[r, "{ \prime }"] \& \mathfrak{l}\wp }$ characterizes those among these which satisfy Chern-Simons type equations of motion in that the flux densities vanish on-shell.

Here we found (in \cref{OnIdentifyingEdgeCurrentsInTEDCoh}) that the choice $\prime \defneq \mathfrak{l}(\mathrm{id}, h_{\mathbb{C}}) : \inlinetikzcd{ \mathfrak{l}\mathrm{id}_{S^3} \ar[r] \& \mathfrak{l}h_{\mathbb{C}} }$ removes the spurious order over the closed annulus, while retaining the situation over the closed disk as well as the BBCs in both cases (\cref{OnTheHomotopyOfThePhaseSpace}). While this is just a mathematical fact in itself, its physical relevance is corroborated by analysis of the structure of the gauge potentials in the corresponding effective gauge theory: 

Remarkably, these turn out (\cref{OnIdentifyingThePhaseSpace}) to be of the form known from Floreanini-Jackiw/Wess-Zumino-Witten boundary field theory of abelian Chern-Simons theory in the usual Lagrangian formulation, as traditionally used for effective FQH field theory. Our non-Lagrangian construction (cf. \cite{SS25-ISQS29}) hence reproduces the expected form of the FQH edge currents while sidestepping notorious conceptual issues with Lagrangian Chern-Simons effective field theory applied to this situation (cf. \cite[Rem. A.1]{SS25-FQH}).

This result is the main conclusion drawn here: 
\begin{standout}
With a BBC for strongly interacting systems formulated in generality, the bulk/edge dynamics of FQH anyons neatly fits in as an effective higher gauge theory with refined classifying fibration $(\mathrm{id}_{S^3},h_{\mathbb{C}})$.
\end{standout}

To put this in perspective, we observed (in \cref{OnGeometricEngineeringOnMBraneProbes}) that this is most curious, because just this kind of higher gauge theory appears in what is traditionally advertised as ``high energy physics'', namely on suitable M-brane probes of A-type singularities in 11-dimensional supergravity --- if the latter is globally completed by flux-quantization in twisted Cohomotopy (a completion that in turn is motivated from expectations about the conjectural \emph{M-theory}, cf. \cite{FSS20-H,FSS21-Hopf}). 

While it is a well-appreciated phenomenon that the effective dynamics of strongly coupled quantum systems may be embedded into that of branes fluctuating in higher dimensional auxiliary orbifold super-spacetimes, known as \emph{geometric engineering} of quantum field theories, here we see this idea realized in a novel refined form, where the global topological effects are not just a sketchy afterthought to local Lagrangian analysis, but are natively captured by globally completed effective higher gauge theory (cf. \cite{SS25-Complete}). 

It appears plausible that this geometric engineering of effective bulk/boundary FQH theory on M-brane configurations may shed new light on subtle aspects of FQH dynamics of contemporary interest. Because, remarkably, both 
\emph{$W_\infty$-symmetry}
as well as 
\emph{supersymmetry},
which famously characterize M-brane dynamics, have come to be thought to control the collective excitation modes of FQH liquids (cf. \cite{SS26-SDiff}). 
While this topic is now receiving increased attention, a comprehensive theoretical grasp of the situation is still outstanding.

\appendix

\section
{Background}
\label{OnBackground}

For ease of reference, we compile some background facts that are used in the main text.

\subsection
{Some General Topology}
\label{OnSomeGeneralTopology}

For suitable background on general topology cf. \parencites[\S 2]{James1984}[\S 2]{Schwarz1994}[\S 1-2]{Eschrig2011}.
We work in the category $\Top$ of pointed compactly-generated topological spaces with continuous pointed maps between them (just called \emph{maps}, for short). We will also refer to the \emph{opposite category} $({\Top})^{\mathrm{op}}$, whose objects are the same, but whose morphisms are such maps regarded as morphism going in the opposite direction.

We denote the singleton space by $\ast \in \Top$. This is both \emph{initial} and \emph{terminal}, in that for every $X \in \Top$ there are \emph{unique} maps $\inlinetikzcd{ \ast \ar[r] \& X \ar[r] \& \ast }$.

A particular map that plays a key role in the main text:
\begin{example}
  The \emph{complex Hopf fibration} (cf. \parencites{Lyons2003}[\S 3.2.3]{SS25-Orient}) 
  is the map from the 3-sphere to the 2-sphere,
  \begin{equation}
    \label{TheComplexHopfFibration}
    \begin{tikzcd}[
      row sep=-2pt
    ]
      S^3
      \ar[
        r,
        "{ h_{\mathbb{C}} }"
      ]
      \ar[
        d,
        equals
      ]
      &
      S^2
      \ar[
        d,
        equals
      ]
      \\[8pt]
      \bracket({
        \mathbb{C}^2 \setminus \{0\}
      })/ \mathbb{R}^\times_+
      \ar[r]
      &
      \bracket({
        \mathbb{C}^2 \setminus \{0\}
      })/ \mathbb{C}^\times
      \\
      (z_1, z_2) 
      \mod \mathbb{R}^\times_+
      \ar[
        r,
        |->,
        shorten=5pt
      ]
      &
      (z_1, z_2) 
      \mod \mathbb{C}^\times
      \mathrlap{\,,}
    \end{tikzcd}
  \end{equation}
  given by sending $\mathbb{R}_+$-lines in $\mathbb{C}^2$ to the $\mathbb{C}$-lines which they span, under the canonical inclusion $\mathbb{R}^\times_+ \subset \mathbb{C}^\times$ of the multiplicative group of positive real numbers into that of non-zero complex numbers.
\end{example}

A particular role in our discussion is played by spaces of maps:
\begin{definition}
\label{TheMappingSpaces}
Given pointed CW-complexes $(X,Y)$, write $\UnpointedMap(X,Y)$ for the space of maps from $X$ to $Y$, and
\begin{equation}
\label{ThePointedMappingSpace}
  \mathrm{Map}^{\!\ast}({X,Y})
  \subset
  \mathrm{Map}({X,Y})
\end{equation}
for the subspace of the general mapping space on the
basepoint preserving map. This space is itself pointed by the map constant on the basepoint of $Y$. 

As such, this construction extends to a functor:
  \begin{equation}
    \label{MappingSpaceFunctor}
    \begin{tikzcd}[
      row sep=-2pt, column sep=35pt
    ]
    (\Top)^{\mathrm{op}}
    \times
    \Top
    \ar[
      r,
      "{ 
        \Map(-,-)
      }"
    ]
    &
    \Top
    \\
    (X,A)
    \ar[
      d, 
      shift left=5pt,
      "{ f }"
    ]
    \ar[
      r,
      |->, 
      shorten=4pt
    ]
    &
    \Map(X,A)
    \ar[
      d,
      "{
        g^\ast f_\ast
      }"
    ]
    \\[15pt]
    (Y,B)
    \ar[
      u, 
      shift left=5pt,
      "{ g }"
    ]
    \ar[
      r,
      |->, 
      shorten=4pt
    ]
    &
    \Map(Y,B)
    \mathrlap{\,.}
    \end{tikzcd}
  \end{equation}
\end{definition}

For example, for $X \in \Top$ its \emph{based loop space} and \emph{free loop space} are, respectively:
\begin{equation}
  \label{FreeAndBasedLoopSpace}
  \Omega X
  :=
  \Map\bracket({S^1,X})
  \,,
  \;\;\;
  \mathcal{L}X
  :=
  \UnpointedMap\bracket({S^1,X})
  \mathrlap{\,.}
\end{equation}

\subsection
{Some Category Theory}

For suitable background on category theory see \parencites[\S 1]{James1984}[\S 2]{Geroch1985}{Awodey2010}.

A commuting square of maps is called \emph{Cartesian} (or a \emph{pullback}), denoted by a hook ``$\lrcorner$'' in its top left corner, if it uniquely factors all commuting square completions of its bottom and right maps:
\begin{equation}
  \label{CartesianSquare}
  \begin{tikzcd}
    {} 
      \ar[r] 
      \ar[d]
    \ar[
      dr,
      phantom,
      "{ \lrcorner }"{pos=.1}
    ]
    & {} \ar[d]
    \\
    {} \ar[r] & {}
  \end{tikzcd}
  \;\;\;
  \Leftrightarrow
  \;\;\;
  \begin{tikzcd}
    {}
      \ar[
        drr,
        bend left=20,
        "{ \forall }"{description}
      ]
      \ar[
        ddr,
        bend right=20,
        "{ \forall }"{description}
      ]
    \ar[
      dr,
      dashed,
      "{ \exists! }"{description}
    ]
    &[-8pt]
    \\[-3pt]
    &
    {} 
      \ar[r] 
      \ar[d]
    & {} \ar[d]
    \\
    &
    {} \ar[r] & {}
    \mathrlap{\,.}
  \end{tikzcd}
\end{equation}
Dually, a commuting square is called \emph{co-Cartesian} (or a \emph{pushout}), denoted by a dual hook ``$\ulcorner$'' in its bottom right corner, if it is Cartesian when seen in the opposite category $\bracket(\Top)^{\mathrm{op}}$.

\begin{example}
Given a map \inlinetikzcd{X \ar[r] \& Y}, then:
\begin{enumerate}
  \item its \emph{fiber} is the subspace $\inlinetikzcd{F \ar[r, hook] \& X}$, in that:
  \begin{equation}
    \label{Fiber}
    \begin{tikzcd}[row sep=small, column sep=large]
      F 
        \ar[d]
        \ar[r]
        \ar[
          dr,
          phantom,
          "{ \lrcorner }"{pos=.1}
        ]
      & X \ar[d]
      \\
      \ast \ar[r] & Y
      \mathrlap{\,,}
    \end{tikzcd}
  \end{equation}
  \item 
  its \emph{cofiber} 
  is the quotient space $\inlinetikzcd{Y \ar[r, ->>] \& Q = Y/X}$ in that
  \begin{equation}
    \label{Cofiber}
    \begin{tikzcd}[row sep=small, column sep=large]
      X 
        \ar[d]
        \ar[r]
        \ar[
          dr,
          phantom,
          "{ \ulcorner }"{pos=.9}
        ]
      & Y \ar[d]
      \\
      \ast \ar[r] & Q
      \mathrlap{\,.}
    \end{tikzcd}
  \end{equation}
\end{enumerate}
\end{example}
\begin{example}
  Given $X,Y \in \Top$, 
  \begin{enumerate}
  \item
  the pullback of their joint basepoint projections is their \emph{product space} $X \times Y$,
  \item
  the pushout of their joint basepoint inclusions is their \emph{wedge sum} $X \vee Y$ (the result of identifying the basepoints in their disjoint union $X \sqcup Y$):
  \end{enumerate}
  \begin{equation}
    \label{ProductSpaceAndWedgeSum}
    \begin{tikzcd}[row sep=small, 
      column sep=15pt
    ]
      X \times Y
      \ar[r]
      \ar[d]
      \ar[
        dr,
        phantom,
        "{ \lrcorner }"{pos=.1}
      ]
      &
      X
      \ar[d]
      \\
      Y
      \ar[r]
      &
      \ast
      \mathrlap{\,,}
    \end{tikzcd}
    \hspace{1.5cm}
    \begin{tikzcd}[row sep=small, 
      column sep=20pt
    ]
      \ast 
        \ar[r]
        \ar[d]
        \ar[
          dr,
          phantom,
          "{ \ulcorner }"{pos=.9}
        ]
      & 
      X
      \ar[d]
      \\
      Y
      \ar[r]
      & 
      X \vee Y
      \mathrlap{\,.}
    \end{tikzcd}
  \end{equation}
\end{example}

\begin{lemma}[Pasting law]
\label[lemma]{PastingLaw}
For a commuting diagram of maps of the form
\[
  \begin{tikzcd}[row sep=small]
    {} \ar[r] \ar[d] 
    & {} \ar[r] \ar[d]
    \ar[dr, phantom, "{\lrcorner}"{pos=.1}]
    & {} \ar[d]
    \\
    {} \ar[r] & {} \ar[r] & {}    
  \end{tikzcd}
\]
if the right square is Cartesian \cref{CartesianSquare}, as indicated, then the left square is Cartesian iff the total rectangle is.
\end{lemma}

\begin{lemma}
  \label[lemma]{MapPreservesLimits}
  The mapping space construction \cref{MappingSpaceFunctor} 
  preserves Cartesian squares in each argument separately, hence:
  \begin{subequations}
    \begin{align}
    \label{MapTakingPullbackInSecondArgumentToPullbackI}
    \Map\left(
      X,\;
      \begin{tikzcd}[row sep=small,
        ampersand replacement=\&
      ]
        A \ar[r] \ar[d] 
        \ar[
          dr,
          phantom, 
          "{ \lrcorner }"{pos=.1}
        ] 
        \& B \ar[d]
        \\
        C \ar[r] \& D
      \end{tikzcd}
    \right)
    \simeq
      \begin{tikzcd}[
        ampersand replacement=\&
      ]
        \Map(X,A) \ar[r] \ar[d] 
        \ar[
          dr,
          phantom, 
          "{ \lrcorner }"{pos=.1}
        ] 
        \& 
        \Map(X,B) \ar[d]
        \\
        \Map(X,C) \ar[r] 
        \& 
        \Map(X,D)
        \mathrlap{\,,}
      \end{tikzcd}    
      \\
    \label{MapTakingPushoutInFirstArgumentToPullback}
    \Map\left(
      \begin{tikzcd}[row sep=small, 
        ampersand replacement=\&
      ]
        X \ar[r, <-] \ar[d, <-] 
        \ar[
          dr,
          phantom, 
          "{ \lrcorner }"{pos=.1}
        ] 
        \& Y \ar[d, <-]
        \\
        Z \ar[r, <-] \& W
      \end{tikzcd}
      \;
      ,
      A
    \right)
    \simeq
      \begin{tikzcd}[
        ampersand replacement=\&
      ]
        \Map(X,A) \ar[r] \ar[d] 
        \ar[
          dr,
          phantom, 
          "{ \lrcorner }"{pos=.1}
        ] 
        \& 
        \Map(Y,A) \ar[d]
        \\
        \Map(Z,A) \ar[r] 
        \& 
        \Map(W,A)
        \mathrlap{\,.}
      \end{tikzcd}    
    \end{align}
  \end{subequations}
\end{lemma}
\begin{example}
  Combined with \cref{ProductSpaceAndWedgeSum} this yields:
  \begin{equation}
    \label{PointedMapsOutOfWedgeSum}
    \Map\bracket({
      X \vee Y, Z
    })
    \simeq
    \Map\bracket({X,Z})
    \times
    \Map\bracket({Y,Z})
    \mathrlap{\,.}
  \end{equation}
\end{example}

\subsection
{Some Homotopy Theory}

For $X \in \UnpointedTop$ we write $\pi_0(X)$ for its set of path-connected components. Applied to mapping spaces \cref{ThePointedMappingSpace}, these are the \emph{homotopy classes} of maps:
\begin{equation}
  \label{HomotopyClass}
  \begin{tikzcd}[
    row sep=0pt
  ]
    \Map\bracket({X,A})
    \ar[
      r,
      ->>
    ]
    &
    \pi_0
    \Map\bracket({X,A})
    \\
    f 
    \ar[
      r,
      |->,
      shorten=8pt
    ]
    &
    {[f]}
    \mathrlap{\,.}
  \end{tikzcd}
\end{equation}
If here we think of $A$ as the classifying space of a nonabelian cohomology theory \cite[\S 2]{FSS23-Char}, then we may abbreviate
\begin{equation}
  \label{CohomologyAsHomotopyClasses}
  \tilde H^0\bracket({
    X; A
  })
  :=
  \pi_0\, \Map\bracket({
    X, A
  })
  \mathrlap{\,.}
\end{equation}

\begin{definition}[Homotopy groups]
Set $I := [0,1] \in \UnpointedTop$. We take the $n$th homotopy group $(n \in \mathbb{N}_{\geq 1})$ of $X \in \Top$ to be given by homotopy classes \cref{HomotopyClass} by maps $\inlinetikzcd{I^n \ar[r] \& X}$ whose restriction to $\partial\bracket({I^n})$ is constant on the base point, hence:
\begin{equation}
  \label{HomotopyGroup}
  \pi_n(X)
  \defneq
  \pi_0 \, 
  \Map\bracket({
    I^n / \partial I^n,
    X
  })
  \mathrlap{\,.}
\end{equation}
The group operations are represented by concatenation and reversion of loops, respectively.

This construction is clearly functorial, taking maps $f$ to the group homomorphism obtained by postcomposing representatives:
\begin{equation}
\label{FunctorialityOfHomotopyGroups}
  \begin{tikzcd}[
    column sep=0pt,
    row sep=0pt
  ]
    \Top
    \ar[
      rr,
      "{ \pi_n }"
    ]
    &&
    \mathrm{Grp}
    \\
    X
    \ar[
      dd,
      "{ f }"
    ]
    &&
    \pi_n(X)
    \ar[
      dd,
      "{
        \pi_n(f)
      }"{description}
    ]
    &
    {[\gamma]}
    \ar[
      dd,
      |->,
      shorten=5pt
    ]
    \\[7pt]
    &
    \mapsto
    \\[7pt]
    Y
    &&
    \pi_n(Y)
    &
    \bracket[{
      f \circ \gamma
    }]
    \mathrlap{\,.}
  \end{tikzcd}
\end{equation}
\end{definition}
\begin{remark}
In particular,
\begin{equation}
  \label{IterationOfHomotopyGroups}
  \pi_{n+1}(X)
  \simeq
  \pi_0 \, \Map\bracket({ S^{n+1}, X })
  \simeq
  \pi_1 \, \Map\bracket({ S^n, X })
  \mathrlap{\,.}
\end{equation}
Also notice that
\begin{equation}
  \label{PiNPreservesProducts}
  \pi_n(X \times Y)
  \simeq
  \pi_n(X)
  \times 
  \pi_n(Y)
  \mathrlap{\,.}
\end{equation}
\end{remark}
\begin{example}[Hopf degree theorem, cf. {\cite[\S IX 5.8]{Kosinski1993}}]
  For $n \in \mathbb{N}_+$ we have
  \begin{equation}
    \label{HopfDegreeTheorem}
    \begin{aligned}
      \pi_0\, 
      \UnpointedMap\bracket({
        S^n, S^n
      })
      &
      \simeq
      \mathbb{Z}
      \mathrlap{\,,}
      \\
      \pi_0\, 
      \Map\bracket({
        S^n, S^n
      })
      \simeq
      \pi_n\bracket({S^n})
      &
      \simeq
      \mathbb{Z}
      \mathrlap{\,.}
    \end{aligned}
  \end{equation}
\end{example}
\begin{example}[Hopf 1931]
  The remarkable property of the complex Hopf fibration $h_{\mathbb{C}}$ \cref{TheComplexHopfFibration} is that it freely generates the 3rd homotopy group of the 2-sphere, and in fact that it identifies all the higher homotopy groups of these spheres:
  \begin{equation}
    \label{ComplexHopfFibrationOnPi3}
    \forall_{n \geq 3}
    :
    \;\;\;
    \begin{tikzcd}[column sep=large]
      \pi_n\bracket({S^3})
      \ar[
        r,
        "{ \pi_n(h_{\mathbb{C}}) }",
        "{ \sim }"{swap}
      ]
      &
      \pi_n\bracket({S^2})
      \,,
    \end{tikzcd}
    \;\;\;\text{in particular}:\;
    \begin{tikzcd}[
      row sep=4pt
    ]
      \pi_3\bracket({S^3})
      \ar[
        d,
        equals
      ]
      \ar[
        rr,
        "{ \pi_3(h_{\mathbb{C}}) }"
      ]
      &&
      \pi_3\bracket({S^2})
      \ar[
        d,
        equals
      ]
      \\
      \mathbb{Z}
      \ar[
        rr,
        "{ \sim }"
      ]
      &&
      \mathbb{Z}
      \mathrlap{\,.}
    \end{tikzcd}
  \end{equation}
\end{example}
\begin{proof}
  Using that 
  \begin{enumerate}
    \item $h_{\mathbb{C}}$ is a Serre fibration (\cref{SerreFibration}) with fiber $S^1$,
    \item $\pi_{\geq 2}\bracket({S^1}) \simeq 0$,
  \end{enumerate}
  this follows by exactness of the associated homotopy LES (\cref{HomotopyLES}):
  \begin{equation}
    \forall_{n \geq 3} :
    \;\;\;
    \begin{tikzcd}[sep=20pt]
      \underbrace{
        \pi_n\bracket({S^1})
      }_{ 0 }
      \ar[r]
      &
      \pi_n\bracket({S^3})
      \ar[
        rr,
        "{ \pi_n(h_{\mathbb{C}}) }",
      ]
      &&
      \pi_n\bracket({S^2})
      \ar[r]
      &
      \underbrace{
        \pi_{n-1}\bracket({S^1})
      \mathrlap{\,.}
      }_{ 0 }
    \end{tikzcd}
    \qedhere
  \end{equation}
\end{proof}
\begin{example}[Pontrjagin 1936-, {\cite{Whitehead1950}}]
  For $n \geq 2$ we have:
  \begin{equation}
    \label{SecondStableStem}
    \pi_{n+2}\bracket({S^{n}}) 
      \simeq 
    \mathbb{Z}_{/2}
    \mathrlap{\,.}
  \end{equation}
\end{example}
\begin{definition} 
\label[definition]
{WeakHomotopyEquivalence}
  A map $f$ is a \emph{weak homotopy equivalence}, denoted 
  $\inlinetikzcd{ X \ar[r, "\sim", "{f}"{swap}] \& Y }$, of just $X \sim Y$, if its induced homomorphisms \cref{FunctorialityOfHomotopyGroups} of homotopy groups are all isomorphisms, for all $n \in \mathbb{N}$ and all choices of basepoints.
\end{definition}
\begin{definition}
\label[definition]{SerreFibration}
  A map $\inlinetikzcd{X \ar[r, "{ p }"] \& Y }$ is a \emph{Serre fibration}, denoted $p \in \mathrm{Fib}$, if all paths of $n$-paths in $Y$ may be lifted through $p$ for every lift over the starting point:
  \begin{equation}
    \label{LiftingConditionForSerreFib}
    \begin{tikzcd}[column sep=large]
      \{0\} \times I^n
      \ar[
        d,
        hook
      ]
      \ar[
        r, 
        "{ \forall }"{description}
      ]
      &
      X
      \ar[
        d,
        "{
          p
        }"
      ]
      \\
      I \times I^n
      \ar[
        r,
        "{ \forall }"{description}
      ]
      \ar[
        ur,
        dashed,
        "{ \exists }"{description}
      ]
      &
      Y
      \mathrlap{\,.}
    \end{tikzcd}
  \end{equation}
\end{definition}
\begin{example}
  For pointed CW-complexes $X,Y \in \Top$, the map $\inlinetikzcd{\UnpointedMap\bracket({X,Y}) \ar[r, "{ \mathrm{ev}_\ast }"] \& Y}$, given by evaluation at the basepoint, is a Serre fibration (\cref{SerreFibration}) with fiber the pointed mapping space \cref{ThePointedMappingSpace}:
  \begin{equation}
    \label{TheEvaluationFibration}
    \begin{tikzcd}[row sep=10pt, 
      column sep=25pt
    ]
      \Map\bracket({
        X,Y
      })
      \ar[d]
      \ar[r]
      &
      \UnpointedMap\bracket({X,Y})
      \ar[
        d,
        "{
          \mathrm{ev}_\ast
        }"{swap},
        "{ \in \mathrm{Fib} }"
      ]
      \\
      \ast
      \ar[r]
      &
      Y
    \end{tikzcd}
  \end{equation}
\end{example}
\begin{proposition}[{cf. \cite[Prop. 5.9]{Hirschhorn2019}}]
\label[proposition]
{FactorizationWFib}
  Every map factors (nonuniquely) through a weak homotopy equivalence \textup{(\cref{WeakHomotopyEquivalence})} followed by a Serre fibration \textup{(\cref{SerreFibration})}:
  \begin{equation}
    \begin{tikzcd}[
      row sep=-1pt, column sep=large 
    ]
      X 
      \ar[rr, "{ f }"]
      \ar[
        dr,
        "{ \sim }"{sloped, swap}
      ]
      &&
      Y
      \mathrlap{\,.}
      \\
      &
      \widehat X
      \ar[
        ur,
        "{ 
          \in \mathrm{Fib} 
        }"{swap, pos=0.3}
      ]
    \end{tikzcd}
  \end{equation}
\end{proposition}
\begin{definition}[Homotopy pullback, homotopy fiber]
  \begin{enumerate}
  \item
  A Cartesian square \cref{CartesianSquare} is \emph{homotopy Cartesian} (a \emph{homotopy pullback}) if at least one of the maps to the base is a Serre fibration (\cref{SerreFibration}).

  More generally, a commuting square is \emph{homotopy Cartesian} (a \emph{homotopy pullback}), denoted by ``$\lrcorner_h$'', if, after factoring either of the maps to the base  through a Serre fibration according to \cref{FactorizationWFib}, the comparison morphism into the pullback of that fibration is a weak homotopy equivalence (\cref{WeakHomotopyEquivalence}):
  \begin{equation}
  \label{HomotopyCartesianSquare}
    \begin{tikzcd}
      Z
      \ar[r]
      \ar[
        d,
        "{ g^\ast f }"{swap}
      ]
      \ar[
        dr,
        phantom,
        "{ \lrcorner_h }"{pos=.1}
      ]
      & 
      X
      \ar[d, "{ f }"]
      \\
      Y
      \ar[r, "{ g }"]
      & 
      B
    \end{tikzcd}
    \;\;\;\;
    \Leftrightarrow
    \;\;\;\;
    \begin{tikzcd}
      Z
      \ar[r]
      \ar[
        dd,
        leftvertright,
        "{
          g^\ast f
        }"{swap}
      ]
      \ar[
        d,
        shorten >=-3pt,
        "{ \sim }"{sloped, swap}
      ]
      \ar[
        dr,
        phantom,
        "{ \lrcorner_h }"{pos=.1}
      ]
      & 
      X
      \ar[d, "{ \sim }"{sloped}]
      \ar[
        dd,
        rightvertleft,
        "{ f }"
      ]
      \\[-5pt]
      \;Y \!\!\times_{\!_{B}}\!\! \widehat Y
      \ar[r]
      \ar[d]
      \ar[
        dr,
        phantom,
        "{ \lrcorner }"{pos=.1}
      ]
      & 
      \widehat Y
      \ar[
        d,
        "{ 
          \in \mathrm{Fib} 
        }"{description, pos=.45}
      ]
      \\[+3pt]
      Y
      \ar[r, "{ g }"]
      & 
      B
    \end{tikzcd}
  \end{equation}
  
  \item
  A fiber sequence $\inlinetikzcd{F \ar[r] \& X \ar[r, "{f}"] \& Y}$ \cref{Fiber} is a \emph{homotopy fiber sequence} if $f$ is a Serre fibration.

  More generally, any such pair of consecutive morphisms is a \emph{homotopy fiber sequence} if its completion, through the basepoint, to a commuting square is homotopy Cartesian \cref{HomotopyCartesianSquare}:
  \begin{equation}
    \label{HomotopyFiberSequence}
    \begin{tikzcd}[
      row sep=small, 
    ]
      F 
      \ar[r]
      \ar[d]
      \ar[
        dr,
        phantom,
        "{ \lrcorner_h }"{pos=.1}
      ]
      &
      X
      \ar[
        d,
        "{ f }"
      ]
      \\
      \ast \ar[r]
      &
      Y
      \mathrlap{\,.}
    \end{tikzcd}
  \end{equation}
  \end{enumerate}
\end{definition}

\begin{definition}[Connecting homomorphism in Homotopy LES]
\label[definition]{ConnectingHomomorphism}
Given a homotopy fiber sequence $\inlinetikzcd{F \ar[r, "{ \iota }"] \& X \ar[r, "{f}"] \& Y}$ (\cref{HomotopyFiberSequence}), the $n$th \emph{connecting homomorphism} is the homotopy group homomorphism
\begin{equation}
  \label{AConnectingHomomorphism}
  \begin{tikzcd}[row sep=-2pt]
    &&
    \pi_{n+1}(Y)
    \ar[
      dll, 
      snake left,
      "{ \partial^f_n }"{description}
    ]
    \\
    \pi_n(X)
  \end{tikzcd}
\end{equation}
defined as follows (cf. \cref{ConnectingHomomorphismSchematics}):
Given $[\gamma] \in \pi_{n+1}(Y)$, we may choose 
\begin{enumerate}
\item
a representative 
$\inlinetikzcd{\gamma: I^{n+1} \ar[r] \& Y}$,
\item
a map $\inlinetikzcd{\widehat{\gamma}: I^{n+1} \ar[r] \& Y}$,

\item such that 
{\bf (a)}
  $\widehat \gamma_{
    \vert 
    (
      \partial I^{n+1} 
        \setminus
      \{1\} \times I^n
    )
   }
   = \mathrm{const}_{x_0}$
and {\bf (b)}
  $f \circ \widehat{\gamma} = \gamma$.
\end{enumerate}
Then 
\begin{equation}
  \partial^f_n\bracket({
    [\gamma]
  })
  :=
  \bracket[{\,
    \widehat{\gamma}_{\vert \{1\} \times I^n}
  }]
  \in
  \pi_n(F)
\end{equation}
is (well-defined and as such) the value of the $n$th connecting homomorphism on $[\gamma]$.
\end{definition}

\begin{figure}[htb]
\centering
\adjustbox{
  rndfbox=4pt,
  scale=.75
}{
\hspace{5pt}
\begin{tikzpicture}

\node at (0,0) {
  \includegraphics[width=10cm]{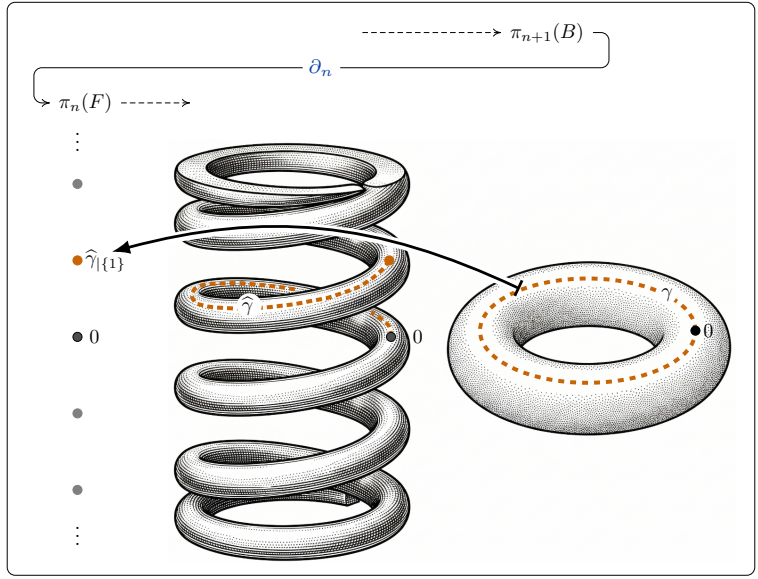}
};

\draw[
  shift={(2.4,.36)},
  darkorange,
  dashed,
  line width=2
]
  (0,0) ellipse (1.9 and .92);
\draw[
  shift={(2.4,.36)},
  fill=black
]
  (0:1.9) circle (.08);

\draw[
  darkorange,
  densely dashed,
  line width=2
]
  plot[
    smooth, 
    tension=1
  ]
  coordinates {
    (-1.1,1.6)
    (-1.45,1.25)
    (-2,1)
    (-2.5,.88)
    (-3,.8)
    (-3.5,.76)
    (-4,.76)
    (-4.4,.81)
    (-4.55,.9)
    (-4.56,1.05)
    (-4.45,1.14)
    (-4.3,1.18)
    (-4.1,1.19)
    (-3.8,1.2)
    (-3.5,1.18)
    (-3.2,1.15)
    (-2.9,1.11)
    (-2.73,1.08)
  };

\draw[
  darkorange,
  densely dashed,
  line width=2
]
  plot[
    smooth, 
    tension=1
  ]
  coordinates {
    (-1.07,.25)
    (-1.09,.4)
    (-1.17,.51)
    (-1.3,.61)
    (-1.42,.68)
  };

\draw[
  fill=black!70
]
  (-1.07,.25) circle (.08);

\draw[
  fill=darkorange,
  draw=darkorange
]
  (-1.1,1.6) circle (.08);

\node at
  (-.6,.25) {0};
\node at
  (4.53,.36) {0};

\node at
  (-6.6,+3.8) {$\vdots$};

\draw[
  fill=black!70
]
  (-6.6,.25) circle (.08);
\node at
  (-6.3,.25) {0};

\draw[
  fill=darkorange,
  draw=darkorange
]
  (-6.6,1.6) circle (.08);

\draw[
  draw=gray,
  fill=gray
]
  (-6.6,1.6+1.6-.25) circle (.08);

\draw[
  draw=gray,
  fill=gray
]
  (-6.6,.25-1.6+.25) circle (.08);

\draw[
  draw=gray,
  fill=gray
]
  (-6.6,.25-1.6+.25-1.6+.25) circle (.08);

\node at
  (-6.6,.25-1.6+.25-1.6+.25-.7) {$\vdots$};

\node at 
 (-7.2,5) {$
   \mathrlap{
     \begin{tikzcd}[
       column sep=110pt,
       ampersand replacement=\&
     ]
       \phantom{A}
       \&[-45]
       \phantom{A}
       \ar[
         r,
         dashed,
         shorten <=40pt
       ]
       \&
       \pi_{n+1}(B)
       \ar[
         dll,
         snake left,
         "{ 
           \mathcolor{darkblue}
           {\partial_n}
          }"{description, scale=1.4}
       ]
       \\
       \pi_n(F)
       \ar[
         r,
         dashed,
         shorten >=30pt
       ]
       \&
       \phantom{A}
     \end{tikzcd}
   }
 $};

 \draw[
   white,
   line width=5
 ]
   (1.2,1.1) to[
     bend right=20
   ]
   (-6,1.75);
 \draw[
   |-Latex,
   line width=1.4
 ]
   (1.2,1.1) to[
     bend right=20
   ]
   (-6,1.75);

\draw[
  draw=white,
  fill=white
]
  (3.8,1) circle (.16);
\node at (3.8,1) 
  {$\gamma$};

\draw[
  draw=white,
  fill=white
]
  (-3.57,.8) circle (.21);
\node at (-3.6,.8) 
  {$\widehat{\gamma}$};

\node at (-6.1,1.6) 
  {$\widehat{\gamma}_{\vert \{1\}}$};

\end{tikzpicture}
}
\caption{
  \label{ConnectingHomomorphismSchematics}
  \textbf{The connecting homomorphisms}
  $\partial_n$ (\cref{ConnectingHomomorphism}) in the homotopy LES (\cref{HomotopyLES})
  of a homotopy fibration sequence $F \hookrightarrow E \overset{p}{\twoheadrightarrow} B$ (\cref{HomotopyFiberSequence})
  takes loops $\gamma$ of $n$-spheres in $B$ \cref{HomotopyGroup} to the endpoint
  $n$-sphere $\widehat{\gamma}_{\vert \{1\}}$ in $F$ of any based path $\widehat{\gamma}$ of $n$-spheres in $E$ which covers the loop, $p \circ \widehat{\gamma} = \gamma.$
}
\end{figure}

\begin{proposition}[Homotopy long exact sequence (cf. {\cite[\S 9.8]{Fomenko2016}})]
\label[proposition]{HomotopyLES}
Given a homotopy fiber sequence $\inlinetikzcd{F \ar[r, "{ \iota }"] \& X \ar[r, "{f}"] \& Y}$, 

\item[\bf (i)] 
the long sequence of homomorphisms of homotopy groups obtained with \textup{\cref{ConnectingHomomorphism}} 
\begin{equation}
  \label{TheHomotopyLES}
  \begin{tikzcd}[row sep=small,
    column sep=50pt
  ]
    {}
    \ar[r, - , dotted]
    &[-30pt]
    \pi_{n+1}(F)
    \ar[r, "{ \pi_{n+1}(\iota) }"]
    &
    \pi_{n+1}(X)
    \ar[r, "{ \pi_{n+1}(f) }"]
    &
    \pi_{n+1}(Y)
    \ar[
      dll,
      snake left,
      "{ 
        \partial^f_n 
      }"{description}
    ]
    &[-30pt]
    \\
    &
    \pi_{n}(F)
    \ar[r, "{ \pi_{n}(\iota) }"]
    &
    \pi_{n}(X)
    \ar[r, "{ \pi_{n}(f) }"]
    &
    \pi_{n}(Y)
    \ar[r, -,  dotted]
    &
    {}
  \end{tikzcd}
\end{equation}
is exact, meaning that at each vertex the image of the incident map coincides with the kernel of the outgoing map. This continues to hold for the maps of pointed sets in degree 0, 
\begin{equation}
  \label{TheHomotopyLESInLowDegree}
  \begin{tikzcd}[row sep=small, column sep=35pt]
  &
  {}
  \ar[
    r,
    -,
    dotted
  ]
  & \pi_1(Y)
  \ar[
    dll,
    snake left,
    "{ \partial^f_0 }"{description}
  ]
  \\
  \pi_0(F)
  \ar[r, "{ \pi_0(\iota) }"]
  &
  \pi_0(X)
  \ar[r, "{ \pi_0(f)}"]
  &
  \pi_0(Y)
  \mathrlap{\,,}
  \end{tikzcd}
\end{equation}
where kernels are understood as preimages of the base point.

\item[\bf (ii)]  Moreover, this construction is natural: Given a commuting square of the form,
\[
  \begin{tikzcd}[column sep=huge]
    X 
    \ar[
      d,
      "{ \varphi_X }"
    ]
    \ar[
      r,
      "{ f }",
      "{ \in \mathrm{Fib} }"{swap}
    ] 
    & 
    Y
    \ar[
      d,
      "{ \varphi_Y }"
    ]
    \\
    X' 
      \ar[
        r,
        "{ f' }",
        "{ \in \mathrm{Fib} }"{swap}
      ] 
    & 
    Y'
  \end{tikzcd}
\]
then the following diagrams, between the corresponding homotopy LESs \cref{TheHomotopyLES}, commute:
\begin{equation}
\label{NaturalityOfHomotopyLES}
  \begin{tikzcd}[column sep=large]
    \pi_{n+1}(Y)
    \ar[
      d,
      "{
        \pi_{n+1}(\phi_Y)
      }"
    ]
    \ar[
      r,
      "{ \partial^f_n }"
    ]
    &
    \pi_n(F)
    \ar[
      d,
      "{
        \pi_{n}(\phi_F)
      }"
    ]
    \ar[
      r,
      "{ \pi_n(\iota) }"
    ]
    &
    \pi_n(X)
    \ar[
      d,
      "{
        \pi_{n}(\phi_X)
      }"
    ]
    \ar[
      r, 
      "{ \pi_n(f) }"
    ]
    &
    \pi_n(Y)
    \ar[
      d,
      "{
        \pi_{n}(\phi_Y)
      }"
    ]
    \\
    \pi_{n+1}(Y')
    \ar[
      r,
      "{ \partial^{f'}_n }"{swap}
    ]
    &
    \pi_n(F')
    \ar[
      r,
      "{ \pi_n(\iota') }"{swap}
    ]
    &
    \pi_n(X')
    \ar[
      r,
      "{ \pi_n(f') }"{swap}
    ]
    &
    \pi_n(Y')
    \mathrlap{\,.}
  \end{tikzcd}
\end{equation}
\end{proposition}

\begin{example}
  \label[example]{MappingFiberSequenceOfCWInclusion}
  Given $\inlinetikzcd{N \ar[r, hook, "{ \iota }"] \& \BulkDomain}$ an inclusion of cell complexes, then for any space $\ClassifyingA$ we have -- by \cref{Cofiber,MapTakingPushoutInFirstArgumentToPullback} a homotopy fiber sequence
  \begin{equation}
    \begin{tikzcd}
      \Map\bracket({
        \BulkDomain/\BoundaryDomain
        , \ClassifyingA
      })
      \ar[
        r,
        "{ q^\ast }"
      ]
      &
      \Map\bracket({
        \BulkDomain, \ClassifyingA
      })
      \ar[
        r,
        "{ \iota^\ast }"
      ]
      &
      \Map\bracket({
        N, \ClassifyingA
      })
    \end{tikzcd}
  \end{equation}
  and hence the induced homotopy LES\textup{ \cref{HomotopyLES}}.
\end{example}

\begin{lemma}
  \label[lemma]{HomotopyGroupsOfFreeMappingSpace}
  For CW-complexes $X,Y \in \Top$ and for $n \in \mathbb{N}$ we have split short exact sequences:
  \begin{equation}
    \label{HomotopySESForEvaluationFibration}
    \begin{tikzcd}[
      column sep=14pt
    ]
      1
      \ar[r]
      &
      \pi_n
        \,
        \mathrm{Map}^\ast({X,Y})
      \ar[rr]
      &&
      \pi_n\, \UnpointedMap({X,Y})
      \ar[
        rrr, 
        "{ 
          \pi_n \mathrm{ev} 
        }"{swap}
      ]
      &&&
      \pi_n\, Y
      \ar[r]
      \ar[
        lll,
        bend right=20pt,
        "{
          \pi_n \sigma
        }"{sloped}
      ]
      &
      1
    \end{tikzcd}
  \end{equation}
  and hence on homotopy groups:
  \begin{equation}
    \label{ComparingHomotopyGroupsUpPointedAndUnpointedMappingSpace}
    \begin{tikzcd}[
      column sep=5pt
    ]
    \pi_n\bracket({\Map(X,Y)})
    \ar[r]
    \ar[
      rr,
      uphordown,
      "{
        (\mathrm{id},0)
      }"
    ]
    &[10pt]
    \pi_n \bracket({
      \UnpointedMap(X,Y)
    })
    \ar[
      r,
      phantom,
      "{ \simeq }"
    ]
    &
    \pi_n \bracket({
      \Map(X,Y)
    })
    \times
    \pi_n (Y)
    \end{tikzcd}
  \end{equation}
  \textup{
  (for all $n \geq 2$, while for $n = 1$ we may have a semidirect product of groups on the right).}
\end{lemma}
\begin{proof}
  The homotopy LES (\cref{HomotopyLES}) of the evaluation fibration \cref{TheEvaluationFibration} is of the form
  \[
    \begin{tikzcd}[row sep=small]
      &&
      \pi_{n+1}(Y)
      \ar[
        dll,
        snake left,
        "{ \partial_n }"{description}
      ]
      \\
      \pi_n \Map(X,Y)
      \ar[r]
      &
      \pi_n \UnpointedMap(X,Y)
      \ar[r, "{ \pi_n(\mathrm{ev}_\ast) } "]
      &
      \pi_n(Y)
      \ar[
        dll,
        snake left,
        "{ \partial_{n-1} }"{description}
      ]
      \\
      \pi_{n-1} \Map(X,Y)\,.
    \end{tikzcd}
  \]
  But the evaluation map has a pointed section $\sigma$, taking points in $Y$ to the maps on $X$ which are constant on these points. Therefore, all $\pi_n \mathrm{ev}$ are surjective, and hence the connecting homomorphisms (\cref{ConnectingHomomorphism}) in the homotopy long exact sequence (\cref{HomotopyLES}) are all null.
\end{proof}
\begin{example}
  In particular this holds for $X \defneq S^1 \sqcup \{\infty\}$, hence for free loop spaces \cref{FreeAndBasedLoopSpace} pointed at the zero-loop
  \begin{equation}
    \label{HomotopyGroupsOfFreeLoopSpace}
    \pi_n\bracket({
      \mathcal{L}Y
    })
    \simeq
    \pi_n\bracket({
      \Omega Y
    })
    \times
    \pi_n\bracket({
      Y
    })\,.
  \end{equation}
\end{example}

\begin{remark}
  Beware that the splitting of the evaluation map used in \cref{HomotopyGroupsOfFreeMappingSpace} in general only exists into the connected component of the zero-map, hence that the analogous splitting need not hold when the pointed mapping space is equipped with a noncanonical base point.
\end{remark}
One exception to this caveat is when $Y$ admits group structure:

\begin{lemma}
  For $G$ a topological group, its free loop space $\mathcal{L}G := \Map\bracket({S^1, G})$ is homeomorphic to the product space of $G$ with its based loop space (based at the neutral element):
  \begin{equation}
    \label{FreeLoopSpaceHomeomorphicToBasedTimesG}
    \mathcal{L}G
    \simeq
    G \times \Omega G
    \mathrlap{\,.}
  \end{equation}
\end{lemma}
\begin{example}
For $G = S^1$ the circle group, \cref{FreeLoopSpaceHomeomorphicToBasedTimesG} reduces to:
\begin{equation}
  \label{FreeLoopSpaceOfTheCircle}
  \mathcal{L}S^1
  \simeq
  S^1 \times \Omega S^1
  \sim
  S^1 \times \mathbb{Z}\,.
\end{equation}
\end{example}

\begin{proposition}
[Homotopical Mayer-Vietoris Sequence {\cite{DyerRoitberg1980}}]
\label[proposition]
{MayerVietorisSequence}
  Given a homotopy pullback square in $\Top$ \textup{(\cref{HomotopyCartesianSquare})},
  \begin{equation}
    \begin{tikzcd}[column sep=large]
      Z 
        \ar[r, "{ u }"]
        \ar[d, "{ v }"{swap}]
        \ar[
          dr,
          phantom,
          "{ \lrcorner }"{pos=.1}
        ]
      &
      X
      \ar[
        d,
        "{f}"{swap},
        "{ \in \mathrm{Fib} }"
      ]
      \\
      Y 
        \ar[
          r,
          "{ g }"
        ] 
      & 
      B
      \mathrlap{\,,}
    \end{tikzcd}
  \end{equation}
  there is a homotopy fiber sequence \cref{HomotopyFiberSequence} of the form
  \begin{equation}
    \begin{tikzcd}
      \Omega B
      \ar[r]
      &
      Z 
      \ar[
        r,
        "{
          (u,v)
        }"
      ]
      &
      X \times Y
      \mathrlap{\,,}
    \end{tikzcd}
  \end{equation}
  whose corresponding homotopy LES \cref{TheHomotopyLES} is of this form:
  \begin{equation}
    \begin{tikzcd}[
      column sep=50pt
    ]
      {}
      \ar[
        r,
        -, 
        dotted
      ]
      &[-40pt]
      \pi_{n+2}(B)
      \ar[r]
      &[-15pt]
      \pi_{n+1}(Z)
      \ar[
        r,
        "{
          \big(
            \pi_{n+1}(u), \pi_{n+1}(v)
          \big)
        }"{description}
      ]
      &[+40pt]
      \pi_{n+1}(X) 
        \oplus 
      \pi_{n+1}(Y)
      \ar[
        dll,
        snake left,
        "{ 
          \pi_{n+1}(f) 
            - 
          \pi_{n+1}(g)
        }"{description}
      ]
      &[-50pt]
      \\
      &
      \pi_{n+1}(B)
      \ar[r]
      &
      \pi_n(Z)
      \ar[
        r,
        "{
          \big(
            \pi_{n}(u), 
            \pi_{n}(v)
          \big)
        }"{description}
      ]
      &
      \pi_n(X) \oplus \pi_n(Y)
      \ar[
        r,
        -, 
        dotted
      ]
      &[-10pt]
      {}
      \\
      {}
      \ar[
        r, - , dotted
      ]
      &
      \pi_2(B)
      \ar[r]
      &
      \pi_1(Z)
      \ar[
        r,
        "{
          \big(
            \pi_1(u),
            \pi_1(v)
          \big)
        }"{description}
      ]
      &
      \pi_1(X) \times \pi_1(Y)
      \ar[
        dll,
        snake left,
        "{
          \pi_1(f) 
            \cdot
          \pi_1(g)^{-1}
        }"{description}
      ]
      \\
      & 
      \pi_1(B)
      \ar[
        r
      ]
      &
      \pi_0(Z)
      \ar[
        r,
        "{
          \big(
            \pi_0(u), 
            \pi_0(v)
          \big)
        }"{description}
      ]
      &
      \pi_0(X) \times_{\pi_0(B)} \pi_0(Y)
      \ar[
        dll,
        snake left
      ]
      \\
      &
      \ast
      \mathrlap{\,.}
    \end{tikzcd}
  \end{equation}
\end{proposition}
\begin{proof}
  This is essentially discussed in \cite[p. 661]{DyerRoitberg1980}, where the sequence ends with $(\pi_0(u), \pi_0(v)) : \inlinetikzcd{\pi_0(Z) \ar[r] \& \pi_0(X) \times \pi_0(Y)}$. But using the factorization lemma one finds that the image of this last map is $\pi_0(X) \times_{\pi_0(B)} \times \pi_0(Y)$.
\end{proof}

\subsection
{Some Geometric Homotopy}
\label
{OnSomeGeometricHomotopy}

Introduction to the \emph{geometric homotopy theory} (\emph{$\infty$-topos theory} \cite[\S 6]{Lurie2009}) that we need here is in \parencites[\S 1]{FSS23-Char}[\S 4.2]{SS26-Bun}, with expository survey in \cite{Schreiber2025}.
For expository background on the language of $\infty$-categories cf. \parencites[Lurie2009]{AntolinCamarena2016}[\S 4.1]{SS26-Bun}. We use the \emph{cohesive} flavor (\cite{Sc13-dcct}) of geometric homotopy theory which is laid out in \parencites[\S IV]{SS26-Orb}[\S 4.3]{SS26-Bun}, making crucial use of \cite{PavlovEtAl2024} (cf. \cref{PathGroupoidAndSmoothOkaPrinciple}).

The focus here is on the construction of \emph{nonabelian differential cohomology}
\cite[\S 9]{FSS23-Char}, for which exposition in the context of flux-quantization is in \cite[\S 3.3]{SS25-Flux}.

\subsubsection
{Homotopy Types as $\infty$-Groupoids}
\label{OnInfinityGroupoids}

Before we get to the geometric refinement, we first now need to make explicit that in the above discussion it is only the \emph{homotopy type} of the classifying spaces $\ClassifyingB$ that matters, which identifies these spaces not up to homeomorphism but only up to weak homotopy equivalence \cref{WeakHomotopyEquivalence}. 

It is natural to think of this homotopy type of $\ClassifyingB$ as reflected by the the classical \emph{singular simplicial complex} $\mathrm{Sing}(\ClassifyingB)$, which is usefully thought of as the  \emph{fundamental path $\infty$-groupoid}. As such we, revisionistically, denote it by $\mathrm{Shp}(\ClassifyingB)$ \cref{ShapeAsSmoothPathGroupoid}, and we shall understand all classifying spaces under this operation, from now on:
\begin{equation}
\label
{ClassifyingSpacesAsInfinityGroupoids}
  \ClassifyingB
  \in
  \inlinetikzcd{
    \UnpointedTop
    \ar[r, "{ \mathrm{Shp} }"]
    \&[10pt]
    \mathrm{Grpd}_\infty
    \mathrlap{\,.}
  }
\end{equation}
Here $\mathrm{Grpd}_\infty$ denotes the $\infty$-category of $\infty$-grouoids, presented for instance by the simplicially enriched category category of Kan simplicial complexes, with respect to which the path $\infty$-groupoids \cref{ClassifyingSpacesAsInfinityGroupoids} are given by the classical singular simplicial Kan complexes (cf. \cite{Friedman2012}).

As such, there is a \emph{hom $\infty$-groupoid} $\mathrm{Grpd}_\infty(-,-)$ between any pair of objects in this $\infty$-category. Leaving the $\mathrm{Shp}$ in  \cref{ClassifyingSpacesAsInfinityGroupoids} notationally implicit, this is just (the homotopy type of) the mapping space \cref{ThePointedMappingSpace}:
\begin{equation}
  \mathrm{Grpd}_\infty\bracket({
    \mathcal{X},
    \mathcal{Y}
  })
  =
  \mathrm{Map}\bracket({
    \mathcal{X},
    \mathcal{Y}
  })
  \mathrlap{\,.}
\end{equation}

\subsubsection
{Rationalization over the Reals}

To (the homotopy type of) a simply-connected space $\ClassifyingB \in \mathrm{Grpd}_\infty$ \cref{ClassifyingSpacesAsInfinityGroupoids} 
is associated its \emph{rationalization} (cf. \cite{FHT2000,Hess2007}) considered here over the real numbers $\mathbb{R}$ (\cite[\S 11]{BousfieldGugenheim1976}, cf. \cite[Rem. 5.2, Def. 5.7]{FSS23-Char}), to be denoted: 
\footnote{
  Beware that $\mathbb{R}$-rationalization $\RRationalization(-)$ is not a \emph{localization} in the technical sense, in that it is not idempotent, since already $(-) \otimes_{_{\mathbb{Z}}} \mathbb{R}$ \cref{TheRRationalizationUnit} is not an idempotent operation on abelian groups ($\mathbb{R}$ is not a \emph{solid ring}, in contrast to $\mathbb{Q}$). But otherwise it behaves just like $\mathbb{Q}$-rationalization, cf. \cite[\S 11]{BousfieldGugenheim1976}, in fact it is the composite operation of $\mathbb{Q}$-rationalization followed by derived extension of scalars \cite[Prop. 5.8]{FSS23-Char}.
}
\begin{equation}
  \label{RRationalization}
  \RRationalization\ClassifyingB
  \in
  \mathrm{Grpd}_\infty
  \mathrlap{\,.}
\end{equation}
This may be thought of as the universal approximation of $\ClassifyingB$ all whose homotopy groups have an $\mathbb{R}$-module structure. Concretely, there is a natural comparison map
\begin{equation}
\label{TheRRationalizationUnit}
  \begin{tikzcd}
    \ClassifyingB
    \ar[
      r,
      "{ 
        \eta
          ^{\mathbb{R}}
          _{\ClassifyingB}
      }"
    ]
    &
    \RRationalization
    \ClassifyingB
    \mathrlap{\,,}
  \end{tikzcd}
\end{equation}
which on homotopy groups $\pi_{n \geq 2}$ induces the canonical maps into the tensor product with $\mathbb{R}$: 
\begin{equation}
  \begin{tikzcd}
  \pi_n\bracket({
    \ClassifyingB
  })
  \ar[
    r,
    "{
      \pi_n(\eta^{\mathbb{R}})
    }",
    "{
      = 
      (-) \otimes 
      (\mathbb{Z} 
        \hookrightarrow
      \mathbb{R})
    }"{swap}
  ]
  &[40pt]
  \pi_n\bracket({
    \ClassifyingB
  })
  \otimes_{{}_{\mathbb{Z}}}
  \mathbb{R}
  \simeq
  \pi_n\bracket({
    \RRationalization
    \ClassifyingB
  })
  \mathrlap{\,.}
  \end{tikzcd}
\end{equation}

While rationalization is traditionally and commonly considered over $\mathbb{Q}$,  it is well-known that it exists over any field of characteristic zero \cite[\S 11]{BousfieldGugenheim1976}. The version over the real numbers stands out as pivotal for relating homotopy theory to \emph{differential geometry} and thereby making it an ingredient of the construction of differential cohomology theories --- this is the content of \cref{ShapeOfSmoothSetOfClosedDiffForms} below.

To that end, we note that the homotopical data retained by $\mathbb{R}$-rationalization is entirely encoded in the data of $L_\infty$-algebras: 

Recall (\cite{LadaStasheff1992}, cf. \parencites[\S 6.1]{SatiSchreiberStasheff2009}[\S 4]{SS26-Bun}[\S 3]{FSS19-RationalM}) than an \emph{$L_\infty$-algebra} is an $\mathbb{N}$-graded vector space $\mathfrak{g}$ (for ordinary Lie algebras this is concentrated in degree 0) equipped with a system of graded skew-symmetric brackets $[-,\cdots,-]$ of any arity, subject to a generalization of the Jacobi identity (which is recovered when only the binary bracket $[-,-]$ is non-trivial). Concretely, when $\mathfrak{g}$ is of \emph{finite type} (ft, meaning degreewise finite-dimensional) then the $L_\infty$-structure is dually encoded in its \emph{Chevalley-Eilenberg algebra} $\mathrm{CE}(\mathfrak{g})$, which is the differential graded-commutative algebra $\bracket({ \wedge^\bullet \mathfrak{g}^\vee, \mathrm{d}})$ over the graded Grassmann algebra of the degreewise dual vector space, with differential $\mathrm{d}$ given by the linear dual of all the bracket operations, extended by the graded Leibniz rule: 
\begin{equation}
\label
{LInfinityAlgsViaCEdgcAlgs}
  \begin{tikzcd}[
    sep=0pt
  ]
    L_\infty\mathrm{Alg}
      ^{\mathrm{ft}}
      _{\mathbb{R}}
    \ar[
      r,
      hook,
      "{
        \mathrm{CE}
      }"
    ]
    &[30pt]
    \mathrm{dgcAlg}
      ^{\mathrm{op}}
      _{\mathbb{R}}
    \\
    \bracket({
      \mathfrak{g},
      [-],
      [-,-],
      [-,-,-], 
      \cdots
    })
    \ar[
      r,
      |->,
      shorten=5pt
    ]
    &
    \bracket({
      \wedge^\bullet \mathfrak{g}^\vee,
      \mathrm{d}\vert
        _{\wedge^1\mathfrak{g}^\vee }
      =
      [-]^\ast
      + 
      [-,-]^\ast
      +
      [-,-,-]^\ast
      +
      \cdots
    })
    \mathrlap{\,.}
  \end{tikzcd}
\end{equation}
We consider all this over the ground field of \emph{real} numbers.

For example, for $\ClassifyingB$ a simply connected space with finite-dimensional rational cohomology, there is an essentially unique $L_\infty$-algebra
\begin{equation}
  \label{WhiteheadShLieAlgebraInIntro}
  \mathfrak{l}\ClassifyingB
  \in
  L_\infty\mathrm{Alg}_{\mathbb{R}}
\end{equation}
characterized by these two properties:
\footnote{
  On the left of \cref{CochainCohomologyOfCEAlgebra} we have the cochain cohomology of the CE-algebra with respect to its differential \cref{LInfinityAlgsViaCEdgcAlgs}, on the right we have ordinary real cohomology of a space.
}
\begin{subequations}
\begin{align}
  (\mathfrak{l}\ClassifyingB)_\bullet
  &
  \simeq
  \pi_\bullet\bracket({
    \Omega \ClassifyingB
  })
  \otimes_{_{\mathbb{Z}}}
  \mathbb{R}
  \mathrlap{\,,}
  \\
  \label{CochainCohomologyOfCEAlgebra}
  H^\bullet\bracket({
    \mathrm{CE}\bracket({
      \mathfrak{l}\ClassifyingB
    })
  })
  & 
  \simeq
  H^\bullet\bracket({
    \ClassifyingB;
    \mathbb{R}
  })
  \;
  \text{
    induced by a quasi-iso
  }
  \inlinetikzcd{
    \mathrm{CE}\bracket({
      \mathfrak{l}\ClassifyingB
    })
    \ar[r, "{ \sim }"]
    \&
    \Omega^\bullet_{\mathrm{dR}}\bracket({
      \mathrm{Sing}\ClassifyingB
    })
  }
  \mathrlap{\,.}
\end{align}
\end{subequations}
This essentially unique dgc-algebra $\mathrm{CE}\bracket({\mathfrak{l}\ClassifyingB})$ is famous as the \emph{minimal Sullivan model} of $\ClassifyingB$
(\parencites[Def. 7.2]{BousfieldGugenheim1976}, cf. \parencites[Def. 1.10]{Hess2007}{Menichi2015}[Def. 4.22]{FSS23-Char}), and $\mathfrak{l}\ClassifyingB$ itself is the
\emph{real Whitehead bracket $L_\infty$-algebra} of this space
(\cite[Prop. 3.1]{BelchiEtAl2017}, cf. \cite[Prop. 5.11, Rem. 5.4]{FSS23-Char}).

The key statement of the \emph{fundamental theorem of dg-algebraic rational homotopy theory} (whose traditional formulation culminated in \cite{BousfieldGugenheim1976}, cf. review in \cite[Prop. 5.6]{FSS23-Char}) is that from this Whitehead bracket $L_\infty$-algebra $\mathfrak{l}\ClassifyingB$ of the (simply connected and rational finite type) space $\ClassifyingB$, the latter's $\mathbb{R}$-rationalization \cref{TheRRationalizationUnit} may be recovered: Up to equivalence, it is the simplicial set of closed $\mathfrak{l}\ClassifyingB$-valued smooth differential forms on extended simplices $\mathbf{\Delta}^{(-)}$:
\footnote{
  The equivalence \cref{RRationalizationFromSullianModel} is traditionally considered not for smooth but for polynomial differential forms and not on extended simplices but on the the ordinary simplices $\Delta^n \subset \mathbf{\Delta}^n$ (for which all $x_i \geq 0$). However, the inclusion of the latter simplicial dg-algebra into the former is clearly degreewise a quasi-isomorphism and hence induces a weak homotopy equivalence under mapping out of the cofibrant dgc-algebra $\mathrm{CE}\bracket({ \mathfrak{l}\ClassifyingB })$ if both are Reedy fibrant. This condition is equivalent to the differential forms satisfying the \emph{extension lemma} (cf. \cite[Lem. 9.4]{GriffithsMorgan2013}), which is the case for both versions (the construction for smooth differential forms in \cite[proof of Cor. 9.9]{GriffithsMorgan2013} applies verbatim also over extended simplicies, cf. \cite{nLab:ExtensionLemmaForDifferentialForms}.
}
\begin{equation}
\label{RRationalizationFromSullianModel}
  \RRationalization \ClassifyingB
  \sim
  \mathrm{Hom}_{\mathrm{dgcAlg}}\bracket({
    \mathrm{CE}\bracket({
      \mathfrak{l}\ClassifyingB
    }),
    \Omega^\bullet_{\mathrm{dR}}\bracket({
      \mathbf{\Delta}^{(-)}
    })
  })
  \in
  \mathrm{sSet}
  \sim
  \mathrm{Grpd}_\infty
  \mathrlap{\,.}
\end{equation}
Here the ``extended $n$-simplex'' (extending its closed subspace, where all $x_i \geq 0$, to a smooth manifold without boundary) is
\begin{equation}
  \label{TheExtendedNSimplex}
  \mathbf{\Delta}^n
  :=
  \bracketmid\{{
    \bracket({
      x_0, \cdots, x_n
    })
    \in
    \mathbb{R}^{n+1}
  }{
    \sum_{k=0}^n x_k = 1
  }\}
  \simeq
  \mathbb{R}^n 
  \subset 
  \mathbb{R}^{n+1}
\end{equation}
with the usual coface and codegeneracy maps
\begin{subequations}
\begin{align}
  \begin{tikzcd}[
    ampersand replacement=\&,
    sep=0pt
  ]
    \mathbf{\Delta}^{n-1}
    \ar[
      rr,
      "{
        \delta_k
      }"
    ]
    \&\&
    \mathbf{\Delta}^{n+1}
    \\
    \bracket({
      x_0, \cdots, x_{n-1}
    })
    \&\mapsto\&
    \bracket({
      x_0, \cdots, x_{k-1}, 0 , x_k, \cdots x_{n-1}
    })
  \end{tikzcd}
  \\
  \begin{tikzcd}[
    ampersand replacement=\&,
    sep=0pt
  ]
    \mathbf{\Delta}^{n}
    \ar[
      rr,
      "{ \sigma_k }"
    ]
    \&\&
    \mathbf{\Delta}^{n-1}
    \\
    \bracket({
      x_0, \cdots, x_n
    })
    \&\mapsto\&
    \bracket({
      x_0, 
      \cdots, 
      x_k + x_{k+1},
      \cdots,
      x_n
    })
    \mathrlap{\,.}
  \end{tikzcd}
\end{align}
\end{subequations}

\subsubsection
{Smooth Sets of closed $L_\infty$-valued differential forms}
\label{OnSmoothSetsOfLAlgValuedForms}

As we turn attention to higher flux densities satisfying generalized Gau{\ss} laws, we make use of the remarkable fact that these are equivalently closed (flat) differential forms with coefficients in (connected, nilpotent, finite type) $L_\infty$-algebras (cf. \parencites[\S 6]{FSS23-Char}[\S 3.1]{SS25-Flux}).

For $\mathfrak{g} \in L_\infty \mathrm{Alg}^{\mathrm{ft}}_{\mathbb{R}}$ \cref{LInfinityAlgsViaCEdgcAlgs} and for $\BulkDomain$ a smooth manifold, its \emph{closed $\mathfrak{g}$-valued differential forms} constitute the set
\begin{equation}
\label
{SetOfClosedLInfinityValuedForms}
  \Omega^1_{\mathrm{cl}}\bracket({
    \BulkDomain;
    \mathfrak{g}
  })
  :=
  \mathrm{Hom}_{\mathrm{dgAlg}}
  \bracket({
    \mathrm{CE}(\mathfrak{g})
    ,
    \Omega^\bullet_{\mathrm{dR}}
    \bracket({\BulkDomain})
  })
\end{equation}
of dg-algebra homomorphism from the CE-algebra \cref{LInfinityAlgsViaCEdgcAlgs} of $\mathfrak{g}$ into the de Rham algebra of smooth differential forms on $\BulkDomain$ (also known as the \emph{Maurer-Cartan elements} in $\Omega^\bullet_{\mathrm{dR}}(\BulkDomain) \otimes \mathfrak{g}$).

We need to promote this set to a kind of \emph{smooth space} $\mathbf{\Omega}^1_{\mathrm{cl}}\bracket({\BulkDomain, \mathfrak{g}})$ (the boldface indicates the additional smooth structure) that makes it behave like a \emph{moduli space} of such differential forms, in that smooth maps into it, from any smooth manifold $X$, are in natural bijection with such differential forms on the product manifold $X \times \Sigma$:
\begin{equation}
  \label{MapsIntoSmoothSetOfClosedForms}
  \mathbf{Map}\bracket({
    X,
    \mathbf{\Omega}^1_{\mathrm{cl}}
    \bracket({
      \BulkDomain,
      \mathfrak{g}
    })
  })
  \simeq
    \mathbf{\Omega}^1_{\mathrm{cl}}
    \bracket({
      X \times \BulkDomain,
      \mathfrak{g}
    })  
  \mathrlap{\,.}
\end{equation}

The tautological solution to this problem is to declare \emph{smooth sets} (\parencites[Def. 1.2.16, 1.3.58]{Sc13-dcct}[Def. 2.1]{KhavkineSchreiber2026}{GS25-FieldsI}{Giotopoulos2025}{IbortMas2025}) to be sheaves on the site $\mathrm{SmthMfd}$ of smooth manifolds
\begin{equation}
\label{CategoryOfSmoothSets}
  \mathrm{SmthSet}
  :=
  \mathrm{Sh}\bracket({
    \mathrm{SmthMfd},
    \mathrm{Set}
  })
  \mathrlap{\,,}
\end{equation}
observe that ordinary smooth manifolds are faithfully included among smooth sets (the \emph{Yoneda embedding}) 
\begin{equation}
  \begin{tikzcd}[
    row sep=0pt
  ]
    \mathrm{SmthMfd}
    \ar[r, hook]
    &
    \mathrm{SmthSet}
    \\
    X 
    \ar[r, |->, shorten=7pt]
    &
    \bracket({
      U 
        \mapsto 
      C^\infty\bracket({
        U, X
      })
    })
    \mathrlap{\,,}
  \end{tikzcd}
\end{equation}
and then define the smooth set of closed $\mathfrak{g}$-valued differential forms on $\BulkDomain$ to be:
\begin{equation}
\label
{SmoothSetOfClosedLInfinityValuedDiffForms}
  \begin{tikzcd}[row sep=-3pt, 
    column sep=0pt
  ]
    \mathbf{\Omega}^1_{\mathrm{cl}}
    \bracket({
      \Sigma;
      \mathfrak{g}
    })
    :
    &
    \mathrm{SmthMfd}^{\mathrm{op}}
    \ar[r]
    &[30pt]
    \mathrm{Set}
    \\
    &
    U
    \ar[
      r,
      |->,
      shorten=5pt
    ]
      &
    \Omega^1_{\mathrm{cl}}\bracket({
      U \!\times\! \Sigma;
      \mathfrak{g}
    })
    \mathrlap{\,.}
  \end{tikzcd}
\end{equation}

\subsubsection
{Smooth $\infty$-Groupoids}

In order to unify the differential geometry of smooth sets (\cref{OnSmoothSetsOfLAlgValuedForms}) with the homotopy theory of $\infty$-groupoids (\cref{OnInfinityGroupoids}), we now pass to \emph{smooth $\infty$-groupoids} (\parencites{SS26-Orb,SS26-Bun}, exposition in \cite{Schreiber2025}), namely to $\infty$-sheaves over the site of smooth manifolds
\begin{equation}
  \mathrm{SmthGrpd}_\infty
  :=
  \mathrm{Sh}_\infty\bracket({
    \mathrm{SmthMfd},
    \mathrm{Grpd}_\infty
  })
  \mathrlap{\,,}
\end{equation}
in straightforward generalization of the smooth sets (smooth 0-groupoids) from \cref{CategoryOfSmoothSets}.
This means that a smooth $\infty$-groupoid $\mathbf{X} \in \mathrm{SmthGrpd}_\infty$ is characterized by its \emph{$\infty$-groupoid of plots} by any probe smooth manifold $U$:
\begin{equation}
  \label{PlotsOfSmoothInfinityGroupoid}
  \begin{tikzcd}[
    sep=0pt
  ]
    \mathrm{SMthMfd}^{\mathrm{op}}
    \ar[
      rr,
      "{ \mathbf{X} }"
    ]
    &&
    \mathrm{Grpd}_\infty
    \\
    U &\mapsto& \mathbf{X}(U)
    \mathrlap{\,.}
  \end{tikzcd}
\end{equation}

This \emph{geometric homotopy theory} ($\infty$-topos) of \emph{smooth $\infty$-groupoids} faithfully includes and hence combines the differential geometry of smooth sets and the bare homotopy theory of plain (geometrically discrete) $\infty$-groupoids:
\begin{equation}
\label
{DiffGeoAndHomotopyInsideSmoothHomotopy}
  \begin{tikzcd}[row sep=5pt, 
    column sep=20pt
  ]
    \mathrm{SmthSet}
    \ar[
      dr,
      hook'
    ]
    &&
    \mathrm{Grpd}_\infty \,.
    \ar[
      dl,
      hook,
      "{
        \mathrm{Dsc}
      }"{pos=.8}
    ]
    \\
    &
    \mathrm{SmthGrpd}_\infty
  \end{tikzcd}
\end{equation}
The collection of smooth maps between a pair $\mathbf{X}, \mathbf{Y} \in \mathrm{SmthGrpd}_\infty$ forms itself a smooth $\infty$-groupoid, whence we denote it in boldface:
\begin{equation}
\label
{SmoothMappingInfinityGroupoid}
  \mathbf{Map}\bracket({
    \mathbf{X},
    \mathbf{Y}
  })
  \in
  \mathrm{SmthGrpd}
  \mathrlap{\,,}
\end{equation}
as characterized by plots \cref{PlotsOfSmoothInfinityGroupoid} given by
\begin{equation}
  \mathbf{Map}\bracket({
    \mathbf{X},
    \mathbf{Y}
  })
  :\;
  U \longmapsto
  \mathrm{SmothGrpd}_\infty\bracket({
    \mathbf{X} \times U,
    \mathbf{Y}
  })
  \mathrlap{\,.}
\end{equation}

\subsubsection
{The cohesive modalities}
\label{OnTheCohesiveModalities}

\begin{proposition}[{\parencites[Prop. 4.4.8]{Sc13-dcct}[Prop. 4.3.39]{SS26-Bun}[Ex. 9.1.19]{SS26-Orb}}]
The inclusion $\mathrm{Dsc}$ \cref{DiffGeoAndHomotopyInsideSmoothHomotopy} is part of an adjoint quadruple of $\infty$-functors
\begin{equation}
  \label{TheAdjointQuadruple}
  \begin{tikzcd}
    \mathrm{SmthGrpd}_\infty
    \ar[
      rr,
      shift left=28pt,
      "{
        \mathrm{Shp}
      }"{description}
    ]
    \ar[
      rr,
      phantom,
      shift left=21pt,
      "{ \bot }"{scale=.7}
    ]
    \ar[
      rr,
      phantom,
      shift left=7pt,
      "{ \bot }"{scale=.7}
    ]
    \ar[
      rr,
      "{
        \mathrm{Pnt}
      }"{description}
    ]
    \ar[
      rr,
      phantom,
      shift right=7pt,
      "{ \bot }"{scale=.7}
    ]
    &&
    \mathrm{Grpd}_\infty
    \mathrlap{\,.}
    \ar[
      ll,
      hook',
      shift right=14pt,
      "{
        \mathrm{Dsc}
      }"{description}
    ]
    \ar[
      ll,
      hook',
      shift left=14pt,
      "{
        \mathrm{Cht}
      }"{description}
    ]
  \end{tikzcd}
\end{equation}
\end{proposition}

These being adjoint, as shown, means that we have natural equivalences
\begin{equation}
  \label{HomIsoForShp}
  \mathrm{Grpd}_\infty\bracket({
    \mathrm{Shp}\mathbf{X},
    \mathcal{Y}
  })
  \simeq
  \mathrm{SmthGrpd}_\infty\bracket({
    \mathbf{X},
    \mathrm{Dsc}\mathcal{Y}
  })
  \mathrlap{\,,}
\end{equation}
etc.

The composite operation
\begin{equation}
  \label{FlatOperation}
  \begin{tikzcd}
    \flat
      : 
    \mathrm{SmthGrpd}_\infty
    \ar[
      r, 
      "{ \mathrm{\mathrm{Pnt}} }"
    ]
    &
    \mathrm{Grpd}_\infty
    \ar[
      r, 
      hook,
      "{ \mathrm{Dsc} }"
    ]
    &
    \mathrm{SmthGrpd}_\infty
    \mathrlap{\,}
  \end{tikzcd}
\end{equation}
is called the \emph{points modality} (or \emph{flat modality}, cf. \cite[\S 4.4.12]{Sc13-dcct}).
It may be thought of as sending  any smooth $\infty$-groupoid to its underlying geometrically discrete $\infty$-groupoid, obtained by forgetting its smooth structure, in that
\begin{equation}
\label
{PlotsOfFlatOfX}
  \mathrm{Dsc}\bracket({
    \mathcal{X}
  })
  :
  U \mapsto \mathcal{X}
  \,,
  \;\;\;\;
  \mathrm{Pnt}\bracket({
    \mathbf{X}
  })
  = 
  \mathbf{X}(\ast)
  \,,
  \;\;\;\;
  \flat \mathbf{X}
  :
  U \mapsto \mathbf{X}(\ast)
  \mathrlap{\,.}
\end{equation}
Accordingly, the induced natural comparison maps
\begin{equation}
  \label{FlatCoUnitInIntro}
  \begin{tikzcd}
    \flat \mathbf{X}
    \ar[
      r, 
      "{ \epsilon^{\flat}_{\mathbf{X}} }"]
    &
    \mathbf{X}
    \mathrlap{\,,}
  \end{tikzcd}
\end{equation}
may be thought of as generalizing the inclusion of the underlying set into a smooth manifold. Since $\mathrm{Dsc}$ is fully faithful, this is idempotent, in that
\begin{equation}
\label
{IdempotencyOfFlatInIntro}
\begin{tikzcd}
    \flat \flat \mathbf{X}
    \ar[
      r, 
      "{ 
        \epsilon
          ^{\flat}
          _{\flat\mathbf{X}} 
      }",
      "{
        \sim
      }"{swap}
    ]
    &
    \flat\mathbf{X}
    \mathrlap{\,,}
  \end{tikzcd}
\end{equation}
is an equivalence.
\begin{example}
  The points of the internal mapping smooth $\infty$-groupoid \cref{SmoothMappingInfinityGroupoid} form, by \cref{PlotsOfFlatOfX}, the bare $\infty$-groupoid of smooth maps:
  \begin{equation}
    \label{PointsOfSmoothMap}
    \flat \mathbf{Map}\bracket({
      \mathbf{X},
      \mathbf{Y}
    })
    \simeq
    \mathrm{Dsc}\bracket({
    \mathrm{SmthGrpd}_\infty\bracket({
      \mathbf{X},
      \mathbf{Y}
    })
    })
    \mathrlap{\,.}
  \end{equation}
  Therefore, on geometrically discrete $\infty$-groupoids $\mathcal{X}, \mathcal{Y} \in \mathrm{Grpd}_\infty$ we have moreover:
  \begin{equation}
  \begin{alignedat}{2}
      \mathbf{Map}\bracket({
        \mathrm{Dsc}\bracket({
          \mathcal{X}
        }),
        \mathrm{Dsc}\bracket({
          \mathcal{Y}
        })
      })
      & 
      \sim
      \flat
      \mathbf{Map}\bracket({
        \mathrm{Dsc}\bracket({
          \mathcal{X}
        }),
        \mathrm{Dsc}\bracket({
          \mathcal{Y}
        })
      })
      &\;\;&
      \substack{\text{by \cref{PlotsOfFlatOfX}}}
      \\
      & 
      \sim
      \mathrm{Dsc}\,
      \mathrm{SmthGrpd}_\infty\bracket({
        \mathrm{Dsc}\bracket({
          \mathcal{X}
        }),
        \mathrm{Dsc}\bracket({
          \mathcal{Y}
        })
      })
      &\;\;&
      \substack{\text{ by \cref{PointsOfSmoothMap}}}
      \\
      & \sim
      \mathrm{Dsc}\,
      \mathrm{Grpd}_\infty\bracket({
        \mathcal{X},
        \mathcal{Y}
      })
      &\;\;&
      \substack{\text{ by full faithfulness  \cref{DiffGeoAndHomotopyInsideSmoothHomotopy}}}
      \\
      & \sim
      \mathrm{Dsc}\,
      \mathrm{Map}\bracket({
        \mathcal{X},
        \mathcal{Y}
      })
      \mathrlap{\,.}
  \end{alignedat}
  \end{equation}
  But since $\mathrm{Dsc}$ is fully faithful, we may suppress it notationally. With that understood, the above says that between geometrically discrete $\infty$-groupoids we have simply:
  \begin{equation}
    \mathcal{X}, \mathcal{Y}
    \in
    \mathrm{Grpd}_\infty
    \;\;\;
    \vdash
    \;\;\;
    \mathbf{Map}\bracket({
      \mathcal{X},
      \mathcal{Y}
    })
    \sim
    \mathrm{Map}\bracket({
      \mathcal{X},
      \mathcal{Y}
    })
    \mathrlap{\,.}
  \end{equation}
\end{example}

Next, left adjoint to $\flat$, the composite \cref{TheAdjointQuadruple} operation
\begin{equation}
  \label{ShapeOperation}
  \begin{tikzcd}
    \shape 
      : 
    \mathrm{SmthGrpd}_\infty
    \ar[
      r, 
      "{ \mathrm{Shp} }"
    ]
    &
    \mathrm{Grpd}_\infty
    \ar[
      r, 
      hook,
      "{ \mathrm{Dsc} }"
    ]
    &
    \mathrm{SmthGrpd}_\infty
    \mathrlap{\,}
  \end{tikzcd}
\end{equation}
is called the \emph{shape modality} (cf. \cite[ftn. 1]{SS26-Bun}). It may be understood \cref{ShapeAsSmoothPathGroupoid} as sending any object to its \emph{path $\infty$-groupoid}. By its construction from adjoints, this comes with natural comparison maps 
\begin{equation}
  \label{ShapeUnitInIntro}
  \begin{tikzcd}
    \mathbf{X}
    \ar[r, "{ \eta^{\shape}_{\mathbf{X}} }"]
    &
    \shape \mathbf{X}
    \mathrlap{\,,}
  \end{tikzcd}
\end{equation}
which may be thought of as including the constant paths amoung all paths. Since $\mathrm{Dsc}$ is fully faithful, this map exhibits \emph{idempotency} of the shape operation, in that applied to a pure shape it is an equivalence:
\begin{equation}
\label
{IdempotencyOfShapeInIntro}
    \begin{tikzcd}
      \shape\mathbf{X}
      \ar[
        r, 
        "{ 
          \eta
            ^{\shape}
            _{\shape\mathbf{X}}
        }",
        "{ \sim }"{swap}
      ]
      &
      \shape\shape\mathbf{X}
      \mathrlap{\,.}
    \end{tikzcd}
\end{equation}

Less immediate are the following important properties of the shape modality:
\begin{proposition}[{\parencites[Prop. 4.3.8]{SS26-Bun}[Thm. 3.8.19]{Sc13-dcct}}]
\label[proposition]
{ShapePreservesFiberProductsOverDiscete}
The shape operation \cref{ShapeOperation} preserves fiber products over geometrically discrete objects:
\begin{equation}
\label
{ShapePreservingFiberProductsOverDiscrete}
  \left.
  \begin{subarray}{l}
    \mathbf{X}, \mathbf{Y}
    \in
    \mathrm{SmthGrpd}_\infty
    \\
    \mathbf{B} 
    \in 
    \mathrm{Dsc}(\mathrm{Grpd}_\infty)
  \end{subarray}
  \right\}
  \;\;
  \Rightarrow
  \;\;
  \shape\bracket({
    \mathbf{X} \times_{\mathbf{B}}
    \mathbf{Y}
  })
  \simeq
  \bracket({\shape \mathbf{X}})
  \times_{\shape\mathbf{B}}
  \bracket({\shape \mathbf{Y}})
  \simeq
  \bracket({\shape \mathbf{X}})
  \times_{\mathbf{B}}
  \bracket({\shape \mathbf{Y}})
  \mathrlap{\,.}
\end{equation}
\end{proposition}

\begin{proposition}[
  {\cite{PavlovEtAl2024}, cf. \cite[\S 4.3.2]{SS26-Bun}}
]
\label[proposition]
{PathGroupoidAndSmoothOkaPrinciple}
  The shape operation \cref{ShapeOperation}:
  \begin{enumerate}
    \item
    \textbf{\textup{(path $\infty$-groupoid)}}
    is equivalently given by forming smooth singular simplicial complexes, in that
    \begin{equation}
      \label{ShapeAsSmoothPathGroupoid}
      \shape \mathbf{X}
      \simeq
      \mathrm{hocolim}_n
      \mathbf{Map}\bracket({
        \mathbf{\Delta}^n,
        \mathbf{X}
      })
      \simeq
      \mathrm{Dsc}
      \bracket({
      \mathrm{hocolim}_n
      \mathbf{X}\bracket({
        \mathbf{\Delta}^n
      })      
      })
      \mathrlap{\,,}
    \end{equation}
    where $\mathbf{\Delta}^\bullet$ is the cosimplicial smooth manifold of extended simplices \cref{TheExtendedNSimplex},

    \item 
    \textbf{\textup{(smooth Oka principle)}}
    satisfies
    \begin{equation}
    \label{SmoothOkaPrinciple}
      \shape 
      \mathbf{Map}\bracket({
        X, \mathbf{Y}
      })
      \simeq
      \mathbf{Map}\bracket({
        \shape X, 
        \shape \mathbf{Y}
      })
      \simeq
      \mathrm{Dsc}\bracket({
        \mathrm{Map}\bracket({
          X, \mathrm{Shp} \mathbf{Y}
        })
      })\,,
    \end{equation}
    for all $X \in \mathrm{SmthMfd}$ and $\mathbf{Y} \in \mathrm{SmthGrpd}_\infty$. 
  \end{enumerate}
\end{proposition}
\begin{proposition}
We have natural equivalences
\begin{equation}
  \label{MapIntoShapeIsMapFromShape}
  \mathbf{Map}\bracket({
    \mathbf{X},
    \shape
    \mathbf{Y}
  })
  \sim
  \mathbf{Map}\bracket({
    \shape \mathbf{X},
    \shape
    \mathbf{Y}
  })\,.
\end{equation}
\end{proposition}
\begin{proof}
  This follows over plots by manifolds $U$ which are contractible, $\shape U \sim \ast$, by repeated use of the adjunction $\mathrm{Shp} \dashv \mathrm{Dsc}$ \cref{HomIsoForShp} and the full faithfulness of $\mathrm{Dsc}$, and from here in general by the sheaf/descent property, using that every manifold is covered by contractibles. 
\end{proof}

\begin{example}
\label[example]
{ShapeOfSmoothSetOfClosedDiffForms}
For $\mathfrak{g}$ an $L_\infty$-algebra of finite type \cref{LInfinityAlgsViaCEdgcAlgs} and $X$ a smooth manifold, the smooth Oka principle \cref{SmoothOkaPrinciple} shows that the shape \cref{ShapeOperation} of the smooth set of closed $\mathfrak{g}$-valued differential forms \cref{SmoothSetOfClosedLInfinityValuedDiffForms} has the following equivalent incarnations:
\begin{equation}
  \label{ShapeOfClosedFormsOnSigmaIsMapsFromSigmaToShapeOfFormsOverPoint}
  \begin{alignedat}{2}
  \shape
  \mathbf{\Omega}_{\mathrm{cl}}^1\bracket({
    X; 
    \mathfrak{g}
  })
  &
  \simeq
  \shape
  \mathbf{Map}\bracket({
    X,
    \mathbf{\Omega}^1_{\mathrm{cl}}
    ({
      \ast; \mathfrak{g}
    })
  })
  &\;\;&
  \substack{\text{by \cref{MapsIntoSmoothSetOfClosedForms}}}
  \\
  & \simeq
  \mathbf{Map}\bracket({
    \shape X
    ,
    \shape
    \mathbf{\Omega}^1_{\mathrm{cl}}
    ({
      \ast; \mathfrak{g}
    })
  })
  &\;\;&
  \substack{\text{ by \cref{SmoothOkaPrinciple}}}
  \\
  & \simeq
  \mathbf{Map}\bracket({
    X,
    \shape
    \mathbf{\Omega}^1_{\mathrm{cl}}
    ({
      \ast; \mathfrak{g}
    })
  })
  &\;\;&
  \substack{\text{by \cref{MapIntoShapeIsMapFromShape}.}}
  \end{alignedat}
\end{equation}
Moreover, if here $\mathfrak{g} \defneq \mathfrak{l}\ClassifyingB$ is the real Whitehead $L_\infty$-algebra of a space $\ClassifyingB$ (simply connected with finite dimensional rational cohomology) then \cref{ShapeAsSmoothPathGroupoid} implies, with the fundamental theorem of dg-algebraic rational homotopy theory \cref{RRationalizationFromSullianModel}, that
\begin{equation}
  \shape \mathbf{\Omega}^1_{\mathrm{cl}}
  \bracket({
    \ast,
    \mathfrak{l}\ClassifyingB
  })
  \sim
  \RRationalization 
  \ClassifyingB
\end{equation}
is the $\mathbb{R}$-rationalization \cref{TheRRationalizationUnit} of that space. With \cref{ShapeOfClosedFormsOnSigmaIsMapsFromSigmaToShapeOfFormsOverPoint} this implies for any smooth manifold $X$ that
\begin{equation}
\label
{ShapeOfClosedFormsOnSmoothManifold}
  \shape 
  \mathbf{\Omega}^1_{\mathrm{cl}}
  \bracket({
    X;
    \mathfrak{l}\ClassifyingB
  })
  \sim
  \mathrm{Map}\bracket({
    X,
    \RRationalization
    \ClassifyingB
  })
  \mathrlap{\,.}
\end{equation}
\end{example}

More generally:
\begin{lemma}
\label[lemma]
{ShapePreservesFiberProductOfClosedDiffForms}
There is a natural equivalence
\begin{equation}
  \begin{aligned}
    \shape\left(
    \mathbf{\Omega}^1_{\mathrm{cl}}\bracket({
      \BoundaryDomain;
      \mathfrak{l}_{{}_{
        \ClassifyingB
      }}
      \ClassifyingA
    })
    \underset{
        \mathbf{\Omega}^1_{\mathrm{cl}}
        \scaledbracket({
          \BoundaryDomain;
          \mathfrak{l}\ClassifyingB
        })    
    }
    {\times}
    \mathbf{\Omega}^1_{\mathrm{cl}}\bracket({
      \BulkDomain;
      \mathfrak{l}\ClassifyingB
    })
    \right)
    \;\sim\;
    \shape
    \mathbf{\Omega}^1_{\mathrm{cl}}\bracket({
      \BoundaryDomain;
      \mathfrak{l}_{{}_{
        \mathfrak{l}\ClassifyingB
      }}
      \ClassifyingA
    })
    \underset{
        \shape
        \mathbf{\Omega}^1_{\mathrm{cl}}
        \scaledbracket({
          \BoundaryDomain;
          \mathfrak{l}\ClassifyingB
        })    
    }
    {\times}
    \shape
    \mathbf{\Omega}^1_{\mathrm{cl}}\bracket({
      \BulkDomain;
      \mathfrak{l}\ClassifyingB
    })
    \mathrlap{\,.}
  \end{aligned}
\end{equation}
\end{lemma}
\begin{proof}
  The point is that $\inlinetikzcd{
    \mathfrak{l}_{{}_{\ClassifyingB}}
    \ClassifyingA 
      \ar[r, ->>, "{ \mathfrak{l}\wp }"] 
      \& 
    \mathfrak{l}\ClassifyingB
    }$ is a fibration between fibrant objects, by construction, dual to a relative Sullivan model. Hence $\mathrm{CE}\bracket({\mathfrak{l}\wp})$ is a cofibration  and hence also a Reedy cofibration between simplicially constant simplicial dgc-algebra. But this means that the homotopy fiber product on the right is the ordinary fiber product of the simplicial sets representing the the three objects by \cref{ShapeAsSmoothPathGroupoid}, such as $\Omega^1_{\mathrm{cl}}\bracket({N \times \mathbf{\Delta}^{(-)}; \mathfrak{l}\ClassifyingB})$  (where $\mathbf{\Omega}^\bullet_{\mathrm{dR}}\bracket({N \times \mathbf{\Delta}^{(-)}})$ is Reedy fibrant, by the extension lemma). That ordinary fiber product of simplicial sets is computed degreewise, which yields the corresponding simplicial set modelling the left-hand side:
    \[
      [n] 
      \;\longmapsto\;
        {\Omega}^1_{\mathrm{cl}}\bracket({
          \BoundaryDomain 
            \!\times\!
          \mathbf{\Delta}^n;
          \mathfrak{l}_{{}_{
            \ClassifyingB
          }}
          \ClassifyingA
        })
        \underset{
            {\Omega}^1_{\mathrm{cl}}
            \scaledbracket({
              \BoundaryDomain
                \!\times\!
              \mathbf{\Delta}^n;
              \mathfrak{l}\ClassifyingB
            })    
        }
        {\times}
        {\Omega}^1_{\mathrm{cl}}\bracket({
          \BulkDomain
            \times
          \mathbf{\Delta}^n;
          \mathfrak{l}\ClassifyingB
        })
        \mathrlap{\,.}
      \qedhere
    \]
\end{proof}

\subsubsection
{The character map}
\label
{OnTheCharacterMap}

The equivalence \cref{ShapeOfClosedFormsOnSmoothManifold}
is consequential in relating differential form data to homotopy theoretic data, as we  now have canonical maps, \cref{TheRRationalizationUnit,ShapeUnitInIntro}, from both these realms into the same object:
\begin{equation}
  \label{TheDiffHomotopyCospan}
  \begin{tikzcd}[
    row sep=5pt, 
    column sep=40pt
  ]
    &[+15pt]
    &[-79pt]
    \UnpointedMap\bracket({
      \BulkDomain,
      \ClassifyingB
    })
    \ar[
      d,
      "{
        (\eta^{\mathbb{R}})_\ast
      }"
    ]
    \ar[
      ddl,
      start anchor={
        [xshift=10pt]south west
      },
      "{
        \mathbf{ch}
          ^{\ClassifyingB}
          _{\BulkDomain}
      }"{swap}
    ]
    \\[30pt]
    &&
    \;\;\;
    \UnpointedMap\bracket({
      \BulkDomain,
      \RRationalization\ClassifyingB
    })
    \\[2pt]
    \mathbf{\Omega}^1_{\mathrm{cl}}
    \bracket({
      \BulkDomain,
      \mathfrak{l}\ClassifyingB
    })
    \ar[
      r,
      "{
        \eta^{\shape}
         _{\BulkDomain}
      }"
    ]
    &
    \shape
    \mathbf{\Omega}^1_{\mathrm{cl}}
    \bracket({
      \BulkDomain,
      \mathfrak{l}\ClassifyingB
    })
    \mathrlap{\,.}
    \ar[
      ur,
      <->,
      shorten=-3pt,
      "{ \sim }"{sloped, swap}
    ]
  \end{tikzcd}
\end{equation}
Here the composite map on the right is called the differential \emph{character map} \parencites[Def. 9.2]{FSS23-Char}[\S 3.3]{SS25-Flux} since, as we will recall in a moment, it generalizes, from abelian to nonabelian generalized cohomology theories, the \emph{Dold-Chern character map}.

Namely, the universal ``unification'' of the differential-geometric and homotopy-theoretic aspects in \cref{TheDiffHomotopyCospan}, as given by the homotopy fiber product $\mathrm{PhsSp}\bracket({\BulkDomain;\ClassifyingB})$ of these two maps (cf. \cite{SS24-Phase,SS25-Flux}):
\begin{equation}
\label{PlainPhaseSpaceInIntro}
  \begin{tikzcd}
    \mathrm{PhsSp}\bracket({
      \BulkDomain;
      \ClassifyingB
    })
    \ar[r]
    \ar[d]
    \ar[
      dr,
      phantom,
      "{ \lrcorner }"{pos=.1}
    ]
    &
    \UnpointedMap\bracket({
      \BulkDomain,
      \ClassifyingB
    })
    \ar[
      d,
      "{
        \mathbf{ch}^{\ClassifyingB}
        _{\BulkDomain}
      }"
    ]
    \\
    \mathbf{\Omega}^1_{\mathrm{cl}}
    \bracket({
      \BulkDomain,
      \mathfrak{l}\ClassifyingB
    })
    \ar[
      r,
      "{ 
        \eta^{\shape} 
          _{\BulkDomain}
      }"
    ]
    &
    \shape
    \mathbf{\Omega}^1_{\mathrm{cl}}
    \bracket({
      \BulkDomain,
      \mathfrak{l}\ClassifyingB
    })
    \mathrlap{\,,}
  \end{tikzcd}
\end{equation}
is the moduli stack for \emph{differential cohomology} with coefficients in $\ClassifyingB$, in refinement of how $\UnpointedMap\bracket({\BulkDomain, \ClassifyingB})$ is the moduli space for plain cohomology \cref{DefinitionOnNonabelianCohomology} with coefficients in $\ClassifyingB$:
\begin{equation}
  H^0_{\mathrm{dff}}\bracket({
    \BulkDomain;
    \ClassifyingB
  })
  :=
  \pi_0
  \,
  \flat
  \mathrm{PhsSp}\bracket({
    \BulkDomain,
    \ClassifyingB
  })\,.
\end{equation}

Crucially, the shape of the phase space \cref{PlainPhaseSpaceInIntro} is the homotopy type of the plain mapping space:
\begin{equation}
  \shape
  \mathrm{PhasSp}\bracket({
    \BulkDomain,
    \ClassifyingB
  })
  \simeq
  \UnpointedMap\bracket({
    \BulkDomain,
    \ClassifyingB
  })
  \mathrlap{\,.}
\end{equation}
This is a consequence \cref{IdempotencyOfShapeInIntro} of having the same $L_\infty$-coefficients in both objects in the bottom line of \cref{PlainPhaseSpaceInIntro}. Conversely, in a moment we will see that we may ``adjust'' the shape by admitting only more restricted $L_\infty$-algebra coefficients.

\subsubsection
{The twisted relative character map}

It is fairly straightforward to generalize the character map \cref{OnTheCharacterMap} to the twisted relative situation.

First, in generalization of \cref{WhiteheadShLieAlgebraInIntro}, associated with a fibration $\inlinetikzcd{\ClassifyingA \ar[r, "{\wp}"] \& \ClassifyingB}$ (of simply connected spaces with finite-dimensional rational cohomology) is a fibration of  relative Whitehead $L_\infty$-algebras \cite[Def. 6.6, Prop. 5.16]{FSS23-Char}
\begin{equation}
\label
{RelativeWhiteheadLAlgebra}
  \begin{tikzcd}
    \mathfrak{l}_{_{\ClassifyingB}}
    \ClassifyingA
    \ar[
      r,
      ->>,
      "{
        \mathfrak{l}\wp
      }"
    ]
    &
    \mathfrak{l}\ClassifyingB
    \mathrlap{\,,}
  \end{tikzcd}
\end{equation}
this being the formal dual to the \emph{relative minimal Sullivan model} of $\wp$, relative to the minimal Sullivan model of $\ClassifyingB$:
\begin{equation}
  \begin{tikzcd}
    \mathrm{CE}\bracket({
      \mathfrak{l}_{{}_{\ClassifyingB}}
      \ClassifyingA
    })
    &&
    \mathrm{CE}\bracket({
      \mathfrak{l}\ClassifyingB
    })
    \mathrlap{\,.}
    \ar[
      ll,
      hook',
      "{
        \mathrm{CE}(
          \mathfrak{l}\wp
        )
      }"{swap}
    ]
  \end{tikzcd}
\end{equation}

With that and in generalization of \cref{SetOfClosedLInfinityValuedForms,SmoothSetOfClosedLInfinityValuedDiffForms}, given an embedding of smooth manifolds $\inlinetikzcd{ \BoundaryDomain \ar[r, hook, "{ \phi }"] \& \BulkDomain }$ we take the smooth set \cref{CategoryOfSmoothSets} of \emph{twisted relative closed differential forms} on $\phi$ with coefficients in $\mathfrak{l}\wp$ to be:
\begin{equation}
  \mathbf{\Omega}^1_{\mathrm{cl}}
  \bracket({
    \phi;
    \mathfrak{l}\wp
  })
  :=
  \mathbf{\Omega}^1_{\mathrm{cl}}\bracket({
    \BulkDomain;
    \mathfrak{l}\ClassifyingB
  })
  \underset
    {
      \mathbf{\Omega}^1_{\mathrm{cl}}
      ({
        \BoundaryDomain;
        \mathfrak{l}\ClassifyingB
      })    
    }
    {\times}
  \mathbf{\Omega}^1_{\mathrm{cl}}\bracket({
    \BoundaryDomain;
    \mathfrak{l}_{{}_{\ClassifyingB}}
    \ClassifyingA
  })\,.
\end{equation}

While the shape operation \cref{ShapeOperation} does not preserve fiber products in general, it does preserve this one (\cref{ShapePreservesFiberProductOfClosedDiffForms}), since $\mathfrak{l}\wp$ is a fibration. Similarly, the homotopy type of $\UnpointedMap\bracket({\phi,\wp})$ \cref{TwistedRelativeNonabelianCohomology} is well-defined since $\wp$ is a fibration. Therefore the differential character map \cref{TheDiffHomotopyCospan} generalizes immediately, and with it we obtain the \emph{twisted relative phase space}, in evident generalization of \cref{PlainPhaseSpaceInIntro}:
\begin{equation}
\label{TwistedRelativePhaseSpaceInIntro}
  \begin{tikzcd}
    \mathrm{PhsSp}({
      \phi;
      \wp
    })
    \ar[r]
    \ar[d]
    \ar[
      dr,
      phantom,
      "{ \lrcorner }"{pos=.1}
    ]
    &
    \UnpointedMap({
      \phi,
      \wp
    })
    \ar[
      d,
      "{
        \mathbf{ch}^{\wp}
      }"
    ]
    \\
    \mathbf{\Omega}^1_{\mathrm{cl}}
    ({
      \phi,
      \mathfrak{l}\wp
    })
    \ar[
      r,
      "{ \eta^{\shape} }"
    ]
    &
    \shape
    \mathbf{\Omega}^1_{\mathrm{cl}}
    ({
      \phi,
      \mathfrak{l}\wp
    })
    \mathrlap{\,,}
  \end{tikzcd}
\end{equation}

At this point we highlight that there is the  freedom to choose other differential form coefficients $\mathfrak{l}\wp'$, as long as these map to $\mathfrak{l}\wp$. In particular we may consider the case $\mathfrak{l}\wp' := \mathfrak{l}\,\mathrm{id}_{\ClassifyingA}$, for which we have  a natural identification:
\begin{equation}
  \mathbf{\Omega}^1_{\mathrm{cl}}\bracket({
    \phi;
    \mathfrak{l}\,
    \mathrm{id}_{\ClassifyingA}
  })
  \simeq
  \mathbf{\Omega}^1_{\mathrm{cl}}\bracket({
    \BulkDomain;
    \mathfrak{l}
    \ClassifyingA
  })
\end{equation}
and consider the \emph{constrained} phase space $\mathrm{PhsSpc}'(\phi;\wp)$ as the further pullback of \cref{TwistedRelativePhaseSpaceInIntro} along the canonical map:
\begin{equation}
\label{ConstrTwistedRelativePhaseSpaceInIntro}
  \begin{tikzcd}[column sep=large]
    \mathrm{PhsSp}'({
      \phi;
      \wp
    })
    \ar[r]
    \ar[d]
    \ar[
      dr,
      phantom,
      "{ \lrcorner }"{pos=.1}
    ]
    &
    \mathrm{PhsSp}({
      \phi;
      \wp
    })
    \ar[r]
    \ar[d]
    \ar[
      dr,
      phantom,
      "{ \lrcorner }"{pos=.1}
    ]
    &
    \UnpointedMap({
      \phi,
      \wp
    })
    \ar[
      d,
      "{
        \mathbf{ch}^{\wp}
      }"
    ]
    \\
    \mathbf{\Omega}^1_{\mathrm{cl}}\bracket({
      \BulkDomain;
      \mathfrak{l}\ClassifyingA
    })
    \ar[
      r,
      "{
        \mathfrak{l}(
          \mathrm{id}, \wp
        )_\ast
      }"
    ]
    &
    \mathbf{\Omega}^1_{\mathrm{cl}}
    ({
      \phi,
      \mathfrak{l}\wp
    })
    \ar[
      r,
      "{ \eta^{\shape} }"
    ]
    &
    \shape
    \mathbf{\Omega}^1_{\mathrm{cl}}
    ({
      \phi,
      \mathfrak{l}\wp
    })
    \mathrlap{\,.}
  \end{tikzcd}
\end{equation}

This is the key object of interest in \cref{OnIdentifyingThePhaseSpace}.

\printbibliography

\end{document}